\documentclass{article}

\usepackage{amsmath,amsfonts,amssymb,amsthm}
\usepackage{mathrsfs}
\usepackage{hyperref}
\usepackage{enumerate}
\usepackage{graphicx, xcolor}
\usepackage[T1]{fontenc}
\usepackage{footnote,pifont,dsfont}
\usepackage{fullpage}

\usepackage{comment}

\usepackage{algpseudocode}
\usepackage[algoruled,boxed,lined]{algorithm2e}

\usepackage{fancyvrb}
\fvset{fontshape=normal,formatcom=\color{red}}
\fvset{listparameters={\setlength{\topsep}{0pt}\setlength{\partopsep}{0pt}}} 

\usepackage{tikz}

\SetKwInOut{Input}{input}
\SetKwInOut{Output}{output}

\newtheorem{theorem}{Theorem}
\newtheorem{result}{Result}
\newtheorem{lemma}[theorem]{Lemma}
\newtheorem{corollary}[theorem]{Corollary}
\newtheorem{conjecture}[theorem]{Conjecture}
\newtheorem{proposition}[theorem]{Proposition}
\newtheorem{definition}[theorem]{Definition}

\newtheorem{example}[theorem]{Example}

\usepackage{mdframed}
\newtheorem*{algo*}{Algorithm}
\newenvironment{algoboxed}
 	{\begin{mdframed}\begin{algo*}}
 	{\end{algo*}\end{mdframed}}

\newcommand{\mN}{\ensuremath{\mathcal{N}}}
\newcommand{\mV}{\ensuremath{\mathcal{V}}}
\newcommand{\mW}{\ensuremath{\mathcal{W}}}
\newcommand{\mH}{\ensuremath{\mathcal{H}}}

\newcommand{\bz}{\ensuremath{\mathbf{z}}}
\newcommand{\bd}{\ensuremath{\mathbf{d}}}
\newcommand{\bw}{\ensuremath{\mathbf{w}}}
\newcommand{\ba}{\ensuremath{\mathbf{a}}}
\newcommand{\bb}{\ensuremath{\mathbf{b}}}
\newcommand{\bU}{\ensuremath{\mathbf{U}}}

\newcommand{\bx}{\ensuremath{\mathbf{x}}}
\newcommand{\by}{\ensuremath{\mathbf{y}}}

\newcommand{\bc}{\ensuremath{\mathbf{c}}}
\newcommand{\bi}{\ensuremath{\mathbf{i}}}
\newcommand{\bv}{\ensuremath{\mathbf{v}}}

\newcommand{\bn}{\ensuremath{\mathbf{n}}}

\newcommand{\bp}{\mbox{\boldmath$\rho$}}
\newcommand{\bETA}{\mbox{\boldmath$\eta$}}
\newcommand{\mD}{\ensuremath{\mathcal{D}}}

\newcommand{\mWR}{\ensuremath{\mathcal{W}_{\mathbb{R}}}}
\newcommand{\mWRs}{\ensuremath{\mathcal{W}_{\mathbb{R}^*}}}

\newcommand{\bzeta}{\ensuremath{\boldsymbol \zeta}}

\newcommand{\bzer}{\ensuremath{\mathbf{0}}}

\newcommand{\bs}{\ensuremath{\boldsymbol \sigma}}

\newcommand{\bff}{\ensuremath{\mathbf{f}}}

\newcommand{\bbf}{\ensuremath{\mathbf{f}}}

\newcommand{\bZ}{\ensuremath{\mathbf{Z}}}
\newcommand{\bQ}{\ensuremath{\mathbf{Q}}}
\newcommand{\sH}{\ensuremath{\mathscr{H}}}
\newcommand{\sC}{\ensuremath{\mathscr{C}}}
\newcommand{\sL}{\ensuremath{\mathscr{L}}}
\def\Res{\operatorname{Res}}

\def\Abs{\operatorname{Abs}}
\def\boundary{\partial\mD}

\newcommand{\htbz}{\hat{\bz}}
\newcommand{\blambda}{\ensuremath{\boldsymbol \lambda}}
\date{}

\title{Effective Coefficient Asymptotics of Multivariate Rational
Functions via Semi-Numerical Algorithms for Polynomial Systems}

\author{Stephen Melczer\thanks{University of Pennsylvania, Department of Mathematics, 209 South 33rd Street, Philadelphia, PA 19104, USA (smelczer@sas.upenn.edu)} 
\and 
Bruno Salvy\thanks{Univ Lyon, Inria, CNRS, ENS de Lyon, Universit\'e Claude
Bernard Lyon 1, LIP UMR 5668, Lyon, France (Bruno.Salvy@inria.fr)} }

\begin{document}
\maketitle

\begin{abstract}
The coefficient sequences of multivariate rational functions appear in many areas
of combinatorics. Their diagonal coefficient sequences enjoy nice
arithmetic and asymptotic properties, and the field of analytic combinatorics in
several variables (ACSV) makes it possible to compute asymptotic expansions.
We consider these methods from the point of view of effectivity.  
In particular, given a rational
function, ACSV requires one to determine a (generically) finite
collection of points that are called critical and minimal. 
Criticality is an algebraic condition, meaning it is well treated by 
classical methods in computer algebra, while minimality is a
semi-algebraic condition describing points on the boundary of the
domain of convergence of a multivariate power series. 
We show how to obtain dominant asymptotics for the diagonal
coefficient sequence of multivariate rational functions under some
genericity assumptions using symbolic-numeric techniques.  
To our knowledge, this is the first completely automatic treatment and
complexity analysis for the asymptotic enumeration of rational
functions in an arbitrary number of variables.
\end{abstract}


\section{Introduction}
\subsection{Analytic Combinatorics}\label{sec:introAC}
Analytic combinatorics is a powerful technique to compute the
asymptotic behaviour of univariate sequences of complex numbers $
(f_k)_{k \geq 0}$ when the generating function of the sequence, $F(z)
:= \sum_{k \geq 0} f_kz^k$, is analytic in a neighbourhood of the
origin. 
The sequence is recovered by a Cauchy integral 
\begin{equation} f_k = \frac{1}{2\pi i} \int_C \frac{F(z)}{z^k} \cdot \frac{dz}{z}, \label{eq:uCIF} 
\end{equation}
where $C$ is any counter-clockwise circle sufficiently close to the
origin. The asymptotic analysis of this integral as~$k$ tends
to~$\infty$ is then obtained
by deforming the contour of integration so that it gets closer to
the singularities of minimal modulus (called \emph{dominant
singularities}). This process relates the
asymptotic behaviour of the sequence $(f_k)$ to the local behaviour of its generating function $F(z)$ near
these singularities. In particular, in the very frequent case where
the generating function has finitely many singularities in the complex
plane and at each dominant singularity~$\rho$ the
function admits a local expansion as a sum of monomials of the form
\[C(1-z/\rho)^\alpha\log^r\frac1{1-z/\rho},\quad
z\rightarrow\rho,\]
with $r\in\mathbb{N}$, then each such monomial with~$\alpha\in
\mathbb{C}\setminus\mathbb{N}$
contributes 
\begin{equation}\label{eq:sing-behav}
\frac{C}{\Gamma(-\alpha)}\rho^{-k}k^{-\alpha-1}\log^rk\,
(1+\dotsb),\quad
k\rightarrow\infty
\end{equation}
to the asymptotic behaviour of the coefficients (the ellipsis
`$\dotsb$' above corresponds to a full asymptotic
expansion given by Jungen~\cite{Jungen1931}).
When~$\alpha\in\mathbb{N}$ and~$r\neq0$, a simpler formula is
available; the terms with~$r=0,\alpha\in\mathbb{N}$ do not contribute.
Summing these contributions over all dominant
singularities gives arbitrarily many
terms of the asymptotic expansion of~$f_k$ as~$k\rightarrow\infty$.

In many cases, the combinatorial or probabilistic
origin of the sequence translates into simple equations for~$F(z)$,
from where location of singularities and local behaviour can be
computed. This is the heart of analytic combinatorics, for which we
refer to the now standard book of Flajolet and Sedgewick~\cite{FlajoletSedgewick2009}, where
the theory is introduced in detail, with proper handling of singular
behaviour more general than~\eqref{eq:sing-behav}, along with many illuminating
examples. 

\subsection{Analytic Combinatorics in Several Variables (ACSV)}
Over a series of recent papers culminating in a textbook compiling
their results, Pemantle and Wilson~\cite{PemantleWilson2002,PemantleWilson2004,PemantleWilson2008,PemantleWilson2013} 
and their collaborators have developed a theory of analytic
combinatorics in \emph{several variables}. Our aim in this work
is to automate some of this theory and analyze the complexity of
this approach. 

To a multivariate
sequence~$(f_{i_1,\dots,i_n})_{(i_1,\dots,i_n) \in \mathbb{N}^n}$ is
associated a multivariate generating function 
\[ F(\bz)= F
(z_1,\dots,z_n) = \sum_{(i_1,\dots,i_n) \in \mathbb{N}^n} f_
{i_1,\dots,i_n} z_1^{i_1} \cdots z_n^{i_n} = \sum_{\bi \in \mathbb{N}^n} f_{\bi} \bz^{\bi}.\]
As in the univariate case, when this function is analytic in the
neighbourhood of the origin, now in~$\mathbb{C}^n$, the
coefficient sequence is recovered by
a Cauchy integral, 
\[f_{i_1,\dots,i_n} = \frac{1}{(2\pi i)^n} \int_T \frac{F(\bz)}{
z_1^{i_1}\dotsm z_n^{i_n}} \cdot \frac{dz_1 \cdots dz_n}{z_1\cdots
z_n},\]
where the domain of integration~$T$ is now a polytorus sufficiently
close to the origin; i.e., a product of sufficiently small circles. 
The asymptotic analysis of the multivariate sequence of coefficients is turned into a problem of univariate asymptotics by selecting
diagonal rays in the index space: the vector $(i_1,\dots,i_n)/(i_1+\dots+i_n)$ varies in a neighbourhood
of a fixed direction. In our work, we restrict further to the main
diagonal where $i_1=\dots=i_n$, but it is important to note that the
theory brings insight on the uniformity of these results with respect
to the direction. Even under these restrictions, the
asymptotic analysis is made significantly
more delicate than in the univariate case by topological issues
related to how the domain of
integration can be deformed in~$\mathbb{C}^n$ while avoiding the singularities of the integrand. 
Pemantle and
Wilson show that an important part is played by those singularities
of~$F$ that are \emph{critical points} of the map
\[\Abs:(z_1,\dots,z_n)\mapsto|z_1\dotsm
z_n|,\]
on the set of singularities of~$F$ (precise definitions are
given in Section~\ref{sec:ACSV}). Among those critical points, one has
to determine the \emph{minimal} ones, which lie on the boundary of the
domain of convergence of the generating function. The determination of these minimal critical points is the main focus of the present work.

\subsection{Rational Functions and their Diagonals}
In order to automate this approach in computer algebra, we first
restrict the class of functions and sequences under consideration and
study only multivariate \emph{rational} generating functions: $F
(\mathbf{z})=G(\mathbf{z})/H(\mathbf{z})$ with $G$ and~$H$
polynomials in~$\mathbb{Z}[z_1,\dots,z_n]$ and~$H(\mathbf{0})\neq0$. One
motivation is that all the tools of computer algebra related to
polynomial systems become available to us. Another motivation comes
from structural properties. The generating function of the
diagonal coefficients is a classical object called the 
\emph{diagonal} of~$F$, denoted $\Delta F$ and defined by:
\[\Delta F(t):=\sum_{k\ge0}{f_{k,\dots,k}t^k}.\]
Diagonals of rational functions in~$
\mathbb{Q}(\bz)$ form an important class of power series that
contains the algebraic power series~\cite{Furstenberg1967} and 
is contained in the set of
\emph{differentially finite} power series~\cite{Christol1984}; these are the
power series solutions of linear differential equations with
polynomial coefficients. 
Among differentially finite power series, diagonals of
rational power series enjoy special properties: all their
singularities are regular with rational exponents~\cite{Katz1970,ChudnovskyChudnovsky1985,Andre2000a}. This implies that
the asymptotic expansion of their coefficients is a linear combination
of expressions of the form $C^kk^\alpha\log^p k$,
with $C$ an algebraic number, $\alpha$ a \emph{rational} number and
$p$ a non-negative integer. Despite these special properties, a conjecture of
Christol~\cite[Conjecture 4]{Christol1990} asserts that the generating functions of
univariate
\emph{integer} sequences having a finite nonzero radius of
convergence and satisfying a linear differential equation with
polynomial coefficients are all diagonals of rational functions.
Thus, diagonals of rational functions form an important class 
from the point of view of applications.

For generic rational functions, the critical points mentioned above
are obtained as solutions of a system of polynomial equations
\begin{equation}\label{eq:critical}
H=z_1\frac{\partial H}{\partial z_1}=z_2\frac{\partial H}{\partial
z_2}=\dots=z_n\frac{\partial H}{\partial z_n}.
\end{equation}
In the most common situations considered in this article, we can
avoid the use
of amoebas and Morse theory that are developed by Pemantle and Wilson
in their most recent works.
Instead, the computations are reduced to problems of complex
roots of
polynomial systems such as~\eqref{eq:critical} for the determination of critical
points, and
real roots of polynomial systems with inequalities for the selection
of the minimal ones.

\subsection{Combinatorial Case} 
We start with a special case that often arises
in practice, and where determining the minimal critical points
is greatly simplified. A rational function
$F(\bz)$ is called \emph{combinatorial} if every coefficient of its power series expansion is non-negative. This usually occurs when the rational function has been obtained by a combinatorial process. In general, it cannot be detected automatically \emph{a priori}. Indeed, even in the univariate case, the question of nonnegativity of the sequence of Taylor coefficients of a rational function is only conjectured to be decidable in general~\cite{OuaknineWorrell2012,OuaknineWorrell2014}.

Informally, our first main result is the following, which is stated
precisely in Theorem~\ref{thm:final-complexity} below and which we
gave earlier without the genericity analysis~\cite{MelczerSalvy2016}.
\begin{result}\label{result:1}
Let $G(\bz)$ and~$H(\bz)$ be polynomials in~$\mathbb{Z}
[z_1,\dots,z_n]$ of degrees at most~$d$, with coefficients of absolute
value at most $2^h$ and assume that~$H(\mathbf{0})\neq0$. Assume that
$F(\bz)=G(\bz)/H(\bz)$  is combinatorial, has a minimal critical
point, and satisfies certain verifiable assumptions stated in Section~%
\ref{sec:CombCaseResults}, that hold generically.  Then there exists a
probabilistic algorithm computing dominant asymptotics of the diagonal
sequence in $\tilde{O}(hd^{4n+1}+h^3d^{3n+3})$ bit operations\footnote{We 
write $f = \tilde{O}(g)$ when $f=O(g\log^kg)$ for some $k\ge0$; see 
Section~\ref{sec:Algorithms} for more information on our complexity model 
and notation.}.  The algorithm returns three rational functions 
$A,B,C \in \mathbb{Z}(u)$, a square-free polynomial $P \in \mathbb{Z}[u]$ 
and a list $U$ of roots of $P(u)$, specified by isolating regions, such that
\begin{equation}\label{eq:asymptcoeffs}
 f_{k,\dots,k} = (2\pi)^{(1-n)/2}\left(\sum_{u \in U} A(u)\sqrt{B(u)} \cdot C(u)^k \right)k^{(1-n)/2}\left(1 + O\left(\frac{1}{k}\right) \right). 
\end{equation}
The values of $A(u), B(u),$ and $C(u)$ can be determined to precision
$2^{-\kappa}$ at all elements of $U$ in $\tilde{O}(d^{n}\kappa + h^3d^{3n+3})$ bit operations.
\end{result}

\begin{example}\label{ex:Apery1} The sequence of Ap\'ery numbers
$A_k=\sum_{i=0}^k{\binom{k}{i}^2\binom{k+i}{i}^2}$, appearing in Ap\'ery's proof
of the irrationality of $\zeta(3)$, is, like all
multiple binomial sums, the diagonal of a rational function that can
be determined algorithmically by the methods of
Bostan et al.~\cite{BostanLairezSalvy2016}. In this example, they are given for
instance as the diagonal of a rational function in 4~variables:
\[\sum_{k\ge0}A_kt^k=\Delta\left(\frac1{1-z(1+a)(1+b)(1+c)
(1+b+c+bc+abc)}\right).\]
From there, our Maple implementation\footnote{The code for the
examples in this article is available at \url{http://diagasympt.gforge.inria.fr}.}
gives
\begin{Verbatim}
> F:=1/(1-z*(1+a)*(1+b)*(1+c)*(1+b+c+b*c+a*b*c)):
> A, U := DiagonalAsymptotics(numer(F),denom(F),[a,b,c,z],u,k):
> evala(allvalues(subs(u=U[1],A)));
\end{Verbatim}
\[{\color{blue}
\frac{\sqrt{2}\sqrt{24+17\sqrt{2}}\,\left(17+12\sqrt{2}\right)^k}
{8\pi^{3/2}k^
{3/2}}
}\]
\end{example}
In general, it is not possible to provide such an
explicit closed form for the quantities involved in the asymptotic behaviour. The output will then be a combination of an exact symbolic representation and a precise numerical estimate.
\begin{example}
In Pemantle and Wilson~\cite[Section 4.9]{PemantleWilson2008}, the authors study sequence alignment problems with application to molecular biology. The authors give asymptotics for a family of sequences parametrized by natural numbers $k$ and $b$; for any fixed $k$ and $b$ we can automatically recover their results. For instance, when $k=b=2$ one wants to derive asymptotics of
\[ \Delta\left(\frac{x^2y^2-xy+1}{1-(x+y+xy-xy^2-x^2y+x^2y^3+x^3y^2)}\right), \]
which can be shown to be combinatorial through generating function manipulations. Let $F(x,y)$ be this bivariate rational function. Running 
\begin{Verbatim}
> A, U := DiagonalAsymptotics(numer(F),denom(F),[x,y],u,k, u-x-t,t):
\end{Verbatim}
which specifies the optional linear form $u=x+t$ to be used in the algorithm (see below for details) and simplifies the output, returns $A$ equal to
\[{\color{blue}
\,{\frac {\left( 4\,{u}^{4}-14\,{u}^{3}+14\,{u}^{2}-2\,u+
2 \right) }{\sqrt {n}\sqrt {2 \pi} \left( 10\,{u}^{4}-40\,{u}^{3}+54\,{u
}^{2}-26\,u+4 \right) } \left( {\frac {10\,{u}^{4}-40\,{u}^{3}+54\,{u}
^{2}-26\,u+4}{4\,{u}^{4}-19\,{u}^{3}+25\,{u}^{2}-4\,u-6}} \right) ^{n}
\sqrt {{\frac {10\,{u}^{4}-40\,{u}^{3}+54\,{u}^{2}-26\,u+4}{4\,{u}^{4}
-16\,{u}^{3}+20\,{u}^{2}-8\,u+4}}}} } \]
and $U$ equal to
\[{\color{blue}
[\text{RootOf}(2\_Z^5-10\_Z^4+18\_Z^3-13\_Z^2+4\_Z-2, \, 1.4704170\dots)]}
\]
where 366 decimal places are recorded: this is an upper bound on the accuracy needed by the algorithm to rigorously decide numerical equalities and inequalities. Asymptotics are given by evaluating $A$ at the degree 5 algebraic number defined by the single element of $U$ (which is not expressible in radicals).
\end{example}

\subsection{Non-Combinatorial Case}
We also propose an algorithm finding minimal critical points in many
cases, even when combinatoriality is not assumed, at the price of an
increase in complexity. Our result in that case is the
following, which is stated precisely in Theorem~\ref{thm:final-complexity} below.
\begin{result} 
Let $F(\bz) \in \mathbb{Z}(z_1,\dots,z_n)$ be a rational function with
numerator and denominator of degrees at most $d$ and coefficients of
absolute value at most $2^h$.  Assuming that $F$ satisfies certain
verifiable assumptions stated in  Section~\ref{sec:GenResults}, then
$F$ admits a finite number of minimal critical points that can be
determined in $\tilde{O}\left(hd^{9n+5}2^{3n}\right)$ bit operations.
From there, the asymptotics of the diagonal coefficients follow with
the same complexity as in Result~\ref{result:1}.
\end{result}
Aside from the existence of minimal critical points, we
conjecture that the assumptions on $F$ required to apply Theorem~%
\ref{thm:final-complexity} in the non-combinatorial case hold
generically.

\subsection{Previous work}
A very useful introduction to the asymptotics of sequences is given in
the extensive survey by Odlyzko~\cite{Odlyzko1995}. Here, we focus on the case
of coefficients of rational functions and on effective
methods.

\paragraph{Univariate case} Finding the asymptotic behaviour of the
coefficients of a univariate rational function is equivalent to
finding that of a linear recurrence with constant coefficients.
Decision procedures rely on the ability to determine whether two complex algebraic numbers have the same modulus. This can be done purely algebraically, and a semi-numerical algorithm has been given by Gourdon and Salvy~\cite{GourdonSalvy1996}. Some of the ingredients are common with the current work, in particular a semi-numerical approach to those types of decision problems for algebraic numbers.

\paragraph{Probabilistic approach}
Many combinatorial sequences are given as sums of non-negative terms and several techniques are available in that case, surveyed in the classic book by de Bruijn~\cite{De-Bruijn1981}. For instance, completely explicit formulas can be derived for sums of products of binomial coefficients~\cite{McIntosh1996}.

Given the combinatorial generating function~$F(\bz)$, the normalized
sequence~$f_
{i_1,\dots,i_n}/\sum_{j_1+\dots+j_n=k}{f_{j_1,\dots,j_n}}$, where
$k=i_1+\dots+i_n$, is a
discrete probability for which central and local limit theorems
for large $k$ have
been derived in the bivariate case by Bender~\cite{Bender1973} and
later extended by Bender, Richmond, and Gao~\cite{BenderRichmond1983,BenderRichmond1999,
GaoRichmond1992}. The local limit theorems are the most relevant to
our discussion. Let $x_1,\dots,x_n$ be real positive numbers and
consider the
univariate generating function $f_\bx(z)=F(zx_1,\dots,zx_n)$.
Assume that in a neighbourhood of~$\bx$, there exists an analytic
root~$\lambda(\bx)$ of the denominator of~$f_\bx$ such that the other
roots have strictly larger modulus. Then, if the numerator of $F$ does
not vanish in a neighbourhood of~$\lambda(\bx)$ and the 
matrix $(x_1x_j\partial^2\lambda/\partial x_i\partial x_j)$ is not
singular, the monomials $f_{\bi}\bx^{\bi}$ satisfy a local limit theorem with mean
$k(x_1\partial\lambda/\partial x_1,\dots,x_n\partial\lambda/\partial
x_n)$. Shifting the mean by choosing $\bx$ so that this mean is on the
diagonal then gives the desired asymptotic behaviour, provided that
$\lambda(\bx)$ still satisfies the required assumptions. The
derivatives of~$\lambda$ are related to those of the denominator~$H
(\bz)$ of the rational function~$F(\bz)$, and the equality of the
coordinates above then amounts to
\[z_1\frac{\partial H}{\partial z_1}=\dots=z_n\frac{\partial H}
{\partial z_n}\qquad\text{at}\quad\bz=\bx.\]
These are the same \emph{critical point equations} as Equation~\eqref{eq:critical}, to
which we devote most
of this work. (See the precise version given by Gao and Richmond~\cite{GaoRichmond1992}
in their Theorem~4, where the result is expressed in terms of a
$\boldsymbol{t}$, which is what we compute and for which we give
complexity estimates.) As in the case of ACSV, by restricting to the
case of rational functions, we can bring in tools for computer algebra
and design complete algorithms, along with a complexity analysis.

\paragraph{Bivariate ACSV}
Similarly, Pemantle and Wilson's analytic combinatorics in several
variables apply much more generally than for combinatorial rational
generating functions. In terms of algorithms, the situation is much
harder. Only the case of \emph{bivariate} rational functions $F(x,y)$
that are not required to be
combinatorial and under a smoothness hypothesis do we have an
algorithm, due to 
de Vries et al.~\cite{DeVriesHoevenPemantle2011}. It is not immediately clear how to
generalize this technique beyond the bivariate case and keep it 
effective. While our algorithms apply in higher dimension, they work
under stronger minimality assumptions on critical points.

\paragraph{Previous implementations}
A Sage package of Raichev~\cite{Raichev2012} determines asymptotic
contributions of non-de\-gen\-er\-ate critical points where the zero
set
of~$H$ is smooth or locally the transverse intersection of smooth
algebraic varieties. It relies on the assumption that minimality
of these points has been proved beforehand, the most difficult
step of the analysis.

\paragraph{Creative telescoping}
Another approach to the computation of these asymptotic behaviours
exploits the fact that diagonals of rational functions are
differentially finite. A possible starting point is to use an integral
representation for the diagonal as a multidimensional residue:
\[\Delta F(t)=\frac1{(2\pi i)^{n-1}}\oint{F\!\left
(z_1,\dots,z_{n-1},\frac{t}{z_1\dotsm z_{n-1}}\right)\frac{dz_1\dotsm
dz_{n-1}}{z_1\dotsm z_{n-1}}}.\]
Next, a technique called the Griffiths-Dwork method performs a
succession of computations modulo a polynomial ideal. For our
case of a rational function~$F=G/H$, this is the ideal generated by
the
critical point equations~\eqref{eq:critical} again. The result of this
method is a linear differential equation satisfied by~$\Delta F$. An
efficient algorithm with arithmetic complexity in~$d^{O(n)}$ has been
given by Bostan et al.~\cite{BostanLairezSalvy2013} and improved by
Lairez~\cite{Lairez2016}.

From this differential equation, univariate singularity analysis
applies, following Flajolet and Sedgewick~\cite[\S VII.9.1]{FlajoletSedgewick2009}. First,
the possible locations of singularities are the zeros of the
leading coefficient of the equation. 
Next, at such a point~$\rho$,
since the equation is Fuchsian with rational exponents, there exists a
basis of local expansions of the form
\[(z-\rho)^\alpha\left(\phi_r
(z)\log^r\frac1
{1-z/\rho}+\dots+\phi_0(z)\right),\]
where~$r\in\mathbb N$, $\alpha\in\mathbb Q$ and
the~$\phi_k$ are convergent
power series in powers of~$(z-\rho)$
that can be computed to arbitrary order. 
The generating
function~$\Delta F$, known at the origin to arbitrary order can be
analytically continued numerically to~$\rho$ and its coefficients~$c_
{\alpha,r}$ in
that basis can be computed numerically efficiently with arbitrary
precision~\cite{Mezzarobba2010,Mezzarobba2016}. From there, as
outlined in~\S\ref{sec:introAC}, a contribution to the asymptotic
expansion of~$f_{k,\dots,k}$ follows for all~$\alpha,r$ not
in~$\mathbb{N}\times\{0\}$ such that~$c_{\alpha,r}\neq0$.
In the common
case when the coefficient~$c_{\alpha,r}$ corresponding to the dominant
part of the asymptotic behaviour is nonzero, then 
it can be recognized to be so from a certified numerical approximation
and the
asymptotic behaviour follows. It has the same shape as in Equation~%
\eqref{eq:asymptcoeffs}, with three main differences: the constant in
front is given only numerically, approximated rigorously to any fixed accuracy; 
the exponent is not restricted to
being a half-integer; a full asymptotic expansion is easily produced.
This last point in particular shows that when both methods apply, they
are complementary: ACSV yields a closed-form expression for the
relevant scalar factor~$c_{\alpha,0}$ and from there, a full
asymptotic expansion is easily computed from the differential
equation derived by creative telescoping. The methods of ACSV
are also capable of deriving higher order terms in the asymptotic 
expansion, however at a higher computational cost.

\paragraph{Polynomial systems}
There is an extensive literature in computer algebra on the complexity
of analyzing the roots of a polynomial system such as the one provided
by the critical point equations~\eqref{eq:critical}. Our work on
this system relies on ideas by a variety of authors~\cite{GiustiHeintzMoraisMorgensternPardo1998,GiustiLecerfSalvy2001,Schost2001,KrickPardoSombra2001} 
on the use of the Kronecker representation in complex or real geometry, which go far beyond the simple systems we consider here. More precisely, we make use of the recent work of Safey El Din and Schost~\cite{Safey-El-DinSchost2018}, who take into account multi-homogeneity and provide estimates on the height of the representations and the bit complexities of their algorithms.
Note that as this work reached completion, a new preprint by
van der Hoeven and Lecerf~\cite{HoevenLecerf2018} appeared that points to the possibility of
improving further the exponent of $d^n$ in our results, while
retaining the same approach.
The Kronecker representation, and similar constructions, have also appeared in the literature under the name `rational univariate representation'~\cite{Rouillier1999,BasuPollackRoy2006}.  To the best of our knowledge, the connection between the good properties of the Kronecker representation in terms of bit size and the fast and precise algorithms operating on univariate polynomials had not been explored before Melczer and Salvy~\cite{MelczerSalvy2016}, except in the case of bivariate systems~\cite{BouzidiLazardPougetRouillier2015,KobelSagraloff2015}.

\subsection{Outline}
This article is structured as follows.
In Section~\ref{sec:ACSV}, we give an almost
self-contained introduction to Analytic
Combinatorics in Several Variables at a more elementary level than in
the book of Pemantle and Wilson~\cite{PemantleWilson2013}. This can serve as an
introduction to the subject for combinatorialists already acquainted
with analytic combinatorics in one variable. It can also be skipped by
those readers who are only interested in the algorithms. They will
find
in Section~\ref{sec:algo-overview} an overview of the operations that
need be performed in order to compute the asymptotic behaviour. Next,
in Section~\ref{sec:Algorithms}, we introduce the Kronecker
representation. The results of Safey El Din and Schost~\cite{Safey-El-DinSchost2018}
that we need are recalled. They are used to analyze the cost of
several other operations
on solutions of polynomial systems. These results are illustrated
on the polynomial systems arising in ACSV.
Section~\ref{sec:numkro}
turns
to the semi-numerical part of the computation. A numerical Kronecker
representation is defined and the precision required for several
decision problems is analyzed. Again, these are illustrated by
families of examples from ACSV. We then turn back to ACSV in
Section~\ref{sec:MainAlgos}, where the algorithms outlined in
Section~\ref{sec:algo-overview} can finally be specified more
precisely thanks to our semi-numerical tools. Section~%
\ref{sec:Examples} gives a few more examples and Section~\ref{sec:generic}
addresses the genericity of our assumptions in the combinatorial case.

\section{Analytic Combinatorics in Several Variables for Rational Functions}
\label{sec:ACSV}
We give an almost self-contained introduction to the
part of the theory of analytic combinatorics in
several variables that we need, introducing the definitions and
notation
for the rest of the article. Since our algorithms only address 
situations
where the geometry is sufficiently simple, we stick to the ``surgery
method'' of the early works of Pemantle and Wilson~\cite{PemantleWilson2002} and avoid any
mention of Morse theory so that the text is more
accessible to combinatorialists already familiar with the univariate
situation. We also avoid amoebas; while they give a simple
understanding of some properties of domains of convergence, they
introduce logarithms that we want to avoid in the computations.

\subsection{Domains of convergence and minimal points}
For basic
properties of analytic functions in several variables, we refer
to Krantz~\cite{Krantz1992} and Hormander~\cite{Hormander1990}.
We consider a multivariate rational function 
$F(\bz) = {G(\bz)}/{H(\bz)}$ with $G$ and $H$ co-prime polynomials
and $H(\bzer)\neq0$. (A large part of the analysis holds more
generally for meromorphic
functions, $G$ and $H$ being co-prime analytic functions.) The Taylor
expansion at the origin
\begin{equation}\label{eq:powerseries}
F(\bz)= \sum_{\bi \in
\mathbb{N}^n} f_{\bi} \bz^{\bi}
\end{equation}
has a nonempty open domain of convergence $\mD \subset \mathbb{C}^n$.

The main properties of the multivariate case that we use are:
\begin{itemize}
	\item[--] a point~$\mathbf{z}:=(z_1,\dots,z_n)$ is in the closure
$\overline{\mD}$ of $\mD$ if and only if the open
polydisk $D(
\mathbf{z}):=\{\mathbf{w}\in
\mathbb{C}^n\mid |w_i|<|z_i|, i=1,\dots,n\}$ is a
subset of~$\mD$;
\item[--] the domain $\mD$ is \emph{logarithmically convex}:
if
the points $\bz=(z_1,\dots,z_n)$ and $\bzeta=(\zeta_1,\dots,\zeta_n)$
are in~$\mD$, then so are the
points $(|z_1|^t|\zeta_1|^{1-t},\dots,|z_n|^t|\zeta_n|^{1-t})$ for
$t\in[0,1]$.
\end{itemize}

The boundary~$\boundary=\overline{
\mD}\setminus\mD$ of the domain of convergence plays
an important role in the analysis. At its points, the
series~\eqref{eq:powerseries} is not absolutely convergent.
The following
result summarizes
the relation between this boundary and
the algebraic set $\mV:=\{\bz\in\mathbb{C}^n\mid H(\bz)=0\}$, which Pemantle and Wilson call the \emph{singular variety}
of~$F$. 
\begin{lemma}\label{lemma:domconv}
(i) If $\bw\in\boundary$, then there exists $\bz\in\mV$ in
its
polytorus $T(\bw):=\{(w_1e^{i\theta_1},\dots,w_ne^{i\theta_n})\mid 
(\theta_1,\dots,\theta_n)\in\mathbb{R}^n\}$.\par
\noindent
(ii) The intersection $\mV\cap\mD$ is empty. (iii) If
$\bw$ is in $\mV$ and $\mV\cap D(\bw)$ is empty, then
$\bw$ belongs to $\boundary$.
\end{lemma}

\begin{proof}
\emph{(i)}. 
If no point of $T(\bw)$
belongs to~$\mV$, then the
function~$F$
admits a convergent power series expansion at each of these points. As~$T(\bw)$
is compact, there is a $\rho>0$ such that all these power
series converge in a polydisk of radius~$\rho$. These allow for an
analytic continuation of $F$ in a polydisk $D((|w_1|
(1+\rho/2),\dots |w_n|(1+\rho/2)))$, which implies that
the power series~\eqref{eq:powerseries} is absolutely convergent
at~$\bw$, contradicting the
fact that $\bw\in\boundary$.

\emph{(ii).} If $\bw\in\mV$, then $H(\bw)=0$. If $G(\bw)\neq0$ then
$F$ is infinite at~$\bw$ and thus
its Taylor series does not converge in its neighbourhood. Otherwise,
up to renumbering the variables we can assume that $\partial
H/\partial z_n\neq0$ and then,
since $H$ and $G$ are coprime, their gcd when
viewed as
polynomials in~$\mathbb{C}(z_1,\dots,z_{n-1})[z_n]$ is~1 and 
there exist polynomials $U$, $V$ in $\mathbb{C}[\bz]$ and $W$
nonzero in $\mathbb{C}[z_1,\dots,z_{n-1}]$ such that
$W=UG+VH$.
If there was a neighborhood of $\bw$ where $G=H=0$ then the
nonzero polynomial $W$ would be~0
in a neighborhood of $(w_1,\dots,w_{n-1})$, but this is
impossible. Thus
there exists $\bw'$ arbitrarily close to~$\bw$, where $H(\bw')=0$ and
$G(\bw')\neq0$ and therefore $\bw\not\in\mD$.

\emph{(iii).} The proof is similar to that of \emph{(i)}. At any
point
$\bz$ of $D(\bw)$, the function $F$ admits an analytic continuation,
which
implies that the power series~\eqref{eq:powerseries} is absolutely
convergent at $\bz$ and thus that $D(\bw)\subset\mD$. Thus $\bw$ is in
$\overline\mD$. By  \emph{(ii)}, it is not in $\mD$ so that it
belongs to~$\boundary$.
\end{proof}

\noindent Thus a special role is played by points in~$\partial
\mD\cap\mV$.
\begin{definition}
The
elements of $\partial\mD\cap\mV$ are called \emph{minimal points}.
A minimal point~$\mathbf{z}$ is called \emph{finitely minimal} when
its
polytorus~$T(\mathbf{z})$ intersects~$\mathcal{V}$ in finitely many
points. It is called \emph{strictly minimal} when this intersection is
reduced to~$\{\mathbf{z}\}$. It is called \emph{smooth} when the
gradient $\nabla H$ does not vanish at~$\bz$.
\end{definition}


\begin{figure}
\centerline{
\includegraphics[width=.3\textwidth]{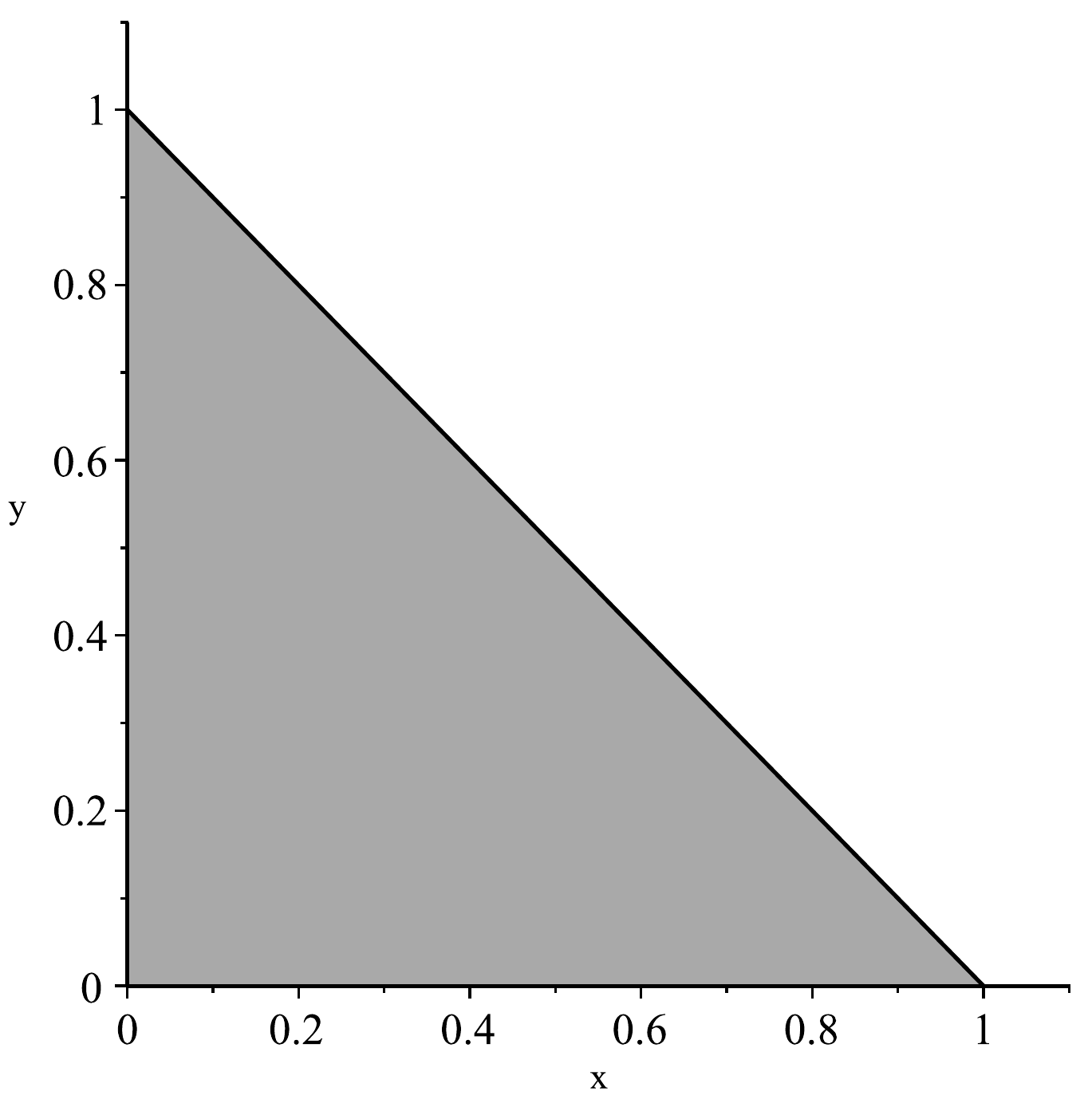}
\qquad
\includegraphics[width=.3\textwidth]{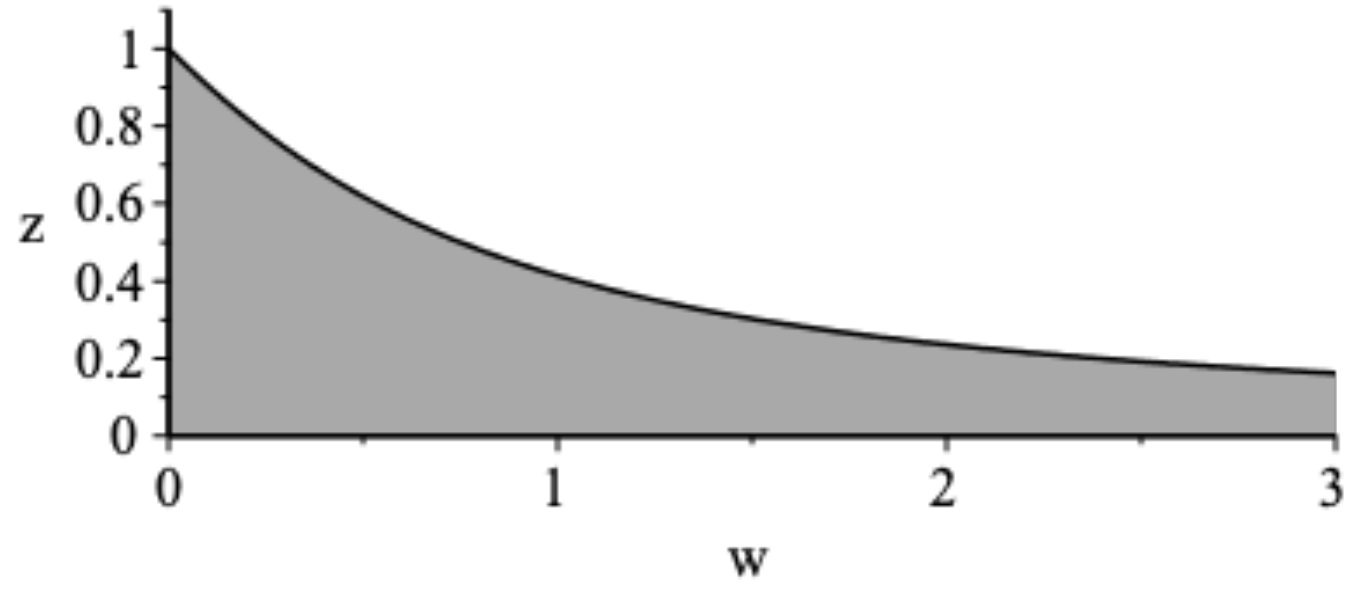}
\qquad
\includegraphics[width=.3\textwidth]{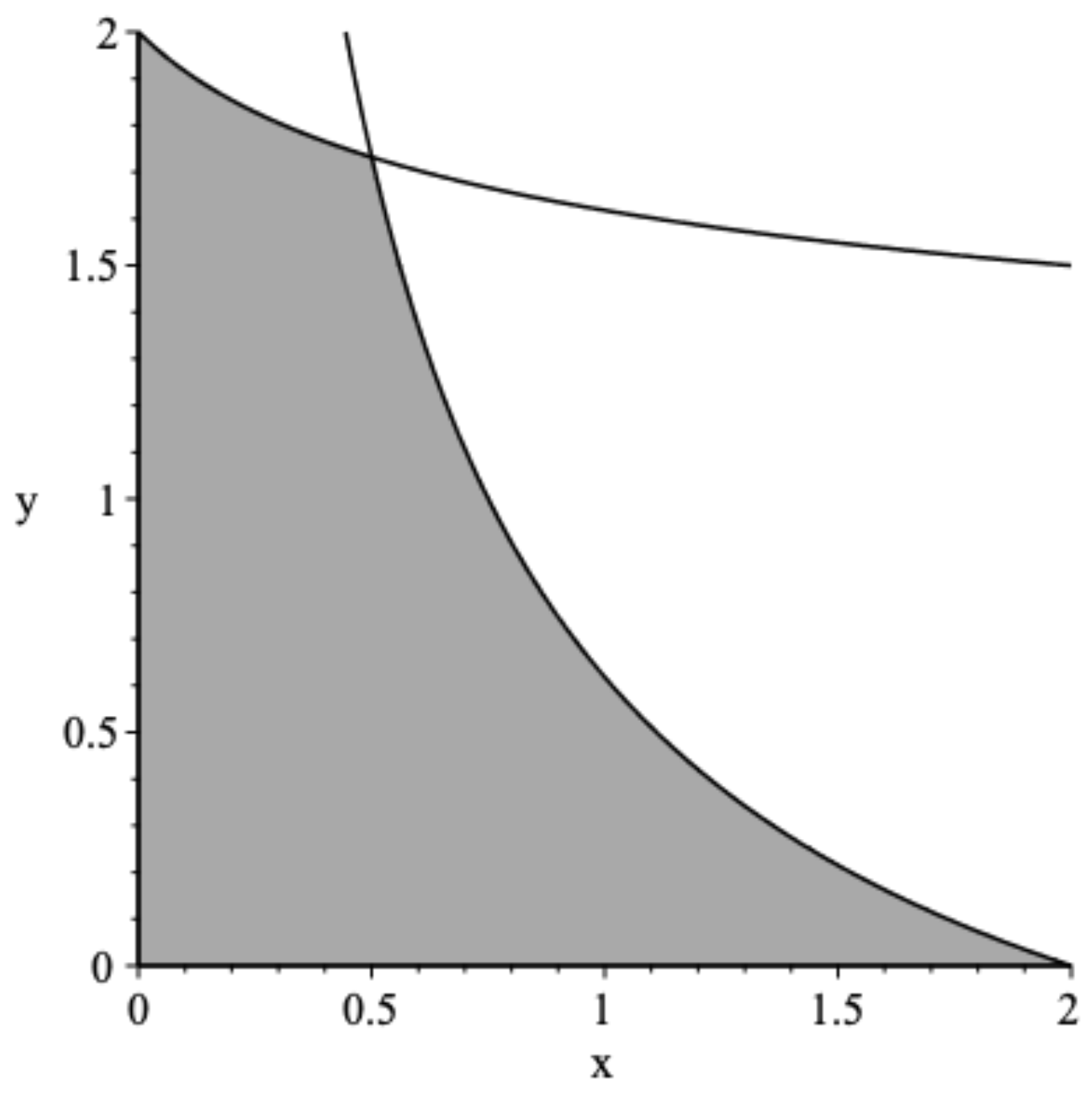}
}
\caption{The gray areas give the intersections of the domains of
convergence with $
\mathbb{R}_+^2$ in three examples. Left: $1/(1-x-y)$; middle: $1/
(1-2wz-z^2)$; right: $1/(2+y-x(1+y)^2)$.
\label{fig:ex-domain}}
\end{figure}

\begin{example}\label{ex:easy1}
If $F(x,y)=1/(1-x-y)$, the singular
variety~$
\mathcal{V}$ is parameterized by~$(x,1-x)$ for $x\in\mathbb{C}$. All
its points are smooth: the gradient is the constant vector~$(-1,-1)$.
A point of~$\mathcal{V}$ is minimal when there does not exist another
point~$(x',y')$ in~$\mathcal{V}$ with $|x'|<|x|$ and~$|y'|<|1-x|$. By
continuity of $1/F$, it is sufficient to check that there does not
exist a
minimal point where one of these inequalities becomes an equality. 

No point of~$\mathcal{V}$ with $|x|>1$ is minimal, since for such
a point, $(1,0)\in\mathcal{V}$ has smaller modulus coordinate-wise. If
$|x|\le 1$
and $x$ is not
real or is negative, then $0\le 1-|x|<|1-x|$, so that the existence
of the point $
(|x|,1-|x|)\in\mathcal{V}$ prevents~$(x,1-x)$ from being minimal.

The conclusion is that the only possible minimal points are of the
form~$(x,1-x)$ with $x$ real in $[0,1]$. These are indeed minimal
since any
point~$(x',y')$ with $|x'|<x$ and~$|y'|<1-x$ satisfies $|x'+y'|<1$ and
thus lies inside the domain of convergence. (See Figure~%
\ref{fig:ex-domain}, left.)
\end{example}
\begin{example}\label{ex:ex7} Consider the rational function~$F
(w,z)=1/(1-2wz-z^2)$.
Its singular variety~$\mathcal{V}$ is parameterized by~$(w(z),z)$ with
$w(z)=(1-z^2)/
(2z)$ and~$z\in\mathbb{C}\setminus\{0\}$. All its points are
smooth: the gradient~$(-2z,-2w-2z)$ does not vanish on~$\mV$.

None of those points with~$|z|>1$ can be
minimal: for the same value of~$w$ the denominator
of~$F$ has another root of
smaller modulus~$1/|z|$.
Similarly, if $|z|<1$ and $z$ it not real, then~$z'=|z|$ is such that
$|1-z'^2|<|1-z^2|$ so that again there is another point of~$\mathcal{V}$
with smaller modulus.  Finally, minimal points with~$|z|=1$ must also
be
real: if $z=\exp(i\theta)$, then $|w(z)|=\left|\sin\theta\right|$
which is
minimal
when~$\theta=0\bmod \pi$. 

In summary, the only possible minimal points are of the form~$
(-u/2+1/(2u),u)$ for $u\in[-1,1]\setminus\{0\}$. These are indeed
minimal
as a consequence of~$u\mapsto-u/2+1/(2u)$ being decreasing for
positive~$u$. Each of them is finitely minimal, its opposite also 
being in~$\mathcal{V}$. (See Figure~\ref{fig:ex-domain}, middle.)
\end{example}
\noindent Smooth minimal points play an important part in this theory.
Their role is explained by the following result.
\begin{proposition}[{PemantleWilson~\cite[Lemma~2.1]{PemantleWilson2002} and 
Pemantle and Wilson~\cite[Prop.~3.12]{PemantleWilson2008}}]\label{prop:3.12}
Let $\bw$ be a smooth minimal point. Then there exist non-negative
\emph{real} numbers $\lambda_1,\dots,\lambda_n$, not all zero, such that:
\begin{enumerate}
	\item   $\left
(w_1\frac{\partial
H}{\partial z_1}(\bw),\dots,w_n\frac{\partial H}{\partial z_n}
(\bw)\right)$ and
$(\lambda_1,\dots,\lambda_n)$ are colinear;
\item
the point $\bw$ is a local maximizer of
the map $\bz\mapsto \left|z_1^{\lambda_1}\dotsm z_n^
{\lambda_n}\right|$ on~$\boundary$.
\end{enumerate}\label{prop:lambda}
\end{proposition}
\begin{proof}
Since $\bw$ is a minimal point, the open polydisk $D(\bw)$ is included
in~$\mD$ and by Lemma~\ref{lemma:domconv}~(ii), it does not
contain any element of~$\mV$. Thus the tangent lines to the torus $T
(\bw)$ at~$\bw$ must belong to the tangent space to~$\mV$ at~$\bw$.
Since $\nabla H$ does not vanish at $\bw$, this leads to relations
between the partial derivatives of~$H$, obtained as follows.

Without loss of
generality,
we assume that~$(\partial H/\partial z_n)(\bw)\neq0$. 
By the implicit function
theorem, there exists an analytic function $g(\htbz)$ where $\htbz := 
(z_1,\dots,z_{n-1})$,
such that $(z_1,\dots,z_{n-1},g(\htbz))$ is a parameterization of
$\mV$ in a
neighbourhood of $\bw$: $w_n=g(\hat\bw)$, $H(\htbz,g(\htbz))=0$ and
$g$ is locally
one-to-one. 
For any~$j\in\{1,\dots,n-1\}$, 
differentiating $H(\htbz,g(\htbz))=0$ with respect to~$z_j$ yields
\[\frac{\partial H}{\partial z_j}(\bz)+\frac{\partial H}
{\partial z_n}(\bz)\frac{\partial g}{\partial z_j}
(\htbz)=0,\]
so that the vector $(\bzer,1,\bzer,\partial g/\partial z_j(\hat\bw))$
with $1$ in the $j$th position lies in the tangent space to~$\mV$
at~$\bw$. In a neighbourhood of~$\theta=0$, the image of~$\left(\hat\bw,g(\hat\bw)\right)$ 
when $w_j$ is replaced by $w_je^{i\theta}$ moves along~$iw_j(\bzer,1,\bzer,\partial
g/\partial z_j(\hat\bw))$, and minimality of $|w_n|$ implies that this
vector should be tangent to the torus; i.e., there exists a real
$\lambda_j$ such that this vector equals
$(\bzer,iw_j,\bzer,-i\lambda_j w_n)$. Moreover, the presence of $
(\bzer,w_j,\bzer,-\lambda_jw_n)$ in the tangent plane to~$\mV$ at~$\bw$
implies $\lambda_j\ge0$, since otherwise~$\mD$ would intersect~$\mV$.
In summary, we have obtained the existence of
$\lambda_1,\dots,\lambda_{n-1},\lambda_n$, with $\lambda_n=1$, all
real and non-negative, such that
\[\lambda_nw_j\frac{\partial H}{\partial z_j}(\bw)=\lambda_j w_n
\frac{\partial H}{\partial z_n}(\bw),\qquad j\in\{1,\dots,n-1\}.\]
Linear combinations give
\[\lambda_kw_j\frac{\partial H}{\partial z_j}
(\bw)=\lambda_k\lambda_jw_n
\frac{\partial H}{\partial z_n}(\bw)=\lambda_jw_k\frac{\partial H}
{\partial z_i}(\bw),\]
which concludes the proof of the first part of the proposition.

That each of the vectors $(\bzer,i\lambda_n w_j,\bzer,i\lambda_j w_n)$
is
tangent to the torus $T(\bw)$ at $\bw$ is equivalent
to the product of $w_je^{i\lambda_j\theta_j}$ for $j=1,\dots,n$ being
multiplied by complex numbers of modulus~1 locally; i.e., its
modulus is locally constant. It then has to be a local maximum
by minimality of~$\bw$ and nonnegativity of the
$\lambda_j$s. 
\end{proof}
\begin{example}\label{example:no-critical-point}
By a reasoning similar to that of the previous
examples, the minimal points of the rational function~$F=1/(2+y-x
(1+y)^2)$ are all smooth and of the form~$((2+y)/(1+y)^2,y)$
for $y\in[-2,-\sqrt{3}]\cup
[0,\sqrt{3}]$ (See Figure~\ref{fig:ex-domain}, right.). At these
points, $(x\partial
H/\partial
x,y\partial
H/\partial y)$ is colinear to the real vector~$(\lambda_1,\lambda_2)=
(2+y,2+y-2/(1+y))$. This is never colinear to $(1,1)$, which puts
this function outside of the scope of our methods, as shown in 
Example~\ref{example:no-critical-point2} below.
\end{example}

\subsection{Exponential Growth}
The starting point in the asymptotic analysis is a Cauchy integral
representation of the diagonal
coefficients: for any~$k\in\mathbb{N}$,
\begin{equation} f_{k,\dots,k} = \frac{1}{(2\pi i)^n} \int_T \frac{F(\bz)}{(z_1\cdots z_n)^k} \frac{dz_1 \cdots dz_n}{z_1\cdots z_n}, \label{eq:mCIF}
\end{equation}
where $T$ is a polytorus sufficiently close to
the origin.

The first step is to determine the exponential growth of the diagonal sequence, 
\[\rho := \limsup_{k \rightarrow \infty} |f_{k,\dots,k}|^{1/k}.\]
A consequence of the integral
representation~\eqref{eq:mCIF} is Cauchy's inequality on the
coefficients of an analytic function, implying $\rho
\leq |z_1 \cdots z_n|^
{-1}$ for
any~$\bz\in\mD$. Then, by the maximum principle, it follows that
\[\rho \leq |z_1 \dotsm z_n|^{-1},\quad\text{for
any~$\bz\in\partial\mD$}\]
while by Lemma~\ref{lemma:domconv}~\emph{(i)}, 
\[\inf_{\mathbf{z}\in\partial D}|z_1\dotsm z_n|^{-1}=\inf_{
\mathbf{z}\in\partial D\cap\mathcal{V}}|z_1\dotsm z_n|^{-1}.\]
Minimal points are those to which the cycle of
integration
$T$ in the Cauchy integral representation~\eqref{eq:mCIF} may be taken
arbitrarily close without changing the value of the integral. In the
neighbourhood of a finitely minimal point where $|z_1\dotsm z_n|$ is
maximal, the contour can be further
deformed so as to capture the contribution of that point to the
asymptotic
behaviour of the integral. This is done in~\S\ref{sec:acsv-local}.

\subsection{Critical Points}
Instead of computing the set of minimal points first and then looking
for those that maximize $|z_1\dotsm z_n|$, it turns out to be easier
to compute a somewhat related set formed by the extrema of~$|z_1\dotsm
z_n|$ on subsets of~$ \mathcal{V}$ and then select its elements that are minimal. In many
cases,
those points are sufficient to complete the asymptotic analysis.

Thus the next step is to focus on the map
\[\Abs:\mathbf{z}\mapsto|z_1\dotsm z_n|,\]
and study its extrema on~$\mathcal{V}$.
These extrema can be obtained as solutions of an optimization
problem for the
map~$\Abs$ from $\mathbb{R}^{2n}$
to~$\mathbb{R}$, restricted to the set~$\mathcal{V}$, viewed as a
subset of~$\mathbb{R}^{2n}$.
A real-valued Lagrangian associated to this
optimization problem is $L(\mathbf{z},\lambda):=\mathbf{z}\overline{\mathbf{z}}
+2\Re(\lambda H(\mathbf{z}))$, 
where $\mathbf{z}\overline{\mathbf{z}}=z_1\overline{z_1}+\cdots+z_n\overline{z_n}$ 
and $\Re(w)$ denotes 
the real part of $w\in\mathbb{R}$. Standard arguments
make it possible to work with complex derivatives only 
(see~\cite[App.2]{SchreierScharf2010},
\cite[ch.1,\S4]{Remmert1991}, \cite{Brandwood1983}): for a 
\emph{real
valued}
function~$f$ of $z=x+iy$ and~$\overline{z}=x-iy$ that is
differentiable as a function of~$(x,y)$, the
simultaneous vanishing of~$\partial
f/\partial x$ and $\partial f/\partial y$ is equivalent to the
vanishing of $\partial f/\partial z=(\partial f/\partial x-i\partial
f/\partial y)/2$, or equivalently to the
vanishing of $\partial f/\partial\overline{z}=(\partial f/\partial
x+i\partial f/\partial y)/2$.

The extrema can only be reached in three 
situations: either at points of~$\mV$ where one of the  
coordinates~$z_i$ is~0, where $\Abs$ is not differentiable, or at 
\emph{critical points} of~$\Abs$, where either
the gradient with
respect to the complex coordinates $\nabla
H:=(\partial H/\partial z_1,\dots,\partial H/\partial z_n)$ is~0 or
where the
optimality condition~$\nabla L=0$ holds. (In
this
last case, the gradients~$\nabla |\mathbf{z}|^2$ and~$\nabla H$ are
colinear, so
that the level surface of~$\Abs(\mathbf{z})$ is tangent to~$
\mathcal{V}$.) 
\begin{lemma}\label{lemma:extrema-crit-pts} The critical points
of the map~$\Abs:\bz\mapsto|z_1\dotsm z_n|$ on~$\mathcal{V}$
are located at solutions of the equations
\begin{equation} H(\mathbf{z})=0, \qquad z_1\frac{\partial H}
{\partial z_1} = \dots = z_n\frac{\partial H}{\partial
z_n} \label{eq:critpt}.
\end{equation} 
\end{lemma}
\begin{proof}
Clearly the equations hold when the gradient of~$H$ is~0. The
remaining case is obtained by writing out the equations for the
coordinates of~$\nabla L=0$. From
\[L=\mathbf{z}\overline{\mathbf{z}}+\lambda H(
\mathbf{z})+\overline{\lambda H(\mathbf{z})},\]
it follows that for~$i\in\{1,\dots,n\}$, 
\[\frac{\partial L}{\partial z_i}=\frac{\mathbf{z}\overline{
\mathbf{z}}}{z_i}+\lambda\frac{\partial H(\mathbf{z})}{\partial
z_i}.\]
The solutions with nonzero coordinates of $\partial L/\partial z_i=0$ for
$i=1,\dots,n$ are
precisely those of Equation~\eqref{eq:critpt}.
\end{proof}
\begin{definition}
The Equations~\eqref{eq:critpt} are called the 
\emph{critical-point equations}. Their solutions are called 
\emph{critical points}, the map $\Abs$ being implicit. Those that do
not cancel the gradient $\nabla
H$ are called \emph{smooth}. 
\end{definition}
\begin{example}\label{ex:easy2}
Consider again the polynomial~$H=1-x-y$ from
Example~\ref{ex:easy1}.
The critical point equations~\eqref{eq:critpt} reduce to
$\{1-x-y=0,x=y\}$,
so that they have a unique solution~$x=y=1/2$, where the level surface
of $|xy|$ is tangent to $\mV$ (see Figure~\ref{fig:ex}, left). This
point is also
minimal, as shown in Example~\ref{ex:easy1}.

This critical point is neither a
maximum nor a minimum of~$\Abs:\bz\mapsto|xy|$ on~$\mV$, but only a
saddle point. This can be seen
either by considering the principal minors of the bordered Hessian
of~$L$, or directly, by observing that for small real positive~$t$,
the points $(1/2+t,1/2-t)$ and
$(1/2+it,1/2-it)$ lie on~$\mV$, while the map $\Abs$
takes values~$1/4-t^2$ and $1/4+t^2$ on them.

It is however a maximum of~$\Abs$ on~$\partial\mD$. Indeed, by
Lemma~\ref{lemma:domconv} and Example~\ref{ex:easy1}, the elements $
(x,y)$ of~$\partial\mD\cap\mV$ satisfy~$|y|=1-|x|$, so that the maximum of~$|xy|$ is
reached when~$|x|=|y|=1/2$. 
\end{example}
The last feature of this example is a reflection of a more general
phenomenon relating minimal critical points and maximizers of $\Abs$.
\begin{lemma}[{Pemantle and Wilson~\cite[Prop.~3.12]{PemantleWilson2008}}]
\label{lemma:min-crit}
If $\bw$ is a smooth minimal critical point, then it is a local
maximizer of the map $\Abs:\bz\mapsto|z_1\dotsm z_n|$ on
$\partial\mD$.
\end{lemma}
\begin{proof}
Since $\bw$ is critical, the vectors $(z_1\partial H/\partial
z_1,\dots,z_n\partial H/\partial z_n)$ and $(1,\dots,1)$ are colinear.
The conclusion follows from
Proposition~\ref{prop:3.12}.
\end{proof}
\noindent It is important to note that it can also happen that the
critical-point equations
do not have any solution.
\begin{example}\label{ex:12}
The generating function of Example~\ref{ex:ex7} leads
to the critical-point equations 
\[H=1-2wz-z^2=0,\quad -2wz-2z^2=-2wz.\]
The second equation forces $z=0$, which is incompatible with
the first one. There is no critical point in that case. In Figure~%
\ref{fig:ex-domain} (middle), this is reflected by the fact that no
level curve of $|wz|$ is tangent to the boundary of the domain of
convergence.
\end{example}
\begin{example}\label{example:no-critical-point2}Example~%
\ref{example:no-critical-point} is another example of a generating
function that
does not have any minimal critical point. Furthermore, it does not
have
critical points at all: the critical-point equations are
\[H=2+y-x(1+y)^2=0,\quad x+y-xy^2=0;\]
multiplying the first equation by $1-y$, the second one by~$
1+y$ and adding gives $2=0$, showing that this system does not have
any solution. Note however that the point~$(1/2,
\sqrt{3})$ is a local
maximizer of
$|xy|$ on~$\partial\mD$ where the critical-point equations are not
satisfied. 

This is due to a phenomenon shown in
Figure~\ref{fig:ex-domain} (right). Suppose $\bp_1=(1/2,\sqrt{3})$ and
$\bp_2=(1/2,-\sqrt{3})$
are the two points in $\mV$ with coordinate-wise moduli 
$(1/2,\sqrt{3})$, and let $\mN_1$ and $\mN_2$ be two small neighbourhoods of 
$\bp_1$ and $\bp_2$ in $\mV$. Then $(1/2,\sqrt{3})$ is not a local 
maximizer of $|xy|$ on the images of 
$\mN_1$ and $\mN_2$ under $(x,y) \mapsto (|x|,|y|)$, it only becomes
a local maximizer on $\partial\mD$ because the images of $\mN_1$ and
$\mN_2$ intersect after taking coordinate-wise moduli (Baryshnikov and 
Pemantle call this a `ghost intersection'). 
In particular, the level surface of $\Abs(x,y)$ is not tangent
to~$\mV$ at either of $\bp_1$ and $\bp_2$.s
\end{example}

The critical-point equations~\eqref{eq:critpt} form a system of $n$ equations in
$n$ unknowns. They are our starting point in this
work, where we focus
on the case when they have finitely many solutions, among which the
minimal points have to be found. Slightly more general cases can be
handled by small variations. For instance, when $H$ is not
square-free, its gradient vanishes, but one can
recover the relevant set of critical points by replacing
$H$ with its square-free part in Equations~\eqref{eq:critpt}. More
involved geometries of~$\mathcal{V}$ become difficult to analyze by
elementary means. This is what led Pemantle and Wilson~\cite{PemantleWilson2013} to
develop
their work in the language of Morse theory. 

\subsection{Asymptotic Analysis}\label{sec:acsv-local}
We first illustrate the main steps of the derivation on a simple
example.

\begin{example}\label{ex:central-binomial}
The power series $F=1/(1-x-y)$ has for diagonal 
\[\Delta\left(\frac1{1-x-y}\right)=\sum_{k\ge0}{\binom{2k}{k}t^k}
=\frac1{\sqrt{1-4t}}.\]
The asymptotic behavior of the diagonal coefficients is easily
seen to be $4^k/\sqrt{k\pi}$, e.g., by Stirling's formula. The
derivation of this result by ACSV starts with the
integral representation
\[a_k
= \frac1{2\pi i}\int_{|x|=r}\left(\frac{1}{2\pi i}\int_{|y|=r}{\frac1{1-x-y}\frac{dy}{(xy)^{k+1}}}\right)dx\]
for any $0<r<1/2$. 
For a fixed $x$ on the circle~$|x|=r$, the
integrand 
admits a unique pole, at $y=1-x$, outside of the initial circle of
integration.
Deforming the contour as indicated in Figure~%
\ref{fig:ex} (middle)
 shows that the integral with respect to~$y$ is
the sum of an integral over a contour~$|y|=1/(3r)>1/2$, and the
opposite of the residue at~$y=1-x$, namely
$1/(x(1-x))^{k+1}$. As $k$ increases, the factor~$(xy)^{-k-1}$ in the
integral over the large circle makes it grow exponentially
like~$(|xy|)^{-k}=3^k$.
The coefficient~$a_k$ thus behaves asymptotically like
\[a_k=\frac1{2\pi i}\oint_{|x|=r}{\frac{dx}{(x(1-x))^{k+1}}}+O(
c^k),\qquad c<4.\]

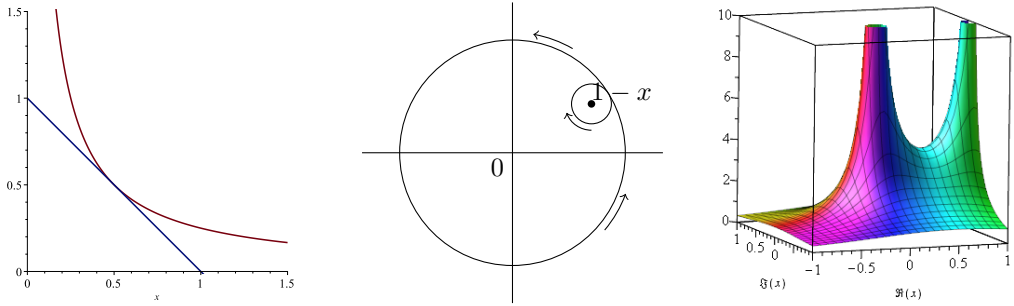
\begin{figure}
\centerline{
\includegraphics[height=4cm]{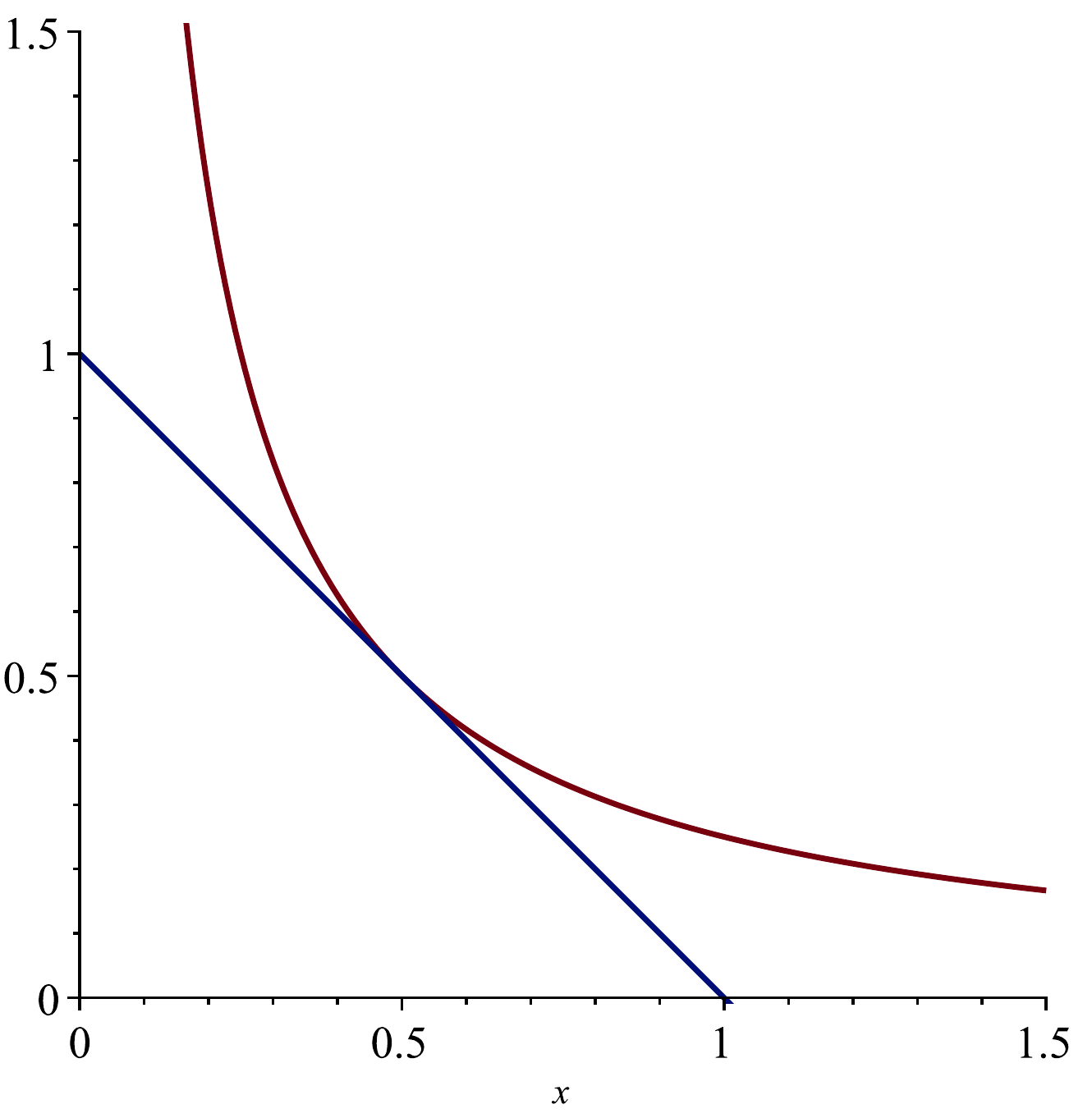}
\qquad
\begin{tikzpicture}
\draw (-2,0) -- (2,0);
\draw (0,-2) -- (0,2);
\draw (0,0) circle [radius=1.5];
\draw (1.05,0.65) circle [radius=0.265];
\fill (1.05,0.65) circle [radius=0.05];
\draw (-0.2,-0.2) node {0};
\draw (1.45,.8) node {$1-x$};
\draw[->] ([shift=(-40:1.6)]0,0) arc (-40:-20:1.6);
\draw[->] ([shift=(60:1.6)]0,0) arc (60:80:1.6);
\draw[->] ([shift=(-90:0.35)]1.05,0.65) arc [start angle=-90, end
angle=-160,radius=0.35];
\end{tikzpicture}\qquad
\includegraphics[height=4cm]{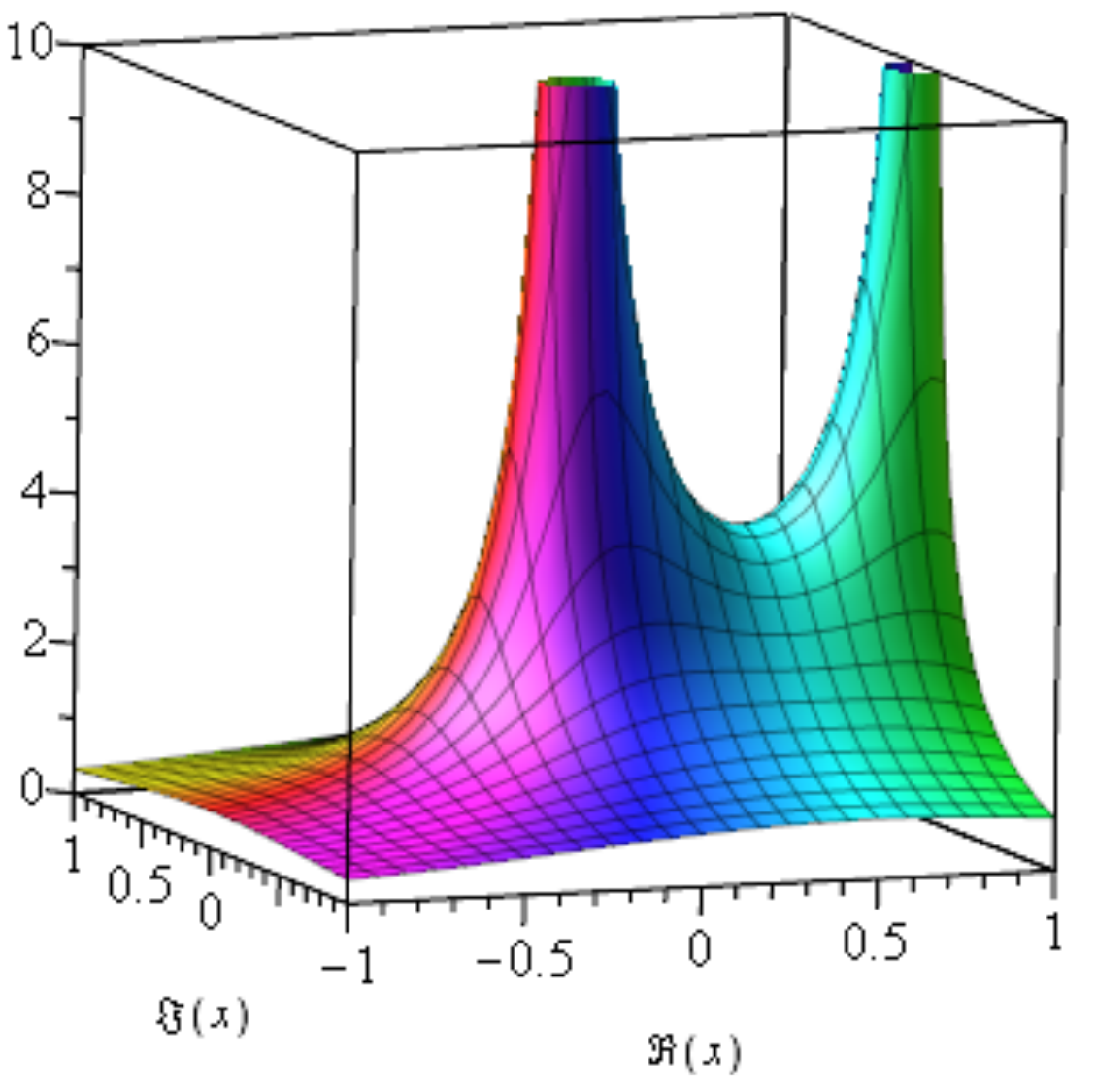}
}
\caption{Example~\ref{ex:central-binomial} in three steps: critical
points where the level curve of $|xy|$ is tangent to the singular
variety (left); contour of integration in the $y$-plane (middle);
modulus of the integrand $1/(x(1-x))$ with a saddle point at $x=1/2$
(right).\label{fig:ex}}
\end{figure}

This last integrand has a saddle point in the complex plane
at~$x=1/2$ (Figure~\ref{fig:ex}, right), where the integral concentrates
asymptotically. The classical saddle-point method (see~%
Olver~\cite{Olver1974}) then consists in: deforming the contour so
that it passes through the saddle point in the direction of
the imaginary axis; changing the variable into~$x=1/2+it$ and
observing that the integrand behaves
locally as
\[(x(1-x))^{-k-1}= 4^{k+1}e^{-4(k+1)t^2}(1+O(t^3)),\quad
t\rightarrow0;\]
reducing the asymptotic behaviour 
to that of a Gaussian integral, thus recovering the expected~$4^
{k}/\sqrt{k\pi}$.
\end{example}
The saddle-point integral in Example~\ref{ex:central-binomial} arose
because of a minimal critical point at $(1/2,1/2)$. 
One aspect of the computation that is missing from this simple example
is the selection of those critical points that are minimal. In
the context of this work, this is
the most expensive step computationally. It is discussed in the
next sections.

The techniques used in this example generalize.
If $\bzeta\in\mathcal{V}$ is a smooth point, then as in the proof of
Proposition~\ref{prop:3.12}, we assume without loss of generality that $(\partial
H/\partial z_n)(\bzeta) \neq 0$ and introduce the implicit
function~$g(\htbz)$ such that $H(\htbz,g(\htbz))=0$  and $g$ is
locally one-to-one in the neighbourhood of~$\bzeta$. Next, we consider
\[\psi(\htbz)=z_1\dotsm z_{n-1}g(\htbz).\]
\begin{lemma}\label{lemma:critpointpsi}With these notations, the
critical point equations are
equivalent to~$\nabla\psi
(\hat\bzeta)=0$.
\end{lemma}
\begin{proof}For $i\in\{1,\dots,n-1\}$, differentiating~$H(\hat\bz,g
(\hat\bz))=0$ with respect to~$z_i$ yields
\[\frac{\partial g}{\partial z_i}=-\frac{\partial
H}{\partial z_i}\big/\frac{\partial
H}{\partial z_n},\]
which can be injected into
\[\frac1\psi\frac{\partial\psi}{\partial
z_i}=\frac1{z_i}+\frac{1}{g(\hat\bz)}\frac{\partial g}{\partial z_i},
\]
and the conclusion follows from~$\zeta_n=g(\hat\bzeta)$.
\end{proof}
Thus locally on~$\mV$ in the neighbourhood of a smooth critical
point~$\bzeta$, the function $\psi$ behaves like
\begin{equation} \psi(\hat\bz):=\zeta_1\dotsm\zeta_n+\frac12
(\hat\bz-\hat\bzeta)^t\cdot\mH(\bzeta)\cdot(\hat\bz-\hat\bzeta)+O\left(\left|\hat\bz-\hat\bzeta\right|^3\right), \label{eq:psi_Taylor} \end{equation}
where $\mH$ is the Hessian matrix of~$\psi$ (the $(n-1)\times(n-1)$
matrix whose entry $(i,j)$ is $\partial^2\psi/\partial z_i\partial z_j
$).
\begin{definition}
The critical point $\bzeta$ is called \emph{non-degenerate} when the
Hessian $\mH$ of $\psi$ is non-singular at $\htbz = \hat{\bzeta}$.
\end{definition}
\noindent In this favourable situation, the main result of ACSV is
the
following.
\begin{proposition}{Pemantle and Wilson~\cite[Theorem 9.2.7, Corollary 9.2.8]{PemantleWilson2013}}
\label{prop:diag_asm}
Suppose $F(\bz)$ has a smooth, strictly minimal, non-degenerate
critical point with non-zero coordinates at $\bzeta$.  Then the
diagonal coefficients satisfy
\begin{equation}
f_{k,\dots,k} = \bzeta^{-k}k^{\frac{1-n}{2}}
\left(\frac{(2\pi)^{(1-n)/2}}{\sqrt{(\bzeta^{3-n}/\zeta_n^2)|\mH
(\bzeta)|}}  \cdot  \frac{-G(\bzeta)}{\zeta_n\frac{\partial
H}{\partial_{z_n}}(\bzeta)} + O\left(\frac{1}{k}\right)\right),\qquad k\rightarrow\infty,
\label{eq:diag_asm}
\end{equation}
where $|\mH(\bzeta)|$ is the determinant of the Hessian of
$\psi$.
\end{proposition}
The branch of the square-root in Equation~\eqref{eq:diag_asm} is
determined by the saddle-point integral arising in the proof of Proposition~\ref{prop:diag_asm}.
\begin{proof}[Sketch of the proof following~\cite{PemantleWilson2002}]
The starting point is the Cauchy integral~\eqref{eq:mCIF} and the idea
is to first perform the
integration with respect to~$z_n$.
Initially, the domain of integration is the product of the 
polytorus~$T(\hat\bzeta)$ and the circle $|z_n|=|\zeta_n|-\epsilon$ for a small
positive $\epsilon$, so that for $\bz$ on the contour, the open
polydisk $D
(\bz)$ is included in the domain of convergence of the power
series~\eqref{eq:powerseries}.  
For $\hat\bz\in T(\hat\bzeta)$ bounded away from $\hat\bzeta$, the
radius of convergence of $F(\hat\bz,z_n)$ as a function of~$z_n$ is
larger
than~$|\zeta_n|$ as a consequence of the minimality of~$\bzeta$
and
thus the inner integral is bounded by $(|\zeta_n|+\delta)^{-k}$ for
some uniform $\delta>0$,
so that that part of the integral is asymptotically exponentially
smaller than
$|\zeta_1\dotsm\zeta_n|^{-k}$. 

In the remaining part~$T'\ni\bzeta$ of the domain, since
$\bzeta$ is smooth, one
can
use the
implicit function~$g$ defined above. In a sufficiently small
neighbourhood of $\zeta$, the implicit function theorem even
ensures that 
there exists a larger disk~$D'=D((1+\epsilon)\zeta_n)$ such that
the
only singularity of the inner integrand inside~$D'$ with respect to
the
variable~$z_n$ is the simple pole at~$g(\hat\bz)$. There, its
residue
is
\[ 
\operatorname{Res}\left(\left.\frac{G(\bz)}{H(\bz)(z_1\cdots
z_n)^
{k+1}}\right| z_n= g(\htbz)\right) 
= \frac{G(\htbz,g(\htbz))}{\frac{\partial H}{\partial z_n}(\htbz,g
(\htbz))} \frac{1}{\psi(\htbz)^{k+1}}.
\]
By Cauchy's residue theorem, the integral over 
$T'\times D(|\zeta_n|-\epsilon)$ is thus equal to
the integral over~$T'\times D(|\zeta_n|+\epsilon)$ minus this residue
multiplied by $2\pi i$. 
The integral over~$T'\times D(|\zeta_n|+\epsilon)$
decreases asymptotically like~$|\bzeta|^{-k}(1+\epsilon)^{-k}$ as
$k\rightarrow\infty$, so that 
\begin{equation}\label{eq:proof-prop-14}
f_{k,\dots,k}=\frac{1}{(2\pi i)^{n-1}}\int_{T'}{\frac{-G(\htbz,g(\htbz))}{
\frac{\partial H}{\partial z_n}(\htbz,g
(\htbz))} \frac{dz_1\dotsm dz_{n-1}}{\psi(\htbz)^{k+1}}}+O\left(
(|\bzeta|
(1+\epsilon))^{-k}\right).
\end{equation}
Now, this integrand has a saddle point at~$\hat\bzeta$, in the
neighbourhood of which the integral concentrates asymptotically.
There, the Taylor expansion of the integrand is
\[\frac{-G(\bzeta)}{\frac{\partial H}{\partial z_n}(\bzeta)\psi(\bzeta)^{k+1}}\exp\left(-\frac{k+1}{2\psi(\bzeta)}(\hat\bz-\hat\bzeta)^t\cdot\mH(\bzeta)\cdot(\hat\bz-\hat\bzeta) + O\left(\left|\hat\bz-\hat\bzeta\right|^3\right)\right).\]
Since $\mH$ is non-singular, the integrand behaves locally
like a Gaussian integral and saddle-point methods can be applied to
obtain asymptotics (see Wong~\cite{Wong1989}).
\end{proof}

All the proofs up to this point reduce to deformations of univariate
integrals. A genuinely multivariate deformation of the
contour makes it possible to avoid $\mV$ while extending the domain of
integration beyond the minimal points that are not critical, 
leading to the following\footnote{The results of Baryshnikov and Pemantle~\cite{BaryshnikovPemantle2011}
include as a hypothesis that all minimizers of $|z_1\cdots z_n|^{-1}$ on $\partial\mD$
have the same coordinate-wise moduli, however the methods of that paper never
use this property under our conditions.}.

\begin{proposition}[Baryshnikov and Pemantle~\cite{BaryshnikovPemantle2011}]
\label{prop:add-crit-pts}
If the point $\bzeta$ of the previous proposition is not necessarily
\emph{strictly} minimal but $T(\bzeta)$ contains only a finite number
of critical points, all of them being smooth and non-degenerate, then 
asymptotics of the diagonal coefficients of $F(\bz)$ are obtained by
summing up the contributions~\eqref{eq:diag_asm} given by each
of these points.
\end{proposition}

Note that one can apply Proposition~\ref{prop:diag_asm} using any
coordinate $z_k$ such that $(\partial H/\partial z_k)(\bzeta) \neq 0$,
and in the context of Proposition~\ref{prop:add-crit-pts}, this coordinate may change
depending on the minimal critical point under consideration.

When the numerator $G(\bz)$ is 0 at the strictly minimal critical
point $\bzeta$ then Equation~\eqref{eq:diag_asm} gives only an order
bound on the asymptotics of the diagonal sequence.  Generically the numerator does not vanish at the critical points of $F$, however when this does happen one can typically determine dominant asymptotics by computing further terms of the Taylor expansion for $G(\bz)$ in a neighbourhood of $\bzeta$.

The Hessian $\mH(\bzeta)$ in Equation~\eqref{eq:diag_asm} can be
expressed in terms of that of~$H$ itself.  For a critical
point $\bzeta$, define $\lambda$ to be the common value of $\zeta_k
(\partial H/\partial z_k)(\bzeta)$ ($1 \leq k \leq n$) and for $1\leq
k,\ell \leq n$, set
\begin{equation}\label{eq:Ukl}
U_{k,\ell}:= \zeta_k\zeta_\ell\frac{\partial^2H}{\partial z_k\partial z_\ell}(\bzeta).
\end{equation}
Basic multivariate calculus shows that the $(n-1)\times(n-1)$ Hessian matrix $\mH$ at $\bzeta$ has $(i,j)^\text{th}$ entry
\begin{equation} 
\label{eq:Hess}
\mH_{i,j} = 
\begin{cases}
\frac{\zeta_1 \cdots \zeta_n}{\lambda\zeta_i\zeta_j}(U_{i,n}+U_
{j,n}-U_{i,j}-U_{n,n}-\lambda) &\text{if $i \neq j$,} \\[+2mm]
\frac{\zeta_1 \cdots \zeta_n}{\lambda\zeta_i^2}(2U_{i,n}-U_{i,i}-U_
{n,n}-2\lambda) &\text{if $i=j$.}
\end{cases}
\end{equation}
This makes it simple to compute the asymptotic contribution~\eqref{eq:diag_asm} at a non-degenerate minimal critical point $\bzeta$.

\subsection{Combinatorial Case}
\noindent The determination of the minimal points among the critical
points is
significantly easier when further positivity conditions hold.
\begin{definition}A
rational function $F(\bz) = {G(\bz)}/{H(\bz)}$ with $H(\mathbf{0})\neq0$ is
called \emph{combinatorial} when the coefficients of its Taylor
expansion at
the origin are all non-negative.
\end{definition}
In these conditions, the following result is a
multivariate variant,
that
actually extends to meromorphic functions, of a classical theorem in
the univariate case usually attributed to one of Pringsheim, Borel or
Vivanti, see~\cite{Hadamard1954,Vivanti1893}. 
(Pemantle and Wilson~\cite[Prop. 8.4.3]{PemantleWilson2013} give a stronger statement.)
\begin{lemma}
\label{lem:combCase}
If $F(\bz)$ is combinatorial and $\bw$
belongs to the boundary
$\partial\mD$ of its domain of convergence,
then the point $|(\bw)|=
(|w_1|,\dots,|w_n|)$ is minimal.
If moreover
$\bw$ is smooth and
critical, then $|(\bw)|$ is critical too. 
\end{lemma}
Note that the hypothesis can be weakened to allow
a finite number of negative coefficients in the Taylor series of
$F
(\bz)$, by subtracting the corresponding polynomial.

\begin{proof}Since $\bw$ belongs to
the boundary of the domain of convergence, the Taylor expansion of $F$
does not converge absolutely as~$\bz$ tends to~$\bw$ inside $D(\bw)$. Non-negativity of
the
coefficients then implies that the Taylor expansion of~$F$ does
not converge as~$\bz$ tends to~$(|w_1|,\dots,|w_n|)$ inside $D(\bw)$.
Since the function is meromorphic, this implies that it tends to
$\infty$ and that that point is a zero of~$H$. It is therefore both on the boundary
of the domain of convergence and on~$\mV$, as was to be proved.

If $|(\bw)|$ is not smooth, then it is critical. Otherwise, by
Proposition~\ref{prop:3.12}, there exists $\blambda=
(\lambda_1,\dots,\lambda_n)$ with non-negative real coordinates
 such that $(|w_1|\partial H/\partial z_1
(|(\bw)|),\dots,|w_n|\partial H/\partial z_n(|(\bw)|))$ and $
\blambda$ are colinear 
and $|(\bw)|$ is a local
maximizer
of
the map $\Abs_{\blambda}:\bz\mapsto|z_1^{\lambda_1}\dotsm z_n^
{\lambda_n}|$ on
$\boundary$.
It remains to show that $\blambda=(1,\dots,1)$. If
not, there exists $\bv\in\mathbb{R}^n$ such that $\bv\cdot(1,\dots,1)=0$
and $\bv\cdot\blambda>0$. Since $\bw$ is smooth and critical,  by
 Lemma~%
\ref{lemma:min-crit} it is a local maximizer of $\Abs$ on
$\boundary$. Note that
for small enough $\epsilon>0$, the point $(w_1e^{\epsilon
v_1},\dots,w_ne^{\epsilon v_n})$ belongs to $\overline{\mD}$. Then so
does
$(|w_1|e^{\epsilon
v_1},\dots,|w_n|e^{\epsilon v_n}),$ but this contradicts the local
maximality of $(|\bw|)$ for $\Abs_{\blambda}$.
\end{proof}
\noindent The first part of this result is the basis of the following test leading to an
efficient 
algorithm in the next section.
\begin{proposition}
\label{prop:lineMin}
If $F$ is combinatorial and $\bw$ is in~$\mV$, 
then $\bw$
is a minimal point if and only if the line segment
$\{(t|w_1|,\dots,t|w_n|) : 0 < t < 1\}$
does not intersect $\mV$. 
\end{proposition}
\begin{proof}
One direction is straightforward and does not depend on $F$ being
combinatorial: if $\bw$ is minimal, its polydisk $D
(\bw)$ is a subset of the domain
of convergence, so that it cannot contain any point~$
(t|w_1|,\dots,t|w_n|)\in\mV$ with $t \in (0,1)$ by Lemma~%
\ref{lemma:domconv}~\emph{(ii)}.

Conversely, let $\bw$ be in $\mV$, assume that the line segment of the
proposition does not intersect~$\mV$ and let
\[\mathcal{S}:=\{t\ge0\mid (t|w_1|,\dots,t|w_n|)\in\mD\}.\]
Since $H(\mathbf{0})\neq0$ the set~$\mathcal{S}$ is not empty. It is also included
in~$(0,1)$ since $\bw$ is a singularity of~$F$. Thus
$\theta:=\sup\mathcal{S}$ is well defined and finite. 
The point $\theta|\bw|:=(\theta|w_1|,\dots,\theta|w_n|)$ belongs
to~$\partial D$. This means that its torus intersects~$\mV$, by
Lemma~\ref{lemma:domconv} and by the
previous lemma that $\theta|\bw|$ itself belongs to~$\mV$. The
hypothesis then implies~$\theta=1$ and thus that $\bw$ is minimal.
\end{proof}

\begin{example}\label{ex:criteqApery}
The generating function of
Example~\ref{ex:Apery1},
whose diagonal coefficients are the Apery numbers, is
combinatorial. The
critical-point equations are
\begin{gather*}
H(a,b,c,z)=1-z(1+a)(1+b)(1+c)(1+b+c+bc+abc)=0,\\
\begin{align*}
z(1+a)(1+b)(1+c)(1+b+c+bc+abc)
&=az(1+b)(1+c)(1+b+c+2bc+2abc)\\
&=bz(1+a)(1+c)(2+2b+2c+ac+2bc+2abc)\\
&=cz(1+a)(1+b)(2+2b+2c+ab+2bc+2abc).
\end{align*}
\end{gather*}
They have only two solutions with $abc\neq0$:
$a=1\pm\sqrt2, b=\pm{\sqrt{2}}/2, c=\pm{\sqrt{2}}/
{2}, z=-82\pm58\sqrt2$.
Only the solution with $+\sqrt{2}$ has non-negative coordinates so that
it is the only one that is possibly minimal.
Adding the equation $H(ta,tb,tc,tz)=0$, eliminating the
variables $a,b,c,z$ and discarding the point $t=1$ produces the polynomial
\[t^{12}-2t^{11}+t^{10}+4t^9-24t^8-8t^7+20t^6-20t^5+212t^4-400t^3+820t^2-664t-4\]
with no root in the interval $(0,1)$. This proves the minimality of
that solution.
\end{example}
\begin{proposition}If $F$ is combinatorial and $\bw$ is a smooth
minimal
point with positive coordinates, then every point in a neighbourhood
of $\bw$ in $\mV\cap\mathbb{R}^n$ is minimal.
\end{proposition}
\begin{proof}
By Proposition~\ref{prop:lambda}, there exists
non-negative real $(\lambda_1,\dots,\lambda_n)$, not all zero, such
that $(\lambda_1,\dots,\lambda_n)$ is colinear with $
(w_1\partial H/\partial z_1(\bw),\dots,w_n\partial H/\partial z_n
(\bw))$. 

Assume, towards a contradiction, that there exists a sequence of
non-minimal points~$(\bx^{
(k)})$ converging to $\bw$ in~$\mV\cap{\mathbb{R}}^n$. Then, by a
generalization of
Proposition~\ref{prop:lineMin}, with the same proof, there
exists a sequence $(t_k)$ in $(0,1)$ such
that $(t_k^{\lambda_1} x^{(k)}_1,\dots,t_k^{\lambda_n}x^{(k)}_n)$
belongs
to~$\mV\cap\partial\mD$ and, since $\bw$
is minimal, $t_k$ tends to~1. 

Now, consider the system 
\[H(\bz)=H(t^{\lambda_1}z_1,\dots,t^{\lambda_n}z_n)=0\]
in the neighbourhood of its solution $(\bw,1)$.
Since $\bw$ is smooth, $\partial H/\partial z
(\bw)\neq0$ and without loss of generality we can assume that $
(\partial H/\partial z_n)(\bw)\neq0$.
Thus there exists an analytic
function $g(\hat\bz)$ such
that $(z_1,\dots,z_{n-1},g(\hat\bz))$ is a parameterization of $\mV$
in a neighbourhood of $\bw$ and $g$ is locally one-to-one.
Similarly, the derivative of the second equation
with respect to~$t$ at $(\bw,1)$, namely
\[\lambda_1w_1\partial H/\partial z_1(\bw)+\dots+\lambda_nw_n\partial
H/\partial z_n(\bw),\]
is a nonzero multiple of $\lambda_1^2+\dots+\lambda_n^2$ and therefore
nonzero itself, which shows the existence of an analytic function $T
(\hat\bz)$ such that $(\hat\bz,g(\hat\bz),T(\hat\bz))$ parameterizes a
solution of the
system in the neighbourhood of~$(\bw,1)$ and $T$ is locally
one-to-one. 

Thus for $k$ large enough, the system cannot be satisfied by both $
(\bx^{
(k)},1)$ and $(\bx^{(k)},t_k)$ with $t_k<1$, giving the desired
contradiction.
\end{proof}
\begin{corollary}If $F$ is combinatorial and $\bw$ is a smooth
minimal point with real positive coordinates that is a
local
maximizer of $\Abs$ on $\partial\mD$,
then it is critical.
\end{corollary}
\begin{proof}The previous proposition shows that a neighbourhood
of $\bw$ in $\mV\cap\mathbb{R}^n$ is included in~$\partial\mD$, so
that $\bw$ is a local maximizer of $\Abs$ in $\mV$, i.e., a
critical point.
\end{proof}
In some degenerate cases, there are several minimal critical points
with positive coordinates,
but these can be detected easily thanks to the following observation.
\begin{lemma}\label{lemma:finite-nb-min-crit} If $F(\bz)=G(\bz)/H
(\bz)$ is
combinatorial, if $\nabla H$ does not vanish on~$\mV\cap\boundary$ and
if $F$ admits
two
distinct minimal critical points with positive coordinates, then
it
admits an infinite number of
them.
\end{lemma}
\begin{proof} This is a consequence of the logarithmic convexity
of domains of
convergence. 

Let $\ba$ and $\bb$ be two such distinct points. By
Lemma~\ref{lemma:min-crit}, they are local maxima of $\Abs$
on~$\boundary$.
For
any $r\in(0,1)$, the point $\bc_r:=(a_1^rb_1^{1-r},\dots,a_n^rb_n^
{1-r})$ belongs to $\overline\mD$. 

Necessarily, $\operatorname{Abs}(\ba)=\operatorname{Abs}
(\bb)$. Indeed, if for
instance $\operatorname{Abs}(\bb)>
\operatorname{Abs}(\ba)$ then the points $\bc_r$ in a neighbourhood
of~$\ba$
would also satisfy~$\operatorname{Abs}(\bc_r)>\operatorname{Abs}(\ba)$,
which contradicts the fact that $\ba$ is a local maximizer of~$
\operatorname{Abs}$. 
Thus all $\bc_r$ for $r\in(0,1)$ satisfy $\operatorname{Abs}(\bc_r)=
\operatorname{Abs}(\ba)=\operatorname{Abs}(\bb)$. 
If $\bc_r$ lay inside $\mD$, then
so would a neighbourhood of~$\bc_r$. By the
maximum principle, this neighbourhood would contain a point giving a
larger value to~$\Abs$, but this is again a contradiction. This proves
that for all $r\in(0,1)$, $\bc_r$ belongs to~$\boundary$. It is minimal
by Lemma~\ref{lem:combCase} and a local
maximizer of $\operatorname{Abs}$ on~$\boundary$. The conclusion
follows from the previous corollary.
\end{proof}

\section{Overview of Algorithms for ACSV}\label{sec:algo-overview}

We now give a high-level overview of the main algebraic calculations
that must be performed, together with the assumptions that we 
make. These algorithms will be revisited in more detail in
Section~\ref{sec:MainAlgos}, after the tools that we use for the
required
decisions are introduced in Sections~\ref{sec:Algorithms}
and~\ref{sec:numkro}.

\begin{algoboxed}[\textsf{ACSV}]\leavevmode
\begin{enumerate}
\item Determine the set $\mathcal{C}$ of critical points, given
as zeros of the
polynomial system
\begin{equation}\label{eq:critsys}
 \operatorname{Crit} = \left(H,  z_1\frac{\partial H}{\partial
 z_1}-\lambda,
	\dots,z_n\frac{\partial H}{\partial z_n}-\lambda\right)
\end{equation}
in the variables $\bz,\lambda$. If $\mathcal C$ is not finite, FAIL.
\item Construct the subset~$\mathcal{U}\subset\mathcal{C}$ of 
\emph{minimal} critical
points.
\item If $G$ vanishes at all the elements of~$\mathcal{U}$ or if the
matrix $\mH$ defined by Equation~\eqref{eq:Hess} is singular at an element $\bzeta\in\mathcal{U}$, FAIL.\\
Otherwise, return the sum of the asymptotic contributions determined
by
 Equation~%
\eqref{eq:diag_asm} at all the elements of~$\mathcal{U}$.
\end{enumerate}
\end{algoboxed}

\subsection{Algorithm for Minimal Critical Points in the Combinatorial
Case}
\label{sec:CombCaseResults}
The difficult part of the computation in Algorithm~\textsf{ACSV} is
Step~2, where the minimal
critical points are computed. For that step, we start with the case
when $F(\bz)$ is combinatorial, where minimality is easier to prove in light of Proposition~\ref{prop:lineMin}.

\begin{algoboxed}[\textsf{Minimal Critical Points in the
Combinatorial Case}]\leavevmode
\begin{enumerate}
\item Determine the set $\mathcal{S}$ of zeros of the
polynomial system
\begin{equation}\label{eq:extended-sys}
 \bbf = \left(H,  z_1\frac{\partial H}{\partial z_1}-\lambda,
	\dots,z_n\frac{\partial H}{\partial z_n}-\lambda, H	
	(tz_1,\dots,tz_n)\right)
\end{equation}
in the variables $\bz,\lambda,t$. If $\mathcal S$ is not finite, FAIL.
\item Find $\bzeta\in\mathbb{R}_{>0}^{n}$ such that there exists $
(\bzeta,\lambda,t) \in \mathcal{S}$ and for all
such triples,
$t\not\in(0,1)$. If the
number of such $\bzeta$'s is not exactly 1 or if there are such
points with $\lambda=0$, FAIL.
\item Identify~$\bzeta$ among the elements of~$\mathcal{C}$ from
Equation~\eqref{eq:critsys}.
\item  Return $\displaystyle\mathcal{U}:=\left\{\bz\in
\mathbb{C}^n\mid\exists
(\bz,\lambda)\in\mathcal{C},
|z_1|=|\zeta_1|,\dots,|z_n|=|\zeta_n|\right\}$.
\end{enumerate}
\end{algoboxed}
\noindent Section~\ref{sec:MainAlgos} reviews in more detail how
these steps
can be carried out effectively and efficiently.
\begin{example}
Example~\ref{ex:criteqApery} shows the result of Steps~1 and~2 in the combinatorial
case algorithm for the rational function of Example~\ref{ex:Apery1} 
(Ap\'ery numbers), and finds a minimal critical point.
The only other
critical point is the one with the choice $-\sqrt{2}$ and it does not
belong to the polytorus of the first one, so that there is only one
contribution in the asymptotics. The Hessian with respect to~$a,b,c,z$
is computed by first evaluating the coefficients~$U_{k,l}$ from
Equation~\eqref{eq:Ukl}, giving $\lambda=-1$, $U_{1,4}=U_{2,4}=U_{3,4}=-1$, $U_{4,4}=0$
and
\[U_{1,1}=1-\sqrt{2},\quad
U_{1,2}=U_{1,3}=1-\frac32\sqrt{2},\quad
U_{2,2}=U_{3,3}=2(4-3\sqrt2),\quad
U_{2,3}=6-5\sqrt{2}.
\]
The matrix $\mathcal H$ follows from Equation~\eqref{eq:Hess}:
\[\mH=(239-169\sqrt{2})\begin{pmatrix}1&1&1\\
1&4(2+\sqrt{2})&2\\
1&2&4(2+\sqrt{2})
\end{pmatrix}.\]
Putting everything together, we have shown the asymptotic expansion stated in Example~\ref{ex:Apery1},
\[ A_k = \left(17+12\sqrt{2}\right)^k \, k^{-3/2} \pi^{-3/2} \, \frac{\sqrt{48+34\sqrt{2}}}{8}\left(1+O\left(\frac{1}{k}\right)\right). \]
\end{example}

The following assumptions give sufficient conditions for this
algorithm to work as expected:
\begin{itemize}
	\item[(A0)] $F(\bz)=G(\bz)/H(\bz)$ admits at least one minimal
	critical point;
	\item[(A1)] $\nabla H$ does not
	vanish at the minimal critical points;
	\item[(A2)] $G(\bz)$ is non-zero at at least one minimal critical
	point;
	\item[(A3)] all minimal critical points of $F(\bz)$ are non-degenerate;
	\item[(J1)] the Jacobian matrix of the system $\bbf$ in
	Equation~\eqref{eq:extended-sys} of $n+2$
	equations
	with respect to the $n+2$ variables $(\bz,\lambda,t)$ is
	non-singular at
	its solutions.
\end{itemize}
\begin{proposition}\label{prop:algo-combi}Let $F
(\bz)=G(\bz)/H(\bz)$ be a \emph{combinatorial} rational function.
With $G$ and $H$ as input, Algorithm
\textsf{ACSV} either fails or returns the
asymptotic behaviour of the diagonal coefficients of~$F$. The latter
occurs when the
assumptions (A0)--(A3)
and (J1) are satisfied. 
\end{proposition}
\begin{proof}The system $\bff$ is obtained from the critical-point
equations~\eqref{eq:critpt} by adding one equation and one
unknown~$t$. In particular, all points~$(\bzeta,\lambda,1)$ where
$\bzeta$
is a critical point are solutions of~$\bff$. Thus, when $\mathcal{S}$
is finite, there are finitely many critical points.

The function $F$ being combinatorial, Lemma~\ref{lem:combCase}
and Proposition~\ref{prop:lineMin} imply that if an element $\bzeta$
is selected at Step~2, then it is a minimal critical point. Moreover,
all minimal critical points are smooth since solutions
with $\lambda=0$ are excluded.

Thus the set $\mathcal{U}$ constructed at Step~3 contains all the
minimal critical points in the torus $T(\bzeta)$ and they are all
smooth. The final tests of Step~3 ensure that Proposition~%
\ref{prop:add-crit-pts} applies. 
 All
contributions of the minimal critical points are of the form
\[(\zeta_1\dotsm\zeta_n)^{-k}k^{\frac{1-n}{2}}(\alpha+O(1/k)),\]
with $\alpha$ a constant that depends on the actual critical point and
is not zero if $G$ does not vanish at the critical point.
This concludes the proof of the first
part of the Proposition.

When Assumption (J1) holds, the Jacobian matrix of the
system $\bff$ is invertible at its solutions. By the inverse function
theorem, this implies that
these solutions are isolated. The system
being formed of
polynomials, this implies that there are finitely many of them and
thus
that the set $\mathcal S$ computed in Step~1 of the algorithm is
finite. 

Next, since the function~$F$ is combinatorial and since (A0) implies
the existence of a minimal critical point that is smooth by (A1),
Lemma~\ref{lem:combCase} applied to such a point
implies the existence of a critical point in its torus that has
real positive
coordinates and will be selected in Step~2 of the algorithm.
Uniqueness of the solution is a consequence of Lemma~%
\ref{lemma:finite-nb-min-crit} and the fact that~$\bff$ has finitely
many solutions.

Finally, since $\mathcal{S}$ is finite, there are finitely
many critical points on the torus $T(\bzeta)$. By assumptions (A1) and
(A3), they are smooth
and non-degenerate, so that Step~3 does not fail.
\end{proof}
\subsection{Discussion of the Assumptions in the Combinatorial Case}
Pemantle and Wilson~\cite{PemantleWilson2013} almost always
assume (A0). They have results when there are no minimal critical
points but no explicit asymptotic formulas for such
cases. 
All their asymptotic results in dimension $n>2$ need (A3).
They also require isolated critical points, which we obtain as a
consequence of (J1).
Assumption (A2) is not a strong requirement: as mentioned
above, it is
possible to obtain an expansion by computing further terms of the
local behaviour of $G$ at the critical points.

Assumption (J1) is slightly stronger than needed for
ACSV, even in the smooth cases. Our main motivation for it is that it
lets us compute a
Kronecker representation of $\bbf$ with explicit control over the bit
complexity in Section~\ref{sec:Algorithms}. Moreover, 
it
is not as restrictive as it might seem, as we now show.

Recall that there are $m_d := \binom{d+n}{n}$ monic monomials in $\mathbb{C}[\bz]$ of total degree at most $d$.
\begin{definition}
A property $\mathcal{P}$ of \emph{polynomials} in $\mathbb{C}[\bz]$ is
said to hold \emph{generically} if for every positive integer $d$ there exists a proper algebraic subset $\mathcal{C}_d \subsetneq \mathbb{C}^{m_d}$ such that any polynomial of degree $d$ satisfies $\mathcal{P}$ unless its vector of coefficients lies in $\mathcal{C}_d$.

A property of \emph{rational functions} holds \emph{generically} if for
every
pair of positive integers $(d_1,d_2)$ there exists a proper algebraic
subset $\mathcal{C}_{d_1,d_2} \subsetneq \mathbb{C}^{m_{d_1}+m_{d_2}}$
such that any rational function with numerator and denominator of degrees $d_1$ and $d_2$ satisfies $\mathcal{P}$ unless the the vector defined by the coefficients of its numerator and denominator lies in $\mathcal{C}_{d_1,d_2}$.  
\end{definition}

This definition implies that the conjunction of finitely many generic properties is generic.  In Section~\ref{sec:generic} we prove the following result.

\begin{proposition}
\label{prop:generic_assumptions}
The assumptions (A1)--(A3) and (J1) hold generically.
Assuming $F(\bz)$ is combinatorial then (A0) holds generically. 
\end{proposition}

This result means that our algorithm is generically correct for
combinatorial rational functions. This does not mean that it solves
all problems arising in applications: many interesting combinatorial
examples do exhibit a non-generic behaviour.

Combinatoriality of generating functions is a property of a different
nature.
Unfortunately, deciding it is still open even in the univariate
case. The
closest result that is known is recent: Ouaknine and Worrell~\cite{OuaknineWorrell2014}
have shown the decidability of the ultimate positivity problem 
(determining whether the coefficients of a rational power series
expansion are eventually all non-negative) for \emph{univariate}
rational functions with square-free denominators.  In practice, then,
one usually applies these results when $F(\bz)$ is the multivariate
generating function of a combinatorial class with parameters, or when
the form of $F(\bz)$ makes combinatoriality easy to prove (for
instance, when $F(\bz) = {G(\bz)}/({1-J(\bz)})$ with $J(\bz)$ a
polynomial vanishing at the origin with non-negative coefficients).

\subsection{Algorithm for Minimal Critical Points in the
Non-Combinatorial Case}
\label{sec:GenResults}
Determining minimal critical points in the general case is more
involved, as we can no longer simply test the line segment between the
origin and a finite set of points to determine minimality.
We thus revert to Lemma~\ref{lemma:domconv}~\emph{(iii)} 
and set up a polynomial system that encodes the fact that a critical
point is minimal if and only if its polydisk does not intersect $\mV$.
A direct use of algorithms on the emptiness of
semi-algebraic sets would lead to too high a complexity. 
Instead, we exploit the
geometry of
the boundary of the domain of convergence to produce a system with as
many equations as unknowns, which is dealt with efficiently in the
next section.

Given a polynomial $f(\bz) \in \mathbb{C}[\bz]$ we define $f(\bx + i\by) := f(x_1 + iy_1, \dots, x_n + iy_n)$, and note the unique decomposition
\[ f(\bx+i\by) = f^{(R)}(\bx,\by) + if^{(I)}(\bx,\by),\] 
for polynomials $f^{(R)}(\bx,\by),f^{(I)}(\bx,\by)$ in $\mathbb{R}[\bx,\by]$.  The Cauchy-Riemann equations imply
\[ \frac{\partial f}{\partial z_j}(\bx + i \by) = \frac{1}{2} \cdot \frac{\partial}{\partial x_j}\left(f^{(R)}(\bx,\by) + if^{(I)}(\bx,\by)\right) 
- \frac{i}{2} \cdot \frac{\partial}{\partial y_j}\left(f^{(R)}(\bx,\by) + if^{(I)}(\bx,\by)\right), \]
and it follows that the set of \emph{real} solutions of the system
\begin{align}
H^{(R)}(\ba,\bb) = H^{(I)}(\ba,\bb) &= 0  \label{eq:GenSys1} \\
a_j \left(\partial H^{(R)}/\partial x_j\right)(\ba,\bb) + b_j \left(\partial H^{(R)}/\partial y_j\right)(\ba,\bb) - \lambda_R&=0, \qquad j = 1,\dots, n  \label{eq:GenSys2} \\
a_j \left(\partial H^{(I)}/\partial x_j\right)(\ba,\bb) + b_j \left(\partial H^{(I)}/\partial y_j\right)(\ba,\bb) - \lambda_I&=0, \qquad j = 1,\dots, n  \label{eq:GenSys3}
\end{align}
in the variables $\ba,\bb,\lambda_R,\lambda_I$ corresponds exactly to
the real and imaginary parts of all \emph{complex} solutions of the
critical point equations 
with $\bz = \ba+i\bb$ and $\lambda = \lambda_R + i \lambda_I$.

The minimality of $\bz$  is
then encoded using Lemma~\ref{lemma:domconv}~\emph{(iii)} by demanding that the equations
\begin{align}
H^{(R)}(\bx,\by) = H^{(I)}(\bx,\by) &= 0 \label{eq:GenSys4} \\
x_j^2 + y_j^2 - t(a_j^2+b_j^2) &= 0, \qquad j=1,\dots,n \label{eq:GenSys5}
\end{align}
do not have a solution with $\bx,\by,t$ real and $0<t < 1$.  

\paragraph{Critical points of the projection $\pi_t$}
At this stage, the system of equations~\eqref{eq:GenSys1}--%
\eqref{eq:GenSys5} consists of $3n+4$ equations in the $4n+3$ unknowns
$\ba,\bb,\bx,\by$, $\lambda_R,\lambda_I,t$. We denote by $\mW$ its
\emph{complex} solutions and our interest is in its \emph{real} part
$\mWR := \mW \cap
\mathbb{R}^{4n+3}$
 and even in $\mWRs := \mWR \cap (a_1^2+b_1^2\neq0)\cap\dots\cap
 (a_n^2+b_n^2\neq0)$, the points in $\mWR$ with
non-zero coordinates, to which Proposition~\ref{prop:diag_asm} applies.
Testing whether
$\mWRs\cap(
\mathbb{R}^
{4n+2}\times(0,1))$ is empty can be achieved by a direct use of
algorithms from effective real algebraic geometry. However, these
algorithms have a complexity that is generally higher than what can be
achieved by exploiting the geometry of this particular system.
Thus, instead of
considering all the values of $t$ where \eqref{eq:GenSys4}-%
\eqref{eq:GenSys5} have a solution, we consider the extremal values that $t$ takes at these
solutions.
This is computed by finding the critical points of the projection map
$\pi_t:\mWR \rightarrow \mathbb{R}$
defined by
$\pi_t(\ba,\bb,\bx,\by,\lambda_R,\lambda_I,t)=t$. When $H
(\bzer)\neq0$, we know that $0\not\in\pi_t(\mWR)$ while $1\in\pi_t
(\mWR)$ so that
minimality of $(\ba,\bb)$ is equivalent to the absence of a critical
value of $\pi_t$ in~$(0,1)$.
These critical points
are points of $\mWR$ where the differential of 
$\pi_t$ is rank deficient. Since Equations~%
\eqref{eq:GenSys1}--\eqref{eq:GenSys3} do not
depend on~$\bx,\by,t$, the system has a block structure that
can be exploited in the computation. We make the following simplifying
assumption:
\begin{itemize}
	\item[(J2)] the Jacobian matrix of the system \eqref{eq:GenSys1}--%
	\eqref{eq:GenSys5} has full rank at its solutions.
\end{itemize}
As a consequence, it is sufficient to
consider the points where the following matrix
\[
J =
\begin{pmatrix}
\nabla H^{(R)}(\bx,\by) \\
\nabla H^{(I)}(\bx,\by) \\
\nabla (x_1^2 + y_1^2 - t(a_1^2+b_1^2)) \\
\vdots \\
\nabla (x_n^2 + y_n^2 - t(a_n^2+b_n^2)) \\
\nabla (t)
\end{pmatrix}
=
\begin{pmatrix}
\frac{\partial H^{(R)}}{\partial x_1} & \cdots & \frac{\partial H^{
(R)}}{\partial x_n} & \frac{\partial H^{(R)}}{\partial y_1} & \cdots
& \frac{\partial H^{(R)}}{\partial y_n} & 0  \\
\frac{\partial H^{(I)}}{\partial x_1} & \cdots & \frac{\partial H^{
(I)}}{\partial x_n} & \frac{\partial H^{(I)}}{\partial y_1} & \cdots
& \frac{\partial H^{(I)}}{\partial y_n} & 0 \\
2x_1 & \bzer & 0 & 2y_1 & \bzer & 0 & -(a_1^2+b_1^2) \\
\bzer & \ddots & \bzer & \bzer & \ddots & \bzer & \vdots \\
0 & \bzer & 2x_n & 0 & \bzer & 2y_n & -(a_n^2+b_n^2) \\
0 & \cdots & 0 & 0 & \cdots & 0 & 1
\end{pmatrix}
\]
is rank deficient. This is
equivalent to the existence of a non-zero vector $
(\nu_1,\nu_2,\lambda_1,\dots,\lambda_n,\mu)$ in the left
kernel of the matrix~$J$.

Using the Cauchy-Riemann equations to write 
\[ \frac{\partial H^{(I)}}{\partial x_j} = -\frac{\partial H^{
(R)}}{\partial y_j} \quad \text{and} \quad \frac{\partial H^{
(I)}}{\partial y_j} = \frac{\partial H^{(R)}}{\partial x_j}\]
and extracting coordinates leads to the system
\[
\nu_1\frac{\partial H^{(R)}}{\partial x_j} - \nu_2\frac{\partial H^{
(R)}}{\partial y_j} + 2\lambda_j x_j = 0,\quad
\nu_1\frac{\partial H^{(R)}}{\partial y_j} + \nu_2\frac{\partial H^{
(R)}}{\partial x_j} + 2\lambda_j y_j = 0, \qquad j=1,\dots,n.
\]
Eliminating $\lambda_1,\dots,\lambda_n$ by linear combination, this
implies
\[(\nu_1 y_j-\nu_2 x_j)\frac{\partial H^{(R)}}{\partial x_j}
- (\nu_1 x_j+\nu_2y_j)
\frac{\partial H^{
(R)}}{\partial y_j}=0.
\]
Moreover, for each $j$ for which the $j$th coordinate of the
critical point under consideration is non-zero, $\lambda_j$ is
deduced from the previous set of equations.

Another consequence of Assumption (J2) is that $\nu_1\nu_2\neq0$. 
Introducing $\nu=\nu_2/\nu_1$ then simplifies the system further and
this discussion results in the following effective criterion for
minimality.
\begin{proposition}\label{prop:crit-pts-proj}
Let $H\in\mathbb{Q}[\bz]$ be a polynomial that does not vanish at the
origin and $\mV=\{\bz\in\mathbb{C}^n\mid H(\bz)=0\}$. Under Assumption
(J2), the point $\bz = \ba + i\bb \in \left(\mathbb{C}^*\right)^n$ with
	$\ba,\bb \in \mathbb{R}^n$ is a minimal critical point if
	and only if there exists $(\lambda_R,\lambda_I)\in\mathbb{R}^2$ such that $
(\ba,\bb,\lambda_R,\lambda_I)$ satisfies 
Equations~\eqref{eq:GenSys1}--\eqref{eq:GenSys3}
 and there
	does not exist $(\bx,\by,\nu,t) \in \mathbb{R}^{2n+2}$ with $0 < t
	< 1$ satisfying Equations~\eqref{eq:GenSys4}, \eqref{eq:GenSys5}
	and
\begin{equation}
	(y_j-\nu x_j)\left(\partial H^{(R)}/\partial x_j\right)(\bx,\by) - (x_j+\nu y_j)\left(\partial H^{(R)}/\partial y_j\right)(\bx,\by) =0, \quad j = 1,\dots, n. \label{eq:GenSys6}
\end{equation} 
\end{proposition}
Equations~\eqref{eq:GenSys1}--\eqref{eq:GenSys6} form a system of $4n+4$
equations in $4n+4$ unknowns, for which efficient algorithms are
available under mild assumptions, as discussed in the next section.

We can now state our algorithm in the non-combinatorial case.

\begin{algoboxed}[\textsf{Minimal Critical Points in
the Non-Combinatorial Case}]\leavevmode
\begin{enumerate}
\item Determine the set $\mathcal{S}$ of zeros of the
polynomial system~\eqref{eq:GenSys1}--\eqref{eq:GenSys6}
in the variables $\ba,\bb,\bx,\by,\lambda_R,\lambda_I,\nu,t$. If
$\mathcal S$ is not finite, FAIL.
\item Construct the set of minimal critical points
\[\mathcal{U}:=
\{\ba+i\bb\in(\mathbb{C^*})^{n}\mid\exists
(\ba,\bb,\bx,\by,\lambda_R,\lambda_I,\nu,t)\in\mathcal{S}\cap
\mathbb{R}^{4n+4}\text{
and for all such tuples, } t\not\in(0,1)\}.\] 
Return FAIL if either $\mathcal{U}$ is empty, or one of its elements
has $\lambda_R=\lambda_I=0$, or if the elements of $ \mathcal{U}$ do not all belong
to the same torus.
\item Identify the elements of~$\mathcal{U}$ within~$\mathcal{C}$
from Equation~\eqref{eq:critsys} and
return them.
\end{enumerate}
\end{algoboxed}
\noindent Again, we defer to Section~\ref{sec:MainAlgos} the details
of how
these steps can be carried out effectively and efficiently.

\begin{example}The rational function
\[F=\frac1{1-x(1-y)(1-z)(1-yz)}\]
is not combinatorial. Its diagonal coefficients~$(a_k)$ equal~0 for
odd~$k=2m+1$ and $(-1)^m(3m)!/m!^3$ for $k=2m$. Thus asymptotics
obtained by ACSV can be checked using Stirling's formula.

The critical point equations admit two solutions:
\[\lambda=-1,y_\pm=z_\pm,x_\pm=\frac{2+z_\pm}{9},1+z_\pm+z_\pm^2=0.\]
Correspondingly, the system formed with Equations~%
\eqref{eq:GenSys1}--\eqref{eq:GenSys3} admits the real solutions
corresponding to the real and imaginary parts of these two points, as
well as two other solutions with non-real coordinates.

The ideal generated by the polynomials in Equations~%
\eqref{eq:GenSys1}--\eqref{eq:GenSys6} contains a univariate
polynomial in~$t$, of degree~76, with only three real roots different
from~1, all larger than~1. Thus both critical points are minimal.

The Hessian matrices at these minimal critical points are the matrix
\[ -\left(1/9+i\sqrt{3}/27\right)\begin{pmatrix}1 & 1/2 \\ 1/2 & 1 \end{pmatrix} \]
and its conjugate, and adding the asymptotic contributions of each point gives
the dominant asymptotic term as
\[ \frac{\left(i 3\sqrt{3}\right)^k \sqrt{3}}{2 k \pi} + \frac{\left(-i 3\sqrt{3}\right)^k \sqrt{3}}{2 k \pi} 
= 
\begin{cases}
\frac{(-27)^m \sqrt{3}}{2 m \pi} &\text{if $k = 2m$ is even,} \\
0 &\text{otherwise,}
\end{cases}
\]
which matches what Stirling's formula provides.
\end{example}
\begin{proposition}\label{prop:correct-non-combi}
If the rational function $F(\bz)=G(\bz)/H(\bz)$ is
such that $H(\bzer)\neq0$, then with $G$ and $H$ as input
Algorithm~\textsf{ACSV} either fails or returns the 
asymptotic behaviour of the diagonal coefficients of~$F$.
The latter occurs when the assumptions
(A0)-(A3) and (J2) are satisfied.
\end{proposition}
\begin{proof} The proof is similar to that of Proposition~\ref{prop:algo-combi}.

The system~\eqref{eq:GenSys1}--\eqref{eq:GenSys3} has for real
solutions the real and imaginary parts of the solutions of the
critical-point equations. Thus, when $\mathcal{S}$ is finite, there
are finitely many critical points. By Proposition~%
\ref{prop:crit-pts-proj}, the set $\mathcal{U}$ computed in Step~2
contains exactly the minimal critical points with non-zero
coordinates. Moreover, all these points are smooth since solutions
with $\lambda=0$ are excluded. The final tests in Step~3 ensure that
Proposition~\ref{prop:add-crit-pts} applies. This finishes the proof
of the first part of the Proposition.

When Assumption~(J2) holds, the Jacobian matrix
of the system~\eqref{eq:GenSys1}--\eqref{eq:GenSys6} is invertible at
its solutions, which
implies that there are finitely many of them. Thus the set~$\mathcal S$
computed in Step~1 is finite. By assumption (A0), the set $
\mathcal{U}$ computed in Step~2 is not empty. Moreover, the previous
proposition shows that it contains exactly the minimal critical points
with non-zero coordinates.
Then
Assumptions (A1) and (A3) show that they are smooth and non-degenerate,
so that Proposition~\ref{prop:add-crit-pts} applies.
\end{proof}

At this stage, we leave the following to future work.
\begin{conjecture}
Assumption (J2) holds generically.
\end{conjecture}

\section{Kronecker Representation}
\label{sec:Algorithms}
We now consider the algorithms presented in the previous section from
the computer algebra perspective, making more explicit the statements
involving the determination of the set of zeros of a polynomial system
and the more delicate ones of selecting those zeros with specific
properties among them. We produce an estimate for the
bit complexity of these algorithms. First, we recall the complexity
model.

\subsection{Complexity Model}
The \emph{bit complexity} of an algorithm whose input is encoded
by integers (for instance, a multivariate polynomial over the
integers) is obtained by considering the binary representations of
these integers and counting the number of additions, subtractions, and
multiplications of bits performed by the algorithm. This is a 
complexity measure closer to the time complexity than the algebraic
complexity, where the operations over coefficients are counted at unit
cost. In particular, the sizes of the integers in intermediate
computations have to be taken into account in the analysis.
Our algorithms
typically take as input polynomials in $\mathbb{Z}[z_1,\dots,z_n]$,
and the bit complexity of the algorithms is expressed in terms
of the number of
variables $n$, the degrees and the heights of these polynomials. 
Here, the \emph{height} $h(P)$ of a polynomial $P\in\mathbb{Z}[\bz]$
is the maximum of 0 and the base~2 logarithms of the
absolute values of the coefficients of $P$. (Some care is
needed with the literature in this area, as some authors
define the height in terms of the maximum of the moduli
of the coefficients rather than their logarithms.) Unless otherwise
specified we assume that $d$ denotes an integer that is
at least 2 (typically corresponding to polynomial degree) and define $D := d^n$.

For two functions $f$ and $g$ defined and positive over~$(
\mathbb{N}^*)^m$, the notation $f(a_1,\dots,a_m) = O(g
(a_1,\dots,a_m))$ states the existence of a constant~$K$ such that $|f
(a_1,\dots,a_m)|\le Kg(a_1,\dots,a_m)$ over~$(\mathbb{N}^*)^m$. 
 Furthermore, we write $f = \tilde{O}(g)$ when $f=O(g\log^kg)$ for
 some $k\ge0$; for instance, $O(nD)=\tilde{O}(D)$ since $D=d^n$ and we
 assume $d \geqslant 2$. The dominant factor in the complexity of most operations we consider grows like $\tilde{O}(D^c)$ for some constant $c$, and our goal is typically to bound the exponent $c$ as tightly as possible.  
 
It is often convenient to consider a system of polynomials $\bbf = 
(f_1,\dots,f_n)$ of degree at most $d$ as given by a straight-line
program (a program using only assignments, constants, $+,-,$ and
$\times$) that evaluates the elements of $\bbf$ simultaneously at any
point~$\bz$ using at most $L$ arithmetic operations (see Section 4.1
of the book by Burgisser et al.~\cite{BurgisserClausenShokrollahi1997} for additional
details on this complexity model).  For instance, this can allow one
to take advantage of sparsity in the polynomial system. The quantity
$L$ is called the \emph{length} of the straight-line program. An upper
bound on $L$ is obtained by considering $n$ dense polynomials in $n$
variables, leading to $L=O\left(n\binom{n+d}{d}\right)=\tilde{O}(D)$. 
\begin{example}\label{ex:slpcombi}
Let $F(\bz)=G(\bz)/H(\bz)\in \mathbb{Q}(\bz)$ be a
rational function
with $\deg H=d$ and $L$ be the length of a straight-line program
evaluating the polynomial~$H$. By a
classical result of Baur and Strassen (see Burgisser et al.~\cite[Section 7.2]{BurgisserClausenShokrollahi1997}), 
it is possible to construct
a straight-line program
evaluating not only~$H$ but also all its partial derivatives
$(\partial H/\partial z_1,\dots,\partial H/\partial z_n)$, of length
less than $4L$. Then the system of Equations~\eqref{eq:extended-sys},
consisting of $n+2$ polynomials of degree $d$ in $n+2$ variables can
be evaluated in less than~$5L+2n+1=O(L+n)$ operations.
\end{example}
\begin{example}\label{ex:slpnoncombi}
Let $F(\bz)=G(\bz)/H(\bz)\in\mathbb{Q}(\bz)$ be a rational function
with $\deg H=d$ and $L$ be the length of a straight-line program
evaluating~$H$. The system~\eqref{eq:GenSys1}--\eqref{eq:GenSys6} of
$4n+4$ equations in $4n+4$ unknowns can be evaluated by a
straight-line program of length~$O(L+n)$. Indeed, starting from the
straight-line program evaluating~$H$, one constructs a straight-line
program evaluating the real and imaginary parts~$H^{(R)}$
and $H^{(I)}$ in $O(L)$ operations, replacing
each addition and multiplication by the corresponding operations on
real and imaginary parts. From there, a program of length less
than~$O(L)$ computing
simultaneously these polynomials and their gradients with respect to
the variables $x_j$ and $y_j$ is again obtained by the result of
Baur and Strassen. Thus the whole system is evaluated with
$O(L+n)$ operations.
\end{example}

\subsection{Kronecker Representation}
Our complexity estimates rely on the use of a Kronecker representation
for the solutions of zero-dimensional polynomial systems. 

\begin{definition}
Given a zero-dimensional (i.e., finite) algebraic set 
\[ \mV(\bbf) = \{\bz \mid f_1(\bz) = \cdots = f_n(\bz)=0\} \]
defined by the polynomial system $\bbf = (f_1,\dots,f_n) \in 
\mathbb{Z}[z_1,\dots,z_n]^n$, a \emph{Kronecker representation} $\left
[P(u),\bQ\right]$ of this set consists of an integer linear form
\[ u = \lambda_1 z_1 + \cdots + \lambda_n z_n \in \mathbb{Z}[\bz] \] 
that takes distinct values at the elements of $\mV(\bbf)$, a
square-free
polynomial $P \in \mathbb{Z}[u]$, and $Q_1,\dots,Q_n \in \mathbb{Z}[u]$ of degrees smaller than the degree of $P$ such that the elements of $\mV(\bbf)$ are given by projecting the solutions of the system
\begin{equation}
\label{eq:KronRep}
	P(u)=0 ,\qquad 
    \left\{\begin{array}{ll}
    P'(u)z_1 - Q_1(u) &=0 ,\\ 
    &\,\vdots \\
    P'(u)z_n - Q_n(u) &=0,\end{array}\right.
\end{equation}
onto the coordinates $z_1,\dots,z_n$.  
\end{definition}
The \emph{degree} of a Kronecker representation is the
degree of $P$, and its \emph{height} is
the maximum height of its polynomials $P,Q_1,\dots,Q_n$.  Kronecker
representations of zero-dimensional systems date back to work of 
Kronecker~\cite{Kronecker1882} and Macaulay~\cite{Macaulay1916} on polynomial system
solving. We refer to Castro et al.~\cite{CastroPardoHageleMorais2001}
for a detailed history and account of this approach to solving polynomial systems.
\begin{example}\label{ex:apery-kro1}
For the system of critical-point equations for
the Ap\'ery generating function from Example~\ref{ex:criteqApery},
the linear form $u=a$ takes distinct values on the roots and
leads to the following Kronecker representation of them:
\[u^2-2u-1=0,\quad\lambda=-1,\quad a=\frac{u+1}{u-1},\quad b=c=
\frac{1}{u-1},\quad z=-\frac{2(41u-99)}{u-1}.\]
\end{example}
An important observation is that $u$ being a linear form with integer
coefficients in the coordinates
$z_j$, since the polynomials in the Kronecker representation have
integer coefficients, a root of $P(u)$ is real if and only if every
coordinate $z_j$ in the corresponding solution is real. 

\subsection{Bounds and Complexity}
A probabilistic algorithm computing a Kronecker representation of the
solutions of $\bbf$ under mild regularity assumptions, and with a
good complexity, was given by Giusti et al.~\cite{GiustiLecerfSalvy2001}. In our
context, however, it is possible to take advantage of the
multi-homogeneous structure of the systems under study and obtain an algorithm with a
better complexity estimate. This relies on precise bounds on
the coefficient sizes of the polynomials $P$ and
the $Q_j$ appearing in the Kronecker representation in this situation.
Such bounds have
been provided recently by Safey El Din and Schost~\cite{Safey-El-DinSchost2018}, extending
earlier results of
Schost~\cite{Schost2001} thanks to new height bounds by D'Andrea et al.~\cite{DAndreaKrickSombra2013}.
They allow
us
to
determine the complexity of rigorously deciding several properties of
the solutions to the original polynomial system needed in Steps~2 and
3 of the Algorithms \textsf{Minimal Critical Points in the
Combinatorial Case} and
\textsf{Minimal Critical Points in the Non-Combinatorial Case}. 

Consider a polynomial $f(\bz) \in \mathbb{Z}[\bz]$ and let
$\bZ_1,\dots, \bZ_m$ be a partition of the variables $\bz$.  The
polynomial $f$ has \emph{multi-degree at most} $(v_1,\dots,v_m)$ if
the total degree $\deg_{\bZ_j}(f)$ of $f$ considered as a polynomial only in the variables of $\bZ_j$ is at most $v_j$, for each $j=1,\dots,m$.  

When $\bbf=(f_1,\dots,f_n)$ is a polynomial system where $f_j$ has
multi-degree at most $\bd_j\in\mathbb{N}^m$, and the block of
variables $\bZ_j$ contains $n_j$ elements, 
Safey El Din and Schost~\cite{Safey-El-DinSchost2018} give an upper bound $\sC_{\bn}(\bd)$ on
the degrees of
the polynomials appearing in a Kronecker representation of the
non-singular solutions of $\bbf$, where $\bd=(\bd_1,\dots,\bd_n)$.
Writing $\bd_j=\left(d_{j1},\dots,d_{jm}\right)$, this upper bound is
given by 
\begin{multline}\label{eq:Cnd}
\sC_{\bn}(\bd)=\text{sum of the non-zero coefficients of the
truncated power series}\\
(d_
{11}\theta_1+\dots+d_{1m}\theta_m)\dotsm(d_
{n1}\theta_1+\dots+d_{nm}\theta_m)\mod
\left(\theta_1^
{n_1+1},\dots,\theta_m^{n_m+1}\right).
\end{multline}

The heights of the polynomials in the Kronecker representation
are controlled in terms of a similar quantity $\sH_{\bn}
(\bETA,\bd)$ defined as follows. If $h(f)$ denotes the height of a
polynomial $f(\bz)
\in \mathbb{Z}[\bz]$, let
\begin{equation} \eta(f) := h(f) + \sum_{j=1}^m \log(1+n_j)\deg_{\bZ_j}(f). \label{eq:etaHeight} \end{equation}
Given $\bETA \in \mathbb{R}^n$ such that $\eta(f_j) \leq \eta_j$ for
each $j=1,\dots,n$, then
\begin{multline}\label{eq:Hnd}
\sH_{\bn}(\bETA,\bd)=\text{sum of
the non-zero coefficients of the
truncated power series}\\
(\eta_1\zeta+d_{11}\theta_1+\dots+d_{1m}\theta_m)\dotsm
(\eta_n\zeta+d_
{n1}\theta_1+\dots+d_{nm}\theta_m)\mod
\left(\zeta^2,\theta_1^
{n_1+1},\dots,\theta_m^{n_m+1}\right).
\end{multline}
\begin{example} This notation simplifies a lot in the important
special case when the variables are considered as a single block with
$\bd
= (d,\dots,d)$ and $\eta_j =\eta := h + d\log(1+n)$. Then $\sC_{n}(\bd) = d^n=D$ as it is the sum of
the coefficients in
\[ (d\theta_1)^n \mod \left(\theta_1^{n+1}\right). \]
Furthermore, $\sH_{n}(\bETA,\bd) = \tilde{O}(hd^{n-1} + D)$ as it is the sum of the coefficients in 
\[ (\eta \zeta + d\theta_1)^n \mod \left(\zeta^2,\theta_1^{n+1}\right). \]
\end{example}
\noindent Using these notations, Safey El Din and Schost~\cite{Safey-El-DinSchost2018} obtain
the
following.
\begin{lemma}[{\cite[Lemma~23 and top of p.~205]{Safey-El-DinSchost2018}}]
\label{lemma:Chow}
Let $\bbf\in\mathbb{Z}[z_1,\dots,z_n]^n$ be a polynomial system and
let $Z(\bbf)$ be the
solutions of $\bbf$ where the Jacobian matrix of $\bbf$ is invertible. 
Then, fixing a partition $\bZ$ of the variables such that $f_j$ has
multi-degree at most $\bd_j$ and height satisfying $\eta(f_j) \leq
\eta_j$ (with $\eta$ defined in Equation~\eqref{eq:etaHeight}), there
exists a non-zero $a\in\mathbb{Z}$ such that the product
\[a\prod_{\bx\in Z(\bbf)}{(T_0-x_1T_1-\dots-x_nT_n)}\]
is a polynomial in $\mathbb{Z}[T_0,\dots,T_n]$
(a primitive Chow form of~$Z(\bbf)$), of height bounded by
$\sH_{n}(\bETA,\bd)+2\log(n+2)\sC_{\bn}(\bd)$.
\end{lemma}

\noindent From there, they deduce the following result, which we state
here in a less general form, sufficient for our needs.
\begin{proposition}{Safey El Din and Schost~\cite[Theorem~1]{Safey-El-DinSchost2018}}
\label{prop:kronrep}
Let $\bbf\in\mathbb{Z}[z_1,\dots,z_n]^n$ be a polynomial system and
let $Z(\bbf)$ be the
solutions of $\bbf$ where the Jacobian matrix of $\bbf$ is invertible.
Given a partition $\bZ_1,\dots,\bZ_m$ of the variables such that $f_j$ has
multi-degree at most $\bd_j$ and height satisfying $\eta(f_j) \leq
\eta_j$, the set $Z(\bbf)$
admits a Kronecker representation of degree at most $\sC_{\bn}(\bd)$
and height $\tilde{O}\left( \sH_{\bn}(\bETA,\bd) + n\sC_{\bn}(\bd)
\right)$ (using the notation from Eqs.~\eqref{eq:Cnd},\eqref{eq:Hnd}).

If~$\bbf$ is given by a straight-line program $\Gamma$ of length $L$
that uses integer constants of height at most $h$,
the algorithm {\sf NonSingularSolutionsOverZ} of
Safey El Din and Schost~\cite{Safey-El-DinSchost2018}
takes $\Gamma$ and $(\bd_1,\dots,\bd_n)$ as input and
produces either
such a Kronecker representation of $Z(\bbf)$,
or a Kronecker representation of smaller degree,
or \textsc{FAIL}.
The first outcome occurs with probability at least 21/32.  In any
case, the algorithm has bit complexity
\[ \tilde{O}\left(\sC_{\bn}(\bd)\sH_{\bn}(\bETA,\bd)\left(L +
n\mu + n^2\right) n \left(n + \log h \right) \right), \quad
\text{where}\quad
\mu = \max_{1 \leq j \leq n} \left(\deg_{\bZ_1}(f_j) + \cdots +
\deg_{\bZ_m}(f_j)\right).\]  

In the case when the variables consist of a single block
and $\bbf$ is formed of polynomials of
degrees at most $d$ and heights at most $h$, this gives an algorithm with bit complexity $\tilde{O}(D^3+hD^2d^{n-1}) \subset \tilde{O}(hD^3)$, whose output consists of polynomials of degrees at most $D$ and heights in $\tilde{O}(D+hd^{n-1}) \subset \tilde{O}(hD)$.
\end{proposition}
Note that when the Jacobian of $\bbf$ is invertible at
each of its solutions then Proposition~\ref{prop:kronrep} gives a Kronecker representation of all solutions of $\bbf$. 

Repeating the algorithm $k$ times, and taking the output with highest
degree, allows one to obtain a Kronecker representation of $Z(\bbf)$
with probability $1-\left({11}/{32}\right)^k$ which can be made
as close to~1 as desired. The probabilistic aspects are mostly related
to the
choices of prime numbers and of the linear form~$u$ with integer
coefficients that should
not be too large. 
In practice, these  probabilistic aspects are
minor and the bound on the probability given in the proposition is
very pessimistic. We refer to Giusti et al.~\cite{GiustiLecerfSalvy2001},
Safey El Din and Schost~\cite{Safey-El-DinSchost2018} for a discussion of these questions.

\paragraph{Note on the quality of the bounds}
The bounds provided by Proposition~\ref{prop:kronrep} are the best
available to this day and the only ones known to take into account the
multi-homogeneous structure of the input. Unfortunately, they
are not always tight. In particular, if the input system is itself a
Kronecker representation of degree~$\mD$ and height~$\mH$, the
expansion of
\[(\mH\zeta+\mD\theta_u+\theta_{1:n})^n
(\mH\zeta+\mD\theta_u)\bmod(\zeta^2,\theta_{1:n}^{n+1},\theta_u^2)\]
gives a height bound of~$O(n\mH\mD)$, which is bigger than expected.
Thus in our algorithms below, we do not take a Kronecker
representation as input, but the original system, so as to take
advantage of 
the better bounds of Proposition~\ref{prop:kronrep} and Lemma~\ref{lemma:Chow}.

\subsection{Applications to ACSV}
\begin{corollary}\label{coro:kro-combi} Let $H(\bz)\in\mathbb{Z}[\bz]$
of degree~$d$ and
height~$h$ be
evaluated by a straight-line program of length~$L$ using integer
constants of height at most~$h$. Under Assumption~(J1), the
system of critical-point equations~\eqref{eq:critsys}
admits a Kronecker representation of degree at most~$nD$ and 
height~$\tilde{O}(D(d+h))$ that can be computed by a probabilistic
algorithm in~$\tilde{O}(D^3(d+h))$ bit
operations. 
Under the same assumptions, the
extended system of critical-point equations~\eqref{eq:extended-sys}
admits a Kronecker representation of degree at
most~$ndD$ and
height~$\tilde{O}(Dd(d+h))$, that can be computed
by a
probabilistic
algorithm in 
$\tilde{O}
(D^3d^2(d+h))$
bit operations.
\end{corollary}
\begin{proof}
As shown in the proof of Proposition~\ref{prop:algo-combi},
Assumption~(J1) implies that both systems are zero-dimensional.

We start with the extended system, noting that the other one is
similar.
This system
has $n+2$ equations of degree at
most $2d$ in $n+2$ variables. 
Example~\ref{ex:slpcombi} shows that
it is
evaluated by a straight-line program of length~$O(L+n)$.
We partition the variables into three blocks~$\bz,\lambda,t$. 

The
bound~$\sC_{\bn}(\bd)$ is obtained as the sum of the non-zero
coefficients of $\theta_z,\theta_\lambda,\theta_t$ in 
\[d\theta_z(d\theta_z+\theta_\lambda)^n(d\theta_z+d\theta_t)\bmod
(\theta_z^{n+1}\theta_\lambda^2\theta_t^2)=
d(d^n\theta_z^n+nd^{n-1}\theta_z^{n-1}\theta_\lambda)
(d\theta_z+d\theta_t)\bmod
\theta_z^{n}=
nd^{n+1}\theta_z^{n-1}\theta_\lambda\theta_t,
\]
leading to~$\sC_{\bn}(\bd)=nd^{n+1}.$
The computation for the height is similar. With
\[\eta_1=h+\log(n+1)d,\quad \eta_2=\dots=\eta_{n+1}=h+d+\log
(n+1)d+1,\quad\eta_{n+2}=h+2\log(n+1)d\]
and
\begin{multline*}
(\eta_1\zeta+d\theta_z)(\eta_2\zeta+d\theta_z+\theta_\lambda)^n
(\eta_{n+2}\zeta+d\theta_z+d\theta_t)\bmod
(\zeta^2\theta_z^{n+1}\theta_\lambda^2\theta_t^2)=\\
nd^n(\eta_1+(n-1)\eta_2)\theta_z^{n-1}\zeta\theta_t\theta_\lambda
+nd^{n+1}\theta_t\theta_\lambda\theta_z^n
+d^n\left((d\theta_t+n\theta_\lambda)\eta_1+n
(d\theta_t+(n-1)\theta_\lambda)\eta_2+n\eta_
{n+2}\theta_\lambda\right)\zeta\theta_z^n,
\end{multline*}
it follows that $\sH_{n}(\bETA,\bd)=O(Dn(d+n)(h+\log(n+1)d))=
\tilde{O}(dD(h+d)).$
The complexity then follows from injecting these quantities in the
previous proposition.

The bounds for the system formed by the critical point equations
only are derived as above with simpler
computations
\begin{gather*}
d\theta_z(d\theta_z+\theta_\lambda)^n\bmod
(\theta_z^{n+1}\theta_\lambda^2)=
d(d^n\theta_z^n+nd^{n-1}\theta_z^{n-1}\theta_\lambda)
\bmod
\theta_z^{n}=
nd^{n}\theta_z^{n-1}\theta_\lambda,
\\
(\eta_1\zeta+d\theta_z)(\eta_2\zeta+d\theta_z+\theta_\lambda)^n
\bmod
(\zeta^2\theta_z^{n+1}\theta_\lambda^2)=
(\eta_1+(n-1)\eta_2)d^{n-1}\zeta\theta_\lambda\theta_z^{n-1}
+(\zeta\eta_1+n\theta_\lambda+n\zeta\eta_2)d^n\theta_z^n,
\end{gather*}
leading to $nD$ for the degree and 
$\tilde{O}(D(h+d))$ for the height.
\end{proof}
\begin{example}
Our stated bounds reflect a perceptible growth of the sizes with the
number of variables in computations.

Starting from the same system as in Example~\ref{ex:apery-kro1}
and adding the variable~$t$ and the extra equation for the extended
system, using the linear form~$u=a+t$ that takes distinct values on
the roots leads to a Kronecker
representation with
\begin{multline*}
P(u)=u^{14}-18u^{13}+151u^{12}-788u^{11}+2878u^
{10}-7796u^9+16006u^8-24756u^7+27929u^6\\
-21546u^5+9851u^4-1104u^3-1616u^2+1000u-196,
\end{multline*}
the other polynomials having a very similar size. This polynomial
factors into two irreducible factors, one of which, $u^2-4u+2$,
corresponds to the critical points and the solutions with~$t=1$. (It
can also be recovered as the gcd of $P$ and $P'-Q_t$, where
$t=Q_t(u)/P'(u)$ parametrizes $t$ in the Kronecker representation.) Reducing the
Kronecker representation modulo this factor recovers a Kronecker
representation for the critical-point system of height similar to that
of Example~\ref{ex:apery-kro1}.

\end{example}
\begin{corollary}\label{kro-non-combi} Let $H(\bz)\in\mathbb{Z}[\bz]$
of degree~$d$ and
height~$h$ be
evaluated by a straight-line program of length~$L$ using integer
constants of height at most~$h$. Under Assumption~(J2), the
system~\eqref{eq:GenSys1}--\eqref{eq:GenSys6}
admits a Kronecker representation of degree at
most $2^{n-1}dn^4D^3$ and
height $\tilde{O}(h2^nd^3D^3)$, which can be computed by a
probabilistic
algorithm in 
$\tilde{O}(4^{n}hd^4D^7)$
 operations.
\end{corollary}
\begin{proof}
By assumption (J2), the system is zero-dimensional.
It has $4n+4$ equations of degree at most~$d+1$ in $4n+4$
unknowns. By Example~\ref{ex:slpnoncombi}, it can be evaluated by a
straight-line program of length $O(L+n)$. The variables are
partitioned into~6 blocks~$(\ba,\bb),
(\bx,\by),\lambda_R,\lambda_I,t,\nu$. A straightforward computation
gives
\[(d\theta_{a,b})^2(d\theta_{a,b}+\theta_{\lambda_R})^n
(d\theta_{a,b}+\theta_{\lambda_I})^n
(d\theta_{x,y})^2
(2\theta_{x,y}+2\theta_{a,b}+\theta_t)^n
(d\theta_{x,y}+\theta_\nu)^n
\bmod
(\theta_{x,y}^{2n+1}\theta_{a,b}^{2n+1}\theta_{\lambda_R}^2\theta_
{\lambda_I}^2\theta_t^2\theta_\nu^2)=2^{n-1}dn^4D^3
\]
leading to~$\sC_{\bn}(\bd)=2^{n-1}dn^4D^3.$ The computation for the
height is similar but more technical. It leads to
$\sH_{n}(\bETA,\bd)=O((h+\log(n+1)d)n^62^nd^2D^3)=\tilde{O}
(h2^nd^3D^3)$.
The complexity then follows from injecting these quantities in the
previous proposition.
\end{proof}

\subsection{Polynomial Values at Points of Kronecker Representations}
There are several situations in our computations where we need to
compute information concerning the values of another polynomial at the
roots of a polynomial system, either to isolate intersections or to
estimate signs. The following result gives useful bounds on degrees,
heights and complexity for these operations.

\begin{proposition}
\label{prop:KronReduce2}
Let $\bbf\in\mathbb{Z}[z_1,\dots,z_n]^n$ be a polynomial system that
satisfies the hypotheses of Proposition~%
\ref{prop:kronrep}, and let $q \in
\mathbb{Z}[\bz]$ have height $\eta$ and degree
$\delta_i$ in the block of variables $\bZ_i$ for each $i \in 
\{1,\dots,m\}$ and be evaluated by a straight-line program of
length~$\ell$ using constants of height at most~$\eta$. If 
$P(u)$ is the polynomial appearing in a Kronecker
representation with bounds given in Proposition~%
\ref{prop:kronrep} then
\begin{enumerate}
	\item  there
exists a parametrization~$P'(u)T-Q_q(u)$
of the values taken by~$q$ on~$Z(\bbf)$ with~$Q_q\in
\mathbb{Z}[u]$ a polynomial of degree at most~$\sC_\bn(\bd)$ and
height~$\tilde{O}(\sH_{\bn}(\bETA,\bd)
(\eta+\delta)+n\delta\sC_\bn(\bd))$
where~$\delta=\delta_1+\dots+\delta_m+1$.
The polynomial~$Q_q$ can be determined
in
\[\sL:=\tilde{O}(\sC_\bn(\bd)\sH_\bn(\bETA,\bd)\delta
(\delta+\eta)
(L+\ell+n(\eta+\delta)+n^2)n(n+\log h+\log\eta))
\]
bit operations.
	\item there exists a polynomial~$\Phi_q\in\mathbb{Z}[T]$ which
vanishes on the values taken by~$q$ at the elements of~$Z(\bbf)$, of
degree
at most~$\sC_\bn(\bd)$ and height
$\tilde{O}(\sC_\bn(\bd)(\sH_\bn(\bETA,\bd)(\eta+\delta)+n\delta\sC_\bn
(\bd)))$. It can be computed in~$\tilde{O}(\sC_\bn(\bd)^2(\sH_\bn
(\bETA,\bd)(\eta+\delta)+n\delta\sC_\bn
(\bd)))$ bit operations.

\item when $q$ has degree~1, better bounds hold: the height
of~$\Phi_q$ is $O(\sH_\bn(\bETA,\bd)+(2\log(n+2)+\eta)\sC_\bn(\bd))$; 
it can be computed in $O(\sC_\bn(\bd)(\sH_\bn(\bETA,\bd)+(2\log
(n+2)+\eta)\sC_\bn(\bd)))$ bit operations.
\end{enumerate}
In the case when the variables consist of a single block, and $q$ and the
elements of $\bbf$ have degrees at most $d$ and heights at most $h$,
then $Q_q$ has degree at most $D$ and height $\tilde{O}(D(h+d))$, and can
be computed in 
$\tilde{O}(D^2(D+h+d)(h+d)^2)$ bit operations; the polynomial~$\Phi_q$
has
degree at most~$D$ and height 
$\tilde{O}(D^2(h+d))$ and can be computed in 
$\tilde{O}(D^3(h+d))$
bit operations.

In the same conditions, if moreover the degree of~$q$ is~1, $\Phi_q$
has height~$\tilde{O}(Dh)$ and can be computed in~$\tilde{O}(D^2h)$
bit operations.
\end{proposition}

Our proof uses results of the Appendix, which collects properties
and bounds of univariate polynomials and their roots.

\begin{proof}
Adding the polynomial $T-q$ to a polynomial system $\bbf$ gives a new
polynomial system $\bbf'$ with the same number of solutions as $\bbf$,
and any separating linear form $u$ for the solutions of $\bbf$ is a
separating linear form for the solutions of $\bbf'$.  Thus, the degree
of a Kronecker representation of $\bbf'$ is at most the degree of a
Kronecker representation of $\bbf$, which is bounded by $\sC_{\bn}
(\bd)$.  The variables can be partitioned as before, with
one extra block~$T$. The bounds of Proposition~\ref{prop:kronrep}
lead us to consider the previous product multiplied by an extra factor,
i.e.,
\[
(\eta_1\zeta+d_{11}\theta_1+\dots+d_{1m}\theta_m)\dotsm
(\eta_n\zeta+d_
{n1}\theta_1+\dots+d_{nm}\theta_m)
(\eta\zeta+\theta_T+\delta_1\theta_1+\dots+\delta_m\theta_m)\mod
\left(\zeta^2,\theta_1^
{n_1+1},\dots,\theta_m^{n_m+1},\theta_T^2\right).\]
Then the sum of coefficients is bounded by the product of the previous
sum by the value of the last
factor at~$\mathbf{1}$.
Thus the height of~$Q_q$ is
bounded as announced and the complexity follows from
Proposition~\ref{prop:kronrep}.

The minimal polynomial $\Phi_q$ divides the resultant of the
polynomials $P'(u)T-Q_q(u)$ and $P(u)$ with respect to~$u$, so the stated
height and degree bounds on $\Phi_q$ follow from classical bounds
recalled in Lemma~\ref{lemma:resultant} and Lemma~\ref{lemma:height}.
From there, the computation can be obtained by modular methods, using
the fact that~$\Phi_q$ is the minimal polynomial of~$Q_q/P'\bmod P$
and the fast algorithm for this operation due to Kedlaya and Umans~\cite[\S8.4]{KedlayaUmans2011}.

When $q$ has degree~1, the height of~$\Phi_q$ is obtained by
evaluating the primitive Chow form from Lemma~\ref{lemma:Chow} at the
coefficients of~$q$, each monomial of degree at most~$\sC_\bn(\bd)$
contributing an extra $O(\sC_\bn(\bd)\eta)$.

The case of a single block is obtained by specializing these estimates
with the values~$\sC_\bn(\bd)=D$ and~$\sH_\bn(\bETA,\bd)=\tilde{O}
(hd^{n-1}+D)$.
\end{proof}

\section{Numerical Kronecker Representation}\label{sec:numkro}
In order to test minimality we must be able to isolate and argue about
individual elements of a zero-dimensional set. The Kronecker
representation allows us to reduce these questions to problems
involving only univariate polynomials, whose degrees and heights are
under control thanks to the results recalled in the previous section.
Our approach is semi-numerical: we determine a precision such that
questions about elements of the algebraic set can be answered exactly
by determining the zeros of $P(u)$ numerically to such precision. 
Using standard results on univariate polynomial root solving and root
bounds we obtain complexity estimates for basic operations on these
numerical representations. Our estimates always relate to 
absolute rather than relative precision. This is motivated by the
use of separation bounds from Lemma~\ref{lemma:roots}~(%
\ref{item:roots2}) to detect distinct roots and, when they are
real, order them.

\subsection{Definition and Complexity}
\begin{definition}A \emph{numerical Kronecker representation} $[P
(u),\bQ,\bU]$ of a zero-dimensional polynomial system is a Kronecker representation $[P(u),\bQ]$ of the system together with a sequence $\bU$ of isolating intervals for the real roots of the polynomial $P$ and/or isolating disks for the non-real roots of $P$. 
\end{definition}
The \emph{size of an interval} is its length, while
the \emph{size of a disk} is its radius.  In practice the elements of
$\bU$ are stored as approximate roots, whose accuracy is certified to
a specified precision. Most statements below take exactly the same
form for the case of disks or intervals. When a distinction is
necessary, we qualify the numerical
Kronecker
representation as real in the case of intervals and complex otherwise.
\begin{theorem}
\label{prop:numKron}
Suppose the zero-dimensional system $\bbf \subset 
\mathbb{Z}[z_1,\dots,z_n]$ is given by a Kronecker representation $[P
(u),\bQ]$ of degree $\mD$ and height $\mH$.  Given $[P(u),\bQ]$ and
$\kappa>0$, a numerical Kronecker representation $[P
(u),\bQ,\bU]$ with isolating regions in $\bU$ of size at most $2^
{-\kappa}$ can be computed in $\tilde{O}(\mD^3+
\mD^2\mH+\mD\kappa)$ bit
operations. Furthermore, approximations to the elements of $Z
(\bbf)$
whose coordinates are accurate to precision $2^{-\kappa}$ can be
determined in $\tilde{O}(
\mD^3+n(\mD^2\mH+\mD\kappa))$ bit
operations.
\end{theorem}
\begin{proof}The first part of the theorem follows from the complexity
estimates of modern
algorithms for finding numerical roots of polynomials, as recalled in
Lemma~\ref{lemma:fsolve}. 

The second statement of the theorem relies
on Lemma~\ref{lemma:univar} below.
Part~(i) of Lemma~\ref{lemma:univar}, applied to each of the~$Q_j$, gives a
larger complexity than announced in the theorem. The improved
complexity comes from observing that the cost~$\mD^3$ is related to
the high-precision computation of the roots of~$P$, which is only
performed once.
\end{proof}

\begin{lemma}\label{lemma:univar}
Let $P$ be a square-free polynomial in $\mathbb{Z}[u]$ of 
degree at most~$\mD$ and height at most $\mH$. Let also $Q\in
\mathbb{Z}[u]$
have degree at most~$\mD$ and height at most $\mH'$ and let
$\Phi_q=\operatorname{Res}_u(P'(u)T-Q(u),P(u))$ have height~$h
(\Phi_q)$ ($=\tilde{O}(\mD(\mH+\mH'))$). Then
\begin{enumerate}
	\item[(i)] the values of $R(u):=Q(u)/P'(u)$
	can be obtained to
	precision~$2^{-\kappa}$ at all roots of~$P$ in $\tilde{O}
(\mD^3+\mD^2\mH+\mD(\kappa+\mH'))$ bit operations;
	\item[(ii)] given $\tilde{O}(\mD(\mH+
	h(\Phi_q))+\mH')$ bits of the roots of~$P$ after the binary point,
	these roots can be grouped according to the
	distinct
	values they give to $R(u)$ in $\tilde{O}(\mD^2(\mH+
	h(\Phi_q))+\mD\mH')\subset\tilde{O}(\mD^3(\mH+\mH'))$ bit
	operations;
	\item[(iii)] given~$\tilde{O}(\mD(\mH+\mH'))$ bits of the roots
	of~$P$ after
	the binary point, one can decide which of these roots make~$R
	(u)=0$  and which of them  correspond to real roots
	making $R(u)>0$ (or $R(u)<1$) in $\tilde{O}(\mD^2(\mH+\mH'))$ bit
	operations.
\end{enumerate}
\end{lemma}
Note that although this lemma is stated with $\tilde{O}(\cdot)$ estimates
on precision, explicit bounds follow from our proof, allowing for exact algorithms.
\begin{proof}
The bound on $h(\Phi_q)$ is a direct consequence of Lemma~\ref{lemma:resultant}.

Fix a root $v \in \mathbb{C}$ of $P(u)=0$.  
Assume that we have computed approximations $q$ and $p$ to
$Q(v)$ and $P'(v)$ such that $|Q(v)-q|$ and $|P'(v)-p|$ are both
less than $2^{-a}$ for some natural number $a$.  Then the error
on~$Q(v)/P'(v)$ is bounded by
\[ \left|\frac{Q(v)}{P'(v)} - \frac{q}{p} \right| = \left|\frac{(Q
(v)-q)p+ q(p - P'(v))}{P'(v)p}\right| \leq \frac{2^{-a}}{|P'
(v)|} \left(1+
\frac{|q|}
{|p|}\right). \]
Lemma~\ref{lemma:roots} gives the lower bound $|P'(v)|\ge 2^
{-2\mD\mH-(5/2)\mD\log(\mD+1)}$ and the upper bound $|v|\le 2^\mH+1$.
It follows that for any $a>2\mD\mH+(5/2)\mD\log(\mD+1)$,
\[
|p| \geq |P'(v)| - 2^{-a} \geq 2^{-2\mD\mH-(5/2)\mD\log(\mD+1)-1},\qquad
|q|  \leq |Q(v)| + 2^{-a} \leq 2^{\mH'}(1+\cdots+|v|^
\mD) +2^
{-a} \leq
2^{\mH'}(2^\mH+1)^{\mD}(\mD+2).
\]
Thus,
\[ \left|\frac{Q(v)}{P'(v)} - \frac{q}{p} \right|
\le 2^{-a+2\mD\mH+(5/2)\mD\log(\mD+1)}
\left(1+2^{2\mD\mH+(5/2)\mD\log(\mD+1)+1}2^{\mH'}(2^{
\mH}+1)^{\mD}(\mD+2)
\right)=2^{-a+\mH'+\tilde{O}(\mD\mH)}. \]
This implies that a value of $a=\kappa+\mH'+\tilde{O}(\mD\mH)$ 
is sufficient to
evaluate $R(v)$ to precision~$2^{-\kappa}$. By Lemma~%
\ref{lemma:feval}, the
simultaneous
evaluation
of the values of~$Q$ and~$P'$ at this precision~$a$ can be achieved in
$\tilde{O}(\mD(\kappa+\mH'+\mD\mH))$ bit operations,
given~$O(\kappa+\mH'+\mD\mH)$ bits of the roots of~$P(u)$ after
the binary point, which can be computed in~$\tilde{O}
(\mD^3+\mD^2\mH+\mD(\kappa+\mH'))$ by
Lemma~\ref{lemma:fsolve}. This proves part~(i).

Lemma~\ref{lemma:roots}~(\ref{item:roots2}) shows that the distinct roots of
$\Phi_q$ are
at distance at least $2^{-a}$ with 
$a=\frac12(\mD+2)\log\mD+\mD(h(\Phi_q)+\frac12\log\mD)$.
Thus precision $a+1$ is sufficient to separate the distinct values.
Using part~(i) with $\kappa=a+1$ then proves part~(ii) of the lemma.

Finally, in order to evaluate the sign, it is sufficient to determine
the sign of both~$P'$ and~$Q$ at the roots of~$P$. By Lemma~%
\ref{lemma:roots}(\ref{item:roots3}),(\ref{item:roots4}), this can be done by
computing these values
with~$
\tilde{O}(\mD
(\mH+\mH'))$ bits after the binary point. Lemma~\ref{lemma:feval}
shows that is can be achieved in~$\tilde{O}(\mD^2
(\mH+\mH'))$ bit operations. The case $R(u)<1$ is obtained by
computing the signs of~$P'$ and $Q-P'$, which obey the same bounds,
this last polynomial being computed in~$O(\mD\mH)$ bit operations.
\end{proof}

\subsection{Polynomial Equalities and Inequalities from Numerical
Kronecker Representations}
We now give numerical analogues of the results in Proposition~%
\ref{prop:KronReduce2} concerning the values taken by a polynomial at
points defined by a Kronecker representation.

\begin{proposition}\label{prop:signq}
With the same hypotheses and
notation as in
Proposition~\ref{prop:KronReduce2}, let 
\begin{gather*}
\mD=\sC_\bn(\bd),\quad
\mH=\sH_\bn(\bETA,\bd)+n\sC_\bn(\bd),\quad
\mH'=\sH_\bn
(\bETA,\bd)(\eta+\delta)+n\delta\sC_\bn(\bd),\\
\mH_\Phi=\begin{cases}O(\sH_\bn
(\bETA,\bd)+\sC_\bn(\bd)(\eta+\log n)),\quad&\text{if $\delta=1$,}\\
\tilde{O}(\sC_\bn(\bd)(\sH_\bn(\bETA,\bd)(\eta+\delta)+n\delta\sC_\bn
(\bd))),&\text{otherwise.}\end{cases}
\end{gather*}
Then,
\begin{enumerate}
	\item[(i)] given~$\tilde{O}(\mD\mH+\mH'+\kappa)$ bits of the roots
	of~$P$ after the binary point, 
	the values of $q(\bz)$ at the elements of~$Z(\bbf)$
	can be obtained to precision~$2^{-\kappa}$ in $\tilde{O}
	(\mD(\mD\mH+\mH'+\kappa)+\sL)$ bit operations;

	\item[(ii)] given~$\tilde{O}(\mD(\mH+\mH_\Phi)+\mH')$ bits of
	the roots of $P$ after the binary point, these roots
	of $P$ can be grouped
	according to the
	distinct values they give to~$\displaystyle R(u):=q\left(\frac{Q_1
	(u)}{P'
	(u)},\dots,\frac{Q_n(u)}{P'(u)}\right)$
	in $\tilde{O}(\mD^2(\mH+\mH_\Phi)+\mD\mH'+\sL)$ bit
	operations;

	\item[(iii)]given~$\tilde{O}(\mD\mH')$ bits of the
	roots
	of~$P$ after the binary point, one can decide which of
	these
	roots make~$R(u)=0$ and which of them correspond to real
	roots making~$R(u)>0$ (or $R(u)<1)$ in~$\tilde{O}
	(\mD^2\mH'+\sL)$ bit
	operations;

	\item[(iv)]when $q$ is one of the coordinates, then the
	required number of bits of the roots of $P$ in the previous items
	can also be obtained within the stated bit complexities.
\end{enumerate}
In the case when the variables consist of a single block, and $q$ and
the elements of~$\bff$ have degrees at most~$d$ and heights at
most~$h$, then these decisions all take
$\tilde{O}(D^2(D+h+d)(d+h)^2+D\kappa)$ bit operations except for (ii)
when $\delta>1$, where the cost increases to $\tilde{O}(D^4
(h+d)^2+D\kappa)$
bit operations.
(Here as before, $D$ is $d^n$.)
\end{proposition}
\begin{proof}
Proposition~\ref{prop:KronReduce2} shows that there exists a
polynomial $Q_q$ such that the values of~$q(\bz)$ at the
points of~$Z(\bbf)$ are given by the
parametrization~$P'(u)T-Q_q(u)$, that this polynomial has degree at
most~$\mD$ and height~$\tilde{O}(\mH')$, and that it can be
computed in~$\sL$ bit operations. When~$q$ is a coordinate,
then~$Q_q$ is already part of the Kronecker representation and does
not need recomputation.
Combining
these bounds with Lemma~\ref{lemma:univar} gives the result.
\end{proof}
\subsection{Applications to ACSV}
\begin{corollary}\label{coro:tin01}
For the extended system of
critical-point
equations~\eqref{eq:extended-sys}, using a Kronecker
representation~$[P(u),\bQ]$ from Corollary~\ref{coro:kro-combi}, in
$\tilde{O}(D^3d^3(d+h))$
bit operations, one can:
(i) select the
roots of~$P$ corresponding to real solutions
with positive coordinates~$(z_1,\dots,z_n)$; (ii) group these roots by
the
distinct
values
they give to each of the coordinates~$z_i,\lambda,t$;
(iii) select those roots for which $t$ is exactly~1, or lies in
the interval~$
(0,1)$.
\end{corollary}
\begin{proof}
The values~$\mathcal{D}=\sC_\bn(\bd)=ndD$
and~$\mathcal{H}=\tilde{O}(Dd(d+h))$ are
provided by Corollary~\ref{coro:kro-combi}.
Proposition~\ref{prop:signq} is applied to each of the coordinates
$z_1,\dots,z_n,\lambda,t$, i.e., in cases where
$\delta=\eta=1$.
This
leads to $
\mathcal{H}'=\tilde{O}(Dd(d+h))$ and similarly for~$\mathcal{H}_\Phi$.
Using the value $L=\tilde{O}
(D)$ coming from Example~\ref{ex:slpcombi} leads to
$\sL=\tilde{O}(D^3d^2(d+h))$. The conclusion is then a direct
application of the proposition.
\end{proof}
\begin{corollary}\label{coro:tin01-complex}For the system~%
\eqref{eq:GenSys1}--\eqref{eq:GenSys6}, 
using a Kronecker representation~$[P(u),\bQ]$ from
Corollary~\ref{kro-non-combi}, in $\tilde{O}(2^{3n}D^9d^5h)$ bit
operations, one can: (i) group
the roots of~$P$ corresponding to
real solutions by
the
distinct values they give to each of the coordinates~$a_i$ and $b_i$; 
(ii) select those roots for which
$t$ is exactly~1, or lies in
the interval~$
(0,1)$.
\end{corollary}
\begin{proof}
The values $\sC_\bn(\bd)=2^{n-1}dn^4D^3$ and $\sH_\bn(\bETA,\bd)=\tilde{O}
(h2^nd^3D^3)$ are provided by Corollary~\ref{kro-non-combi}.
Proposition~\ref{prop:signq} is applied to the coordinates $a_1,\dots,b_n$ and~$t$,
i.e., in cases  where
$\delta=\eta=1$. This
leads to 
$\mD=2^{n-1}dn^4D^3$, $\mH'=\tilde{O}(h2^nd^3D^3)$, 
$\sL=\tilde{O}(2^{2n-1}D^7d^4h)$, whence a bit complexity of
$\tilde{O}(2^{3n}D^9d^5h).$
\end{proof}

\subsection{Grouping Roots by Modulus}
\label{sec:groupmod}
Grouping roots with the same modulus will turn out to be the most
costly operation in the combinatorial case.  Unlike the separation
bound given in Lemma~\ref{lemma:roots}(\ref{item:roots2}) between distinct complex
roots
of a polynomial, which has order~$2^{-\tilde{O}(hd)}$, the best
separation bound for the \emph{moduli} of roots that we know of has
order $2^{-\tilde{O}(hd^3)}$~\cite[Theorem 1]{GourdonSalvy1996}, and
computing the coordinates of a Kronecker representation to this
accuracy would be costly.  Fortunately, for the cases in which we
need to group roots of a polynomial by modulus it will always be the case that the modulus itself is a root of $P$.  In this situation, we have a better bound.

\begin{lemma}
\label{lemma:modsep}
For a polynomial $A\in\mathbb{Z}
[T]$ of degree $d\ge2$ and
height $h$, if $A(\alpha)=0$ and $A(\pm|\alpha|)\neq 0$, then 
\[\bigl|A(|\alpha|)A(-|\alpha|)\bigr|\ge 
(d+1)^{2(2h+\log(d+1))(1-d^2)}\left(2^{hd+2\log((2d)!)}(d+1)\right)^{-d}
= 2^{-\tilde{O}(hd^2)}
.\]
\end{lemma}
\begin{proof}
By Lemma~\ref{lemma:resultant}, the resultant $R(u)=\Res_T(A(T),T^dA
(u/T))$ has degree at most $d^2$ and height at most $2hd+2\log(
(2d)!)$.  This resultant vanishes at the products $\alpha\beta$ of
roots of $A$, and in particular at the square
$|\alpha|^2=\alpha\overline\alpha$.  By Lemma~\ref{lemma:height}, the
Graeffe polynomial
$G(T):=A(\sqrt{T})A(-\sqrt{T})$ has degree $d$, height
at most $2h+\log(d+1)$ and its positive real roots are the squares of
the real roots of~$A$. The conclusion
follows from Lemma~\ref{lemma:roots}~(\ref{item:roots3}) applied to~$Q=G$
and~$A=R$.
\end{proof}

\begin{corollary}
\label{cor:moduli}
Given $A(T)$ satisfying the same hypotheses as in Lemma~%
\ref{lemma:modsep},
the real positive roots $0<r_1\le\dots\le r_k$ of $A(T)$ and all roots of moduli exactly $r_1,\dots,r_k$ can be computed, with isolating regions of size~$2^{-\tilde{O}(hd^2)}$, in $\tilde{O}(hd^3)$ bit operations.
\end{corollary}
\begin{proof}
Let $G(T)$ be the polynomial in the proof of Lemma~\ref{lemma:modsep}, $b=\tilde{O}(hd^2)$ be the negative of the logarithm of the bound in Lemma~%
\ref{lemma:modsep},
and $\alpha$ be a root of $A$, so that if $\left|G(|\alpha|^2)\right|
< 2^{-b}$ then actually $\left|G(|\alpha|^2)\right|=0$ and at least one of $\pm|\alpha|$
is a root of $A$.  If we know an approximation $a$ to $\alpha$ such that
$|\alpha-a| < 2^{-(b+h+3)}$, then $|\overline{\alpha} - \overline{a}|
< 2^{-(b+h+3)}$ and using the bound $|\alpha| \leq 2^h + 1$ from
Lemma~\ref{lemma:roots}(\ref{item:roots1}) shows that
\[ \left| |\alpha|^2 - a \overline{a} \right| 
= \left| \alpha \overline{\alpha} + \alpha \overline{a} - \alpha \overline{a} - a \overline{a} \right| 
\leq |\alpha| 2^{-b-h-3} + |\overline{\alpha}|2^{-b-h-3} + 2^
{-2b-2h-6} \leq 2^{-b}.\]
The roots of $A(T)$ can be computed to precision~$2^{-(b+h+2)}$
in~$\tilde{O}(hd^3)$ bit operations by Lemma~\ref{lemma:fsolve}.
By Lemma~\ref{lemma:feval}, this accuracy is sufficient to evaluate
$G (|\alpha|^2)$ to accuracy $2^{-b}$ at all the roots
in~$\tilde{O}(hd^3)$ bit operations.

 Thus, knowing an approximation to $\alpha$ of accuracy $2^{-\tilde{O}
 (hd^2)}$ is sufficient to decide whether or not at least one of
 $\pm|\alpha|$ is a root of $A$, and to decide which real $\alpha$ are
 positive.  With these same roots and that same complexity, one can
 evaluate both~$A(|\alpha|)$ and $A(-|\alpha|)$ separately, to an accuracy $2^
 {-\tilde{O}(hd^2)}$. When only one of them is~0, Lemma~%
 \ref{lemma:roots}~(iii) applied to $Q(T)=A(T)$ or $Q(T)=A(-T)$ shows
 that the other one is at least~$2^{-\tilde{O}(hd)}$, which makes the
 decision possible.
\end{proof}

In practice, one would first compute roots only at precision $\tilde{O}(hd)$, in $\tilde{O}(hd^2)$ bit operations, and then check whether any of the non-real roots has a modulus that could equal one of the real positive roots in view of its isolating interval. Only those roots need to be refined to higher precision before invoking Lemma~\ref{lemma:modsep}. 

\begin{corollary}
\label{cor:equalmoduli}
With the hypotheses and notation of Proposition~%
\ref{prop:signq}, finding for all real roots~$u$ of~$P$ all the
elements~$(z_1,\dots,z_n)$ of~$Z(\bbf)$ such that $|z_i|=|Q_i
(u)|/|P'(u)|$ for $i=1,\dots,n$ can be performed in~$\tilde{O}
(\mH\mD^3)$ bit
operations.

In the case when the variables consist of a single block, and the
elements of~$\bbf$ have degrees at most~$d$ and heights at most~$h$,
then this decision takes~$\tilde{O}(hD^4)$ bit operations.
\end{corollary}
\begin{proof}This is obtained by computing the polynomials~$\Phi_q$
with $q=z_1,\dots,z_n$ from Proposition~\ref{prop:KronReduce2} and
applying the previous corollary.
\end{proof}
\begin{corollary}\label{coro:eqmodcombi}
For the system of
critical-point
equations~\eqref{eq:critsys},  finding all solutions $\mathbf{r}_1,\dots,\mathbf{r}_k$ with
positive real coordinates and determining all solutions with the same
coordinate-wise moduli as each $\mathbf{r}_j$ can be done in~$
\tilde{O}(D^4(d+h))$ bit operations.
\end{corollary}
\begin{proof}A direct use of the values provided by 
Corollary~\ref{coro:kro-combi}.
\end{proof}


\section{Algorithms for Effective Asymptotics}
\label{sec:MainAlgos}

The algebraic toolbox provided by the Kronecker
representation and its numerical extension can now be applied to
perform the decisions needed in our algorithms from Section~%
\ref{sec:algo-overview},
both in the
combinatorial and non-combinatorial cases. The latter one has a
higher complexity, with a difference of exponent only.

\begin{theorem}\label{thm:final-complexity}
Let $F(\bz)=G(\bz)/H(\bz)$ be a rational function with numerator and denominator in~$\mathbb{Z}
[z_1,\dots,z_n]$ of degrees at most~$d$
and heights at most~$h$. Let $D=d^n$ and $\sL=D^2(D+h)(d+h)^2d=O
(D^3d^3h^3)$.

If $F$ is combinatorial and Assumptions (A0)--(A3) and (J1) hold, then
Algorithm~\ref{alg:EffectiveCombAsm} is a probabilistic algorithm that computes
the minimal critical points of~$F$
in $\tilde{O}(D^4(d+h))$ bit operations.

When $F$ is not combinatorial but Assumptions (A0)--(A4) and (J2)
hold, then Algorithm~\ref{alg:MinimalCritical} is a probabilistic
algorithm that computes the minimal critical points of~$F$
in $\tilde{O}(2^{3n}D^9d^5h)$ bit operations.

In both cases, Algorithm~\ref{alg:ACSV} is a probabilistic
algorithm that uses these results and
$\tilde{O}(\sL)$ bit operations to compute three rational
functions $A,B,C \in \mathbb{Z}(u)$, a square-free polynomial $P \in \mathbb{Z}[u]$ and a list $U$ of roots of $P(u)$ (specified by isolating regions) such that
\begin{equation} f_{k,\dots,k} =
 (2\pi)^{(1-n)/2}\left(\sum_{u \in U} A(u)\sqrt{B(u)} \cdot C(u)^k \right)k^{(1-n)/2}
 \left(1 + O\left(\frac{1}{k}\right) \right). \label{eq:EffectiveCombAsm}
 \end{equation}
The values of $A(u),B(u)$ and $C(u)$ can be refined to precision $2^
{-\kappa}$ at all elements of $U$ in $\tilde{O}(\sL+D\kappa)$ bit
operations.
\end{theorem}
\begin{algorithm}[t]
\DontPrintSemicolon
\KwIn{Rational function $F(\bz) = G(\bz)/H(\bz)$ which satisfies 
one of the sets of assumptions of Theorem~\ref{thm:final-complexity}}
\KwOut{Polynomials $A,B,C,P$ and a set $\bU$ of roots of $P$ such
that the coefficients $f_{k,\dots,k}$ of $F$ behave
asymptotically as
in Equation~\eqref{eq:EffectiveCombAsm}, or \textsc{fail}}
\tcc{\qquad{\color{blue} Step 1: determine the set $
\mathcal{C}$ of critical points}}
$\operatorname{Crit}\gets \left\{H,  z_1(\partial H/\partial z_1)-\lambda, 
\dots ,z_n(\partial H/\partial z_n)-\lambda\right\}$\;
$[P,\bQ,u] \gets {\sf KroneckerRep}(\operatorname{Crit},
((z_1,\dots,z_n),(\lambda)))$\;
\tcc{\qquad{\color{blue}Step 2: Construct the subset $\bU
\subset\mathcal{C}$ of minimal critical points}}
\lIf{$F$ is known to be combinatorial}{call Algorithm~\ref{alg:EffectiveCombAsm}
with $H,P,\bQ,u$}\lElse{call Algorithm~\ref{alg:MinimalCritical} with
$H,P,\bQ,u$}
\tcc{\qquad{\color{blue}Step 3: check $\mH$ and $G$ at the elements of
$\mathcal{U}$ and return}}
$\tilde{\mH} \gets$ determinant of the matrix $
\lambda\mH$ with $\mH_{i,j}$ defined by Equation~\eqref{eq:Hess}\;
$\operatorname{Sys}\gets\operatorname{Crit}\cup\left\{h-
\tilde{\mH}(\lambda,\bz),T-z_1\dotsm z_n,g+G(\bz)\right\}$ with new
variables $h,T,g$\;
$[P,(Q_1,\dots,Q_n,Q_\lambda,Q_{\tilde{\mH}},Q_{T},Q_{-G}),u] \gets{\sf KroneckerRep}(\operatorname{Sys},
((z_1,\dots,z_n),(\lambda),(h),(T),(g)))$ using the linear form~$u$
from Step~1; the polynomials $Q_{\tilde{\mH}},Q_{T},Q_{-G}$ 
are described in Prop.~\ref{prop:KronReduce2}\;
\lIf{$Q_{\tilde{\mH}}(u)=0$ at any element of $\bU$ or $Q_{-G}(u)=0$
at all of them}{\Return{\textsc{fail}}}
\Return $(Q_{-G}/Q_{\lambda}, Q_n^2Q_{\lambda}^{n-1}Q_T^{n-3}/Q_{\tilde{\mH}}\cdot (P')^{3-2n}, P'/Q_T, P, \bU)$
\caption{{\sf ACSV}}
\label{alg:ACSV}
\end{algorithm}

\begin{algorithm}
\DontPrintSemicolon
\KwIn{Polynomial $H(\bz)$ and Kronecker representation $P,\bQ,u$ of
the set of critical points as in Eq.~\eqref{eq:KronRep}}
\KwOut
{Set $\bU$ of roots of $P$ giving the \emph{minimal}
critical points assuming combinatoriality, or \sc{fail}}
\tcc{\qquad{\color{blue} Step 1: determine the set $\mathcal{S}$ (using separating linear form $u$)}}
$\bbf \gets \left\{H,  z_1(\partial H/\partial z_1)-\lambda, 
\dots ,z_n(\partial H/\partial z_n)-\lambda, H
(tz_1,\dots,tz_n)\right\}$\;
$[\tilde{P},\tilde{\bQ}] \gets {\sf KroneckerRep}(\bbf,
((z_1,\dots,z_n),(\lambda),(t)))$\;
\tcc{\qquad{\color{blue} Step 2: find $\bzeta$ a minimal critical
point with
real positive coords}}
$[\tilde{P},\tilde{\bQ},\tilde{\mathbf{U}}] \gets {\sf
NumericalKroneckerRep}
(\tilde{P},\tilde{\bQ},\kappa)$
with $\kappa=\tilde{O}(D^2d(d+h))$ given by Corollary~%
\ref{coro:tin01}\;
Use this to group the roots of~$\tilde{P}$ according to the distinct
values
they give to each~$\tilde{Q}_i/\tilde{P}'$ at the real roots of
$\tilde{P}$ and to compute the
corresponding signs\;
$S\gets\{\{v\in\tilde{\mathbf{U}}\cap\mathbb{R}\mid \tilde{Q}_1(v)/
\tilde{P}'
(v)>0\wedge\dots\wedge
\tilde{Q}_n
(v)/\tilde{P}'(v)>0\}\}$\;

\ForEach{coordinate $v\in\{1,\dots,n\}$}{
	refine the partition $S$ according to the distinct values
	taken by~$\tilde{Q}_v/\tilde{P}'$ on its elements
}
\lForEach{set $s\in S$}{
	\lIf{one of the values taken by~$t$ on the elements of $s$ is in $
	(0,1)$}{$S\gets S\setminus\{s\}$}
}
\vskip-1em 
\lIf{the partition $S$ is not of the form $\{\{u_\zeta\}\}$ or if
$Q_\lambda(u_\zeta)=0$}{\Return\textsc{fail}}

\tcc{\qquad{\color{blue}Step 3: identify~$\bzeta$ among the elements
of~$\mathcal{C}$}}
$[{P},{\bQ},{\mathbf{U}}] \gets {\sf
NumericalKroneckerRep}
({P},{\bQ},\kappa)$
with $\kappa=\tilde{O}(D^2(d+h))$ given by Corollary~%
\ref{coro:tin01}\;
Use this to group the roots of~$P$ by the distinct values they give
to each coordinate\;
Identify the value $u_\zeta$ of~$u$ corresponding to $\bzeta$ from its
numerical
coordinates obtained in Step~2\;
\tcc{\qquad{\color{blue}Step 4: construct the set $\mathcal{U}$ of
minimal
critical points and return}}
Refine $\mathbf{U}$ to isolating regions for the complex roots of size
at most $2^{-\kappa}$
with $\kappa=\tilde{O}(D^3(d+h))$ given by Corollary~\ref{cor:moduli}.

\Return{subset of $\bU$ such that
$|z_1|=\zeta_1,\dots,|z_n|=\zeta_n$, where $\zeta_1,\dots,\zeta_n$
are given by $u_\zeta$.}
\caption{{\sf Minimal Critical Points in the Combinatorial Case}}
\label{alg:EffectiveCombAsm}
\end{algorithm}
\begin{algorithm}
\DontPrintSemicolon
\KwIn{Polynomial~$H$ and Kronecker representation $P,\bQ,u$ of the set
of critical
points, as in Eq.~\eqref{eq:KronRep}}
\KwOut
{Set $\bU$ of roots of $P$ corresponding to the \emph{minimal}
critical points, or \sc{fail}}
\tcc{\qquad{\color{blue}Step 1: determine the set $\mathcal S$}}
$\tilde{\bbf} \gets $ Polynomials in Equations~\eqref{eq:GenSys1} -- 
\eqref{eq:GenSys6}\;
$[\tilde{P},\tilde{\bQ}] \gets {\sf KroneckerRep}
(\tilde{\bbf},((\ba,\bb),
(\bx,\by),(\lambda_R),(\lambda_I),(t),(\nu))$ using a linear
form in
$(\ba,\bb,\lambda_R,\lambda_I)$ only\;
\tcc{\qquad{\color{blue} Step 2: construct the set $\mathcal U$ of
minimal
critical points}}
$[\tilde{P},\tilde{\bQ},\tilde{\mathbf{U}}] \gets {\sf
NumericalKroneckerRep}
(\tilde{P},\tilde{\bQ},\kappa)$
with $\kappa=\tilde{O}(2^{2n}D^6d^4h)$ given by
Corollary~\ref{coro:tin01-complex}\;
Use this to group the roots of~$\tilde{P}$ according to the distinct
values
they give to each $\tilde{Q}_{a_i}/\tilde{P}'$ and $\tilde{Q}_
{b_i}/\tilde{P}'$\;
$S\gets\{\tilde{\mathbf{U}}\}$\;

\ForEach{coordinate $v\in\{a_1,\dots,a_n,b_1,\dots,b_n\}$}{
	refine the partition $S$ according to the distinct values
	taken by~$\tilde{Q}_v/\tilde{P}'$ on its elements
}
\ForEach{set $s\in S$}{
	\lIf{$a_i=b_i=0$ on the elements of~$s$ for some~$i$}{\Return
	{\textsc{FAIL}}}
	\lIf{one of the values taken by~$t$ on the elements of $s$ is in $
	(0,1)$}{$S\gets S\setminus\{s\}$}
}
$\mathcal{S}_{a,b}\gets$ values of $
(\tilde{Q}_{a_i}/\tilde{P}',\tilde{Q}_{b_i}/\tilde{P}')$ for
$i=1,\dots,n$ at the elements of~$\cup_
{s\in
S}s$\;
\tcc{\qquad{\color{blue} Step 3: identify the elements
of~$\mathcal{U}$ within~$\mathcal{C}$ and return}}
$[P,\bQ,\mathbf{U}] \gets {\sf NumericalKroneckerRep}
(P,\bQ,\kappa)$
with $\kappa=\tilde{O}(D^2d(d+h))$ given by Corollary~%
\ref{coro:tin01}\;
\Return{roots of~$P$ giving
$z_1=a_1+ib_1,\dots,z_n=a_n+ib_n$ for some $(\ba,\bb)\in\mathcal{S}_
{a,b}$}
\caption{{\sf Minimal Critical Points in the Non-Combinatorial Case}}
\label{alg:MinimalCritical}
\end{algorithm}

Algorithm~\ref{alg:ACSV} implements the general strategy: first,
a Kronecker representation of the solution set of the critical point
equations is
computed; next, the set of minimal critical points is extracted from
that set using Algorithm~\ref{alg:EffectiveCombAsm} in the
combinatorial case and Algorithm~\ref{alg:MinimalCritical} otherwise;
finally, the conditions of regularity of the Hessian
and non-vanishing of the numerator are checked  and the
polynomials~$A$, $B$ and $C$ are computed and returned.

The correctness of this algorithm follows from Propositions~\ref{prop:algo-combi} 
and~\ref{prop:correct-non-combi}. We now turn to the complexity
analysis of each of the steps.

\subsection{Complexity of the Steps in Algorithm~\ref{alg:ACSV}}
\paragraph{Step~1}
By Corollary~\ref{coro:kro-combi},
the Kronecker representation $[P,\bQ]$ can be computed
in $\tilde{O}(D^3(d+h))=\tilde{O}(\sL)$ bit operations.

\paragraph{Step 3}
The entries of the $(n-1)\times(n-1)$ matrix $\lambda{\mH}$ are
polynomials of degrees at
most $2d-2$ and heights at most $h + \log(d^2) + 2$. Thus its
determinant, equal to the determinant of ${\mH}$ multiplied by
$\lambda^{n-1}$, has degree at most
$2(n-1)(d-1)$ and height in $\tilde{O}(n(h+\log d))$.
Proposition~\ref{prop:KronReduce2}, together with Corollary~\ref{coro:kro-combi}, then
implies that the
polynomial $Q_{\tilde{\mH}}$ has degree at most~$nD$ and
height in $\tilde{O}(D(d+h)^2)$. The complexity of evaluation of the
second derivatives is bounded by~$O(nD)$ and a straight-line program
of length only~$O(n^4)$ for the determinant is given by Berkowitz's
algorithm. The complexity for the computation of~$Q_{\tilde{\mH}}$ is
therefore~$\tilde{O}(\sL)$ bit
operations, with $\sL$ as in the theorem.
The other polynomials, $Q_T$ and $Q_{-G}$ are obtained in a similar
way. The polynomial $G(\bz)$ has degree at
most $d$ and height at most $h$, and the polynomial $T = z_1 \cdots
z_n$ has degree $n$ and height 1.  Thus, Proposition~%
\ref{prop:KronReduce2} shows that the bounds for~$Q_{\tilde{\mH}}$ also
hold for them.

Testing that the Hessian is not singular is also achieved with a
complexity bound of~$O(\sL)$ bit operations, by
Proposition~\ref{prop:signq}, and similarly for the vanishing of~$G$.

Finally, 
part~(i) of Proposition~\ref{prop:signq} shows that
\[ A(u) = \frac{Q_{-G}(u)}{Q_{\lambda}(u)}, \qquad B(u) = \frac{Q_n(u)^2Q_{\lambda}(u)^{n-1}Q_T(u)^{n-3}}{Q_{\tilde{\mH}}(u)} \cdot P'(u)^{3-2n}, \qquad C(u) = \frac{P'(u)}{Q_T(u)} \] 
can be computed at all roots of $P$ to $\kappa$ bits of precision
in $\tilde{O}(\sL+D\kappa)$
bit operations.

\subsection{Complexity of Algorithm~\ref{alg:EffectiveCombAsm}}

\paragraph{Step~1}
By Corollary~\ref{coro:kro-combi},
the Kronecker representation $[P,\bQ]$ can be computed
in $\tilde{O}(D^3d^2(d+h))$ bit operations. 

\paragraph{Step~2}
Corollary~\ref{coro:tin01} then shows that all the necessary decisions
in Step~2 can be performed in~$\tilde{O}(D^3d^2(d+h))$ bit operations.

\paragraph{Steps 3 and 4} One must determine and refine the roots
of $P$ to accuracy $\tilde{O}(D^3(d+h))$.
Lemma~\ref{lemma:fsolve} and Corollary~\ref{coro:eqmodcombi} show that the
necessary computations can be done in~$\tilde{O}(D^4(d+h))$ bit operations.

\subsection{Complexity of Algorithm~\ref{alg:MinimalCritical}}

\paragraph{Step 1} The computation of the Kronecker representation has bit
complexity~$\tilde{O}(4^nD^7d^4h)$ by
Corollary~\ref{kro-non-combi}.

\paragraph{Step 2} This step starts by a numerical resolution and
grouping
of roots in~$\tilde{O}(D^6
(h+d))$ bit operations by Corollary~\ref{coro:tin01-complex}.
The refinement of the partition has negligible cost  and the most
expensive operation is the filtering by the values of~$t$ in~$(0,1)$,
which is performed in~$\tilde{O}(2^{3n}D^9d^5h)$ bit operations, by
Corollary~\ref{coro:tin01-complex} again. This complexity dominates
that of all the previous and subsequent steps:
the next operation is a numerical Kronecker representation, whose
complexity is much smaller.

\section{Additional Examples}
\label{sec:Examples}

We now discuss additional examples highlighting the above
techniques\footnote{The calculations for these examples, together with
a preliminary Maple implementation of our algorithms for the
combinatorial case and automated examples using that implementation,
can be found in accompanying Maple worksheets at 
\url{http://diagasympt.gforge.inria.fr}. This preliminary implementation computes the Kronecker
representation through Gr{\"o}bner bases computations, meaning it does not run in the complexity stated above, and does not use certified numerical computations.}.

\begin{example}
\label{ex:Apery3a}
The generating function of Ap{\'e}ry's sequence is the diagonal of the combinatorial rational function $F(w,x,y,z) = 1/H(w,x,y,z)$, where
\[H(w,x,y,z) = 1-w(1+x)(1+y)(1+z)(1+y+z+yz+xyz)\]
defines a smooth algebraic set $\mV(H)$. Taking the system 
\[ H(w,x,y,z), \quad w(\partial H/\partial w) - \lambda, \quad \dots \quad , z(\partial H/\partial z) - \lambda, \quad  H(tw,tx,ty,tz), \]
we try the linear form $u=w+x+y+z+t$ and find that it is separating and a Kronecker representation is given by
\begin{itemize}
	\item a polynomial $P(u)$ of degree 14 and coefficients of absolute value less than $2^{65}$;
	\item polynomials $Q_w,Q_x,Q_y,Q_z,Q_{\lambda},Q_t$ of degrees at most 13 and coefficients of absolute value less than $2^{68}$.  
\end{itemize}

The critical points of $F$ are determined by the roots of 
\[ \tilde{P}(u) = \gcd(P,P'-Q_t) = u^2+160u-800,\] 
as these are the solutions of the polynomial system where $t=1$.  Substituting the roots 
\[ u_1 = -80+60\sqrt{2}, \qquad u_2 =  -80-60\sqrt{2} \]
of $\tilde{P}$ (which can be solved exactly since $\tilde{P}$ is quadratic) into the Kronecker representation determines the two critical points
\begin{align*}
\bp = \left(\frac{Q_w(u_1)}{P'(u_1)},\frac{Q_x(u_1)}{P'(u_1)},\frac{Q_y(u_1)}{P'(u_1)},\frac{Q_z(u_1)}{P'(u_1)}\right) &= \left( -82+58\sqrt{2}, 1+\sqrt{2}, \frac{\sqrt{2}}{2},\frac{\sqrt{2}}{2} \right)\\
\bs = \left(\frac{Q_w(u_2)}{P'(u_2)},\frac{Q_x(u_2)}{P'(u_2)},\frac{Q_y(u_2)}{P'(u_2)},\frac{Q_z(u_2)}{P'(u_2)}\right) &= \left( -82-58\sqrt{2}, 1-\sqrt{2}, \frac{-\sqrt{2}}{2},\frac{-\sqrt{2}}{2} \right)
\end{align*}
of which only $\bp$ has non-negative coordinates and thus could be minimal.  Determining the roots of $P(u)=0$ to sufficient precision shows that there are 6 real values of $t$, and none lie in $(0,1)$.  Thus, $\bp$ is a smooth \emph{minimal} critical point, and there are no other critical points with the same coordinate-wise modulus.  

Once minimality has been determined, the Kronecker representation of this system can be reduced to a Kronecker representation which encodes only critical points.  This is done using $\tilde{P}$ by determining the inverse $P'(u)^{-1}$ of $P'$ modulo $\tilde{P}$ (which exists as $\tilde{P}$ is a factor of $P$, and $P$ and $P'$ are co-prime) and setting 
\[ \tilde{Q}_v(u) := Q_v(u) \tilde{P}'(u)P'(u)^{-1} \text{ mod } \tilde{P}(u) \] 
for each variable $v \in \{w,x,y,z,\lambda\}$.  In this case we obtain a Kronecker representation of the critical point equations given by
\[ \tilde{P}(u) = u^2+160u-800=0 \]
and
\[ w=\frac{-164u+800}{2u+160}, \quad x=\frac{2u+400}{2u+160}, \quad y=z=\frac{120}{2u+160}, \quad \lambda = \frac{-2u-160}{2u+160}. \]
Computing the determinant of the polynomial matrix $\tilde{\mH}$ obtained from multiplying each row of the matrix in Equation~\eqref{eq:Hess} by $\lambda$ shows that the values of this determinant, together with the polynomial $T=wxyz$, can be represented at solutions of the Kronecker representation by
\[ \frac{Q_{\tilde{\mH}}}{\tilde{P}'} = \frac{96u-480}{2u+160}, \qquad \frac{Q_T}{\tilde{P}'} = \frac{34u-160}{2u+160}. \]
Ultimately, noting that $-G=-1$ for this example, we obtain diagonal asymptotics
\[ f_{k,k,k,k} = \left(\frac{u+80}{17u-80}\right)^k \cdot k^{-3/2} \cdot \frac{\sqrt{6u+480}}{48\pi^{3/2} \sqrt{5-u}}\left(1+O\left(\frac{1}{k}\right)\right), \quad u \in \bU \]
where $\bU = \{u_1\} = \{-80+60\sqrt{2}\}$. In general, when $\tilde{P}$ is not quadratic, $\bU$ contains isolating intervals of roots of $\tilde{P}$.  Since we have $u$ exactly here we can determine the leading asymptotic term exactly,
\[ \frac{(17+12\sqrt{2})^k}{k^{3/2}} \cdot \frac{\sqrt{48+34\sqrt{2}}}{8\pi^{3/2}}\left(1+O\left(\frac{1}{k}\right)\right)=\frac{(33.97056\ldots)^k}{k^{3/2}}\left(0.22004\ldots+O\left(\frac{1}{k}\right)\right).\]
\end{example}

\begin{example}
\label{ex:Apery3b}
A second Ap{\'e}ry sequence $(c_k)$, related to the irrationality of
$\zeta(3)$, has for generating function $C(z)$ the diagonal of the combinatorial rational function
\[ \frac{1}{1-x-y-z(1-x)(1-y)} = \frac{1}{1-x-y} \cdot \frac{1}{1-z} \cdot \frac{1}{1-\frac{xyz}{(1-x-y)(1-z)}}. \]  
An argument analogous to the one in Example~\ref{ex:Apery3a}, detailed in the accompanying Maple worksheet, shows that there are two critical points
\[ \bp = \left(\frac{3-\sqrt{5}}{2}, \frac{3-\sqrt{5}}{2}, \frac{-1+\sqrt{5}}{2} \right) \quad \text{and} \quad
\bs = \left(\frac{3+\sqrt{5}}{2}, \frac{3+\sqrt{5}}{2}, \frac{-1-\sqrt{5}}{2} \right), \]
of which $\bp$ is minimal.  Ultimately, one obtains 
\[ c_k = \left(\frac{2(5-u)}{11u-30}\right)^k \cdot k^{-1} \cdot \frac{(10-2u)(2u-5)}{\pi(4u-10)\sqrt{10(5u-14)(u-5)}}\left(1+O\left(\frac{1}{k}\right)\right), \]
where $u = 5-\sqrt{5}$ is a root of the polynomial $P(u) = u^2-7u+12$ which can be determined explicitly as $P$ is quadratic, so 
\[ c_k = \frac{\left(\frac{11}{2}+\frac{5\sqrt{5}}{2}\right)^k}{k} \cdot \frac{\sqrt{250+110\sqrt{5}}}{20\pi}\left(1 + O\left(\frac{1}{k}\right)\right). \]
When combined with the BinomSums Maple package of Lairez\footnote{Available at \url{https://github.com/lairez/binomsums}.}, our preliminary implementation allows one to automatically go from the specification 
\[ c_k := \sum_{j=0}^k \binom{k}{j}^2\binom{k+j}{j} \]
to asymptotics of $c_k$, proving the main result of Hirschhorn~\cite{Hirschhorn2015}.  
\end{example}

\begin{example}
The rational function
\[F(x,y) = \frac{1}{(1-x-y)(20-x-40y)-1},\]
has a smooth denominator. It is combinatorial as can be seen by
writing
\[ F(x,y) = \frac{1}{1-x-y} \cdot \frac{1}{20 - x - 4y - \frac{1}{1-x-y}}. \]
A Kronecker representation of the system 
\[ H(x,y), \quad x(\partial H/\partial x) - \lambda, \quad y(\partial H/\partial y) - \lambda, \quad  H(tx,ty), \]
using the linear form $u=x+y$ (which separates the solutions of the
system) shows that the system has 8 solutions, of which 4 have $t=1$ and correspond to critical points.  There are two critical points with positive coordinates: 
\[ (x_1,y_1)\approx (0.548, 0.309) \qquad \text{and} \qquad (x_2,y_2)\approx (9.997, 0.252).\] 
Since $x_1<x_2$ and $y_1>y_2$, it is not immediately clear which (if any) should be a minimal critical point. However, examining the full set of solutions, not just those where $t=1$, shows there is a point with approximate coordinates $(0.092x_2,0.092y_2)$ in $\mV$, so that $x_1$ is the minimal critical point.  To three decimal places the diagonal asymptotics have the form 
\[f_{k,k} = (5.884\ldots)^k k^{-1/2} \left(0.054\ldots + O\left(\frac{1}{k}\right)\right).  \]
\end{example}

\begin{example}
Straub~\cite{Straub2008}, following work of Gillis et al.~\cite{GillisReznickZeilberger1983}, studied the coefficients of
\[ F_c(x,y,z) = \frac{1}{1-(x+y+z)+cxyz} \]
for real parameter $c$, showing that $F_c$ has non-negative coefficients if and only if $c \leq 4$. Baryshnikov et al.~\cite{BaryshnikovMelczerPemantleStraub2018} studied this family of functions, and related families, through the asymptotic lens of ACSV. We can give asymptotics for any fixed $c$ using our algorithms. 

For instance, when\footnote{This choice simplifies the constants involved, making it easier to display the results here, but the runtime of Algorithm~\ref{alg:MinimalCritical} is largely independent of $c$.} $c=81/8$ there are three critical points, where $x=y=z$ and
\[ x \in \left\{ -2/3, -1/3 \pm i/\left(3\sqrt{3}\right) \right\}. \] 
The ideal $J$ generated by Equations
\eqref{eq:GenSys1}--\eqref{eq:GenSys6} contains a degree~37
polynomial $P(t) = \left(81t^3+36t^2+4t-9\right)(1-3t)Q(t)$, 
where $Q(t)$ contains no real roots in $(0,1)$. 
Adding $\left(81t^3+36t^2+4t-9\right)(1-3t)$ to the ideal $J$
shows that the critical point $(-2/3,-2/3,-2/3)$ is the only one which is 
non-minimal.

The Hessian of $\psi$ takes the values
\[ \left(2/21 \pm i 10\sqrt{3}/63\right) \begin{pmatrix} 2 & 1\\1 & 2\end{pmatrix} \]
at the minimal critical points, and the diagonal has dominant asymptotic term 
\[ \left(i 81 \sqrt{3} / 8\right)^k k^{-1} \frac{3\sqrt{3}+3i}{8\pi} + 
\left(-i 81 \sqrt{3} / 8\right)^k k^{-1} \frac{3\sqrt{3}-3i}{8\pi}
=\left(\frac{81\sqrt{3}}{8}\right)^k\frac{3\cos\!\left
(\frac\pi6+k\frac\pi2\right)}{2k\pi}. \]
\end{example}

\section{Genericity Results}
\label{sec:generic}

In this section we show that our assumptions in the combinatorial
case hold generically. We expect that the non-combinatorial case can
be dealt with in a similar way, but have not done so yet.

Given a collection of polynomials $f_1(\bz),
\dots, f_r(\bz)$ of degrees at most $d_1,\dots,d_r$, respectively, we write
\[ f_j(\bz) = \sum_{|\bi| \leq d_j} c_{j,\bi}\bz^{\bi}\]  
for all $j=1,\dots,r$, where $|\bi| = i_1+\cdots+i_n$ for indices $\bi \in \mathbb{N}^n$.  Given a polynomial $P$ in the set of variables $\{ u_{j,\bi} : |\bi| \leq d_j, 1 \leq j \leq r\}$ we let $P(f_1,\dots,f_r)$ denote the evaluation of $P$ obtained by setting the variable $u_{j,\bi}$ equal to the coefficient $c_{j,\bi}$.

Our proofs of genericity make use of multivariate
resultants and discriminants, for which we refer to 
Cox et al.~\cite{CoxLittleOShea2005} and Jouanolou~\cite{Jouanolou1991}.  For all positive
integers $d_0,\dots,d_n$ the resultant defines an explicit polynomial
$\Res = \Res_{d_0,\dots,d_n}\in \mathbb{Z}[u_{j,\bi}]$ such that $n+1$
homogeneous polynomials $f_0,\dots,f_n \in \mathbb{C}[z_0,\dots,z_n]$
of degrees $d_0,\dots,d_n$ share a non-zero solution in $\mathbb{C}^n$
if and only if $\Res(f_0,\dots,f_n)=0$. The general pattern of the
proofs is to first construct a resultant whose vanishing encodes the
property to be shown generic and then exhibit a system where that
resultant is not~0.

\subsection{Generically, \texorpdfstring{$H$}{H} and its partial
derivatives do not
vanish simultaneously} We prove that this statement, stronger than (A1), holds
generically.
Suppose $H$ has degree $d$ and let 
\[ E(z_0,z_1,\dots,z_n) = z_0^d H(z_1/z_0,\dots,z_n/z_0) \]
be the homogenization of $H$.  As
\[ \partial E/\partial z_j = z_0^{d-1} (\partial H/\partial z_j)(z_1/z_0,\dots,z_n/z_0)  \]
for $j=1,\dots,n$, and Euler's relations for homogeneous polynomials states
\[ \sum_{j=0}^n z_j(\partial E/\partial z_j)(z_0,\dots,z_n)= d \cdot E(z_0,\dots,z_n), \]
it follows that the polynomial $H$ and its partial derivatives vanish
at some point $(p_1,\dots,p_n)$ only if the system
\begin{equation} \partial E/\partial z_0 = \cdots = \partial E/\partial z_n=0 \label{eq:Eres} \end{equation}
admits the non-zero solution $(1,p_1,\dots,p_n)$.  Thus, assumption (A1) holds unless the multivariate resultant $P_d$ of the polynomials in Equation~\eqref{eq:Eres}, which depends only on the degree $d$, is zero when evaluated at the coefficients of $H$.  

It remains to show that $P_d$ is not identically zero for any $d$. If $H_d(\bz) = 1-z_1^d - \cdots -z_n^d$ then Equation~\eqref{eq:Eres} becomes
\[ z_0^{d-1} = -z_1^{d-1} = \cdots = -z_n^{d-1} = 0, \]
which has only the zero solution.  This implies the multivariate resultant $P_d$ is non-zero when evaluated at the coefficients of $H_d$, so it is a non-zero polynomial.

\subsection{Generically, \texorpdfstring{$G(\bz)$}{G(z)} is non-zero
at all critical points}
Again, this statement is stronger than (A2) and thus it is sufficient
to prove its genericity. Homogenizing
\[ H(\bz), \quad G(\bz), \quad z_1(\partial H/\partial z_1) - z_2(\partial H/\partial z_2), \dots, 
z_1(\partial H/\partial z_1) - z_n(\partial H/\partial z_n) \]
gives a system of $n+1$ homogeneous polynomials in $n+1$ variables
\footnote{As in all arguments using the multivariate resultant in this section, $G$ and $H$ are considered as dense polynomials of the specified degrees whose coefficients are indeterminates.}.  The multivariate resultant of this system is a polynomial $P_{d_1,d_2}(G,H)$ in the coefficients of $H$ and $G$, depending only on the degrees $d_1$ and $d_2$ of $G$ and $H$, which must be zero whenever $G(\bz)$ vanishes at a critical point.  

It remains to show that $P_{d_1,d_2}$ is non-zero for all $d_1,d_2 \in \mathbb{N}^*$, which we do by showing it is non-zero for an explicit family of polynomials of all degrees.  If 
\[ G(\bz) = z_1^{d_1} \quad \text{ and } \quad H(\bz) = 1-z_1^{d_2} - \cdots - z_n^{d_2} \]
then the system of homogeneous polynomial equations
\begin{align*}
u^{d_2}H\left(z_1/u,\dots,z_n/u\right) = u^{d_2}-z_1^{d_2} - \cdots - z_n^{d_2} &= 0 \\
G = z_1^{d_1} &= 0\\
-d_2z_1^{d_2} + d_2z_j^{d_2} &= 0, \qquad j=2,\dots,n
\end{align*}
has only the trivial solution $(u,z_1,\dots,z_n) = \bzer$. This implies the multivariate resultant $P_{d_1,d_2}$ is non-zero when evaluated on the coefficients of the polynomials $G$ and $H$ given here, so it is a non-zero polynomial.

\subsection{Generically, all critical points are non-degenerate}
We prove that the matrix $\mH$ in Equation~\eqref{eq:Hess} is
generically non-singular at every critical point (here we let
$\zeta_j$ in Equation~\eqref{eq:Hess} be the variable $z_j$, which
will be eliminated from the critical point equations).  After multiplying every entry of $\mH$ by $\lambda = z_1(\partial H/\partial z_1)$, which is non-zero at any minimal critical point, we obtain a polynomial matrix $\tilde{\mH}$ whose determinant vanishes if and only if an explicit polynomial $D$ in the variables $\bz$ and the coefficients of $H$ vanishes. After homogenizing the system of $n+1$ equations consisting of $D=0$ and the critical point equations~\eqref{eq:critpt} we can compute the multivariate resultant to determine an integer polynomial $P_d$ in the coefficients of $H$, depending only on the degree $d$ of $H$, which must be zero at any degenerate critical point.

It remains to show that the polynomial $P_d$ is non-zero for all $d \in \mathbb{N}^*$. Fix a non-negative integer $d$ and consider the polynomial $H(\bz) = 1 - z_1^d - \cdots - z_n^d$.  Calculating the quantities in Equation~\eqref{eq:Hess}, and substituting $z_j^d=z_1^d$ for each $j=2,\dots,n$, shows that $\tilde{\mH}$ is the polynomial matrix with entries of value $a:=-d^2z_1^d$ on its main diagonal and entries of value $b := -2d^2z_1^d$ off the main diagonal.  Such a matrix has determinant 
\[ D  = a^{n-1}(a+(n-1)b) = (-z_1^dd^2)^n(1-2n) \]
so the only solution to the homogenized smooth critical point
equations and $D$ is the trivial zero solution.  This implies that the
polynomial $P_d$ is non-zero when evaluated at the coefficients of
$H$, and thus it is a non-zero polynomial.
\vspace{-0.1in}

\subsection{Generically, the Jacobian of the smooth critical
point equations is
non-singular at the critical points}

The Jacobian of the system
\[ \bbf := \left( H, \quad z_1(\partial H/\partial z_1)-\lambda, \quad \dots \quad ,z_n(\partial H/\partial z_n)-\lambda, \quad H(tz_1,\dots,tz_n)\right) \]
with respect to the variables $\bz,\lambda,$ and $t$ is a square matrix which is non-singular at its solutions if and only if its determinant $D(\bz,t)$ (which is independent of $\lambda$) is non-zero at its solutions.  Any solution of $\bbf$ has $t \neq 0$, so the existence of a solution to $\bbf=D=0$ gives the existence of a non-zero solution to the system obtained by homogenizing the polynomials $H, z_1(\partial H/\partial z_1) - z_j(\partial H/\partial z_j),t^dH(\bz/t),$ and $t^{d-1}D(\bz,1/t),$ where $d$ is the degree of $H$ (note $D$ has degree $d-1$ in $t$).  The multivariate resultant of this system is an integer polynomial $P_d$, depending only on the degree $d$ of $H$, which must vanish if the Jacobian is singular at at least one of its solutions.
\vspace{0.1in}

It remains to show that the polynomial $P_d$ is non-zero for all $d \in \mathbb{N}^*$. Fix a non-negative integer $d$ and consider the polynomial $H(\bz) = 1 - z_1^d - \cdots - z_n^d$.  The Jacobian of $\bbf$ is the matrix
{\small\[
J := 
\begin{pmatrix}
-dz_1^{d-1} & \cdots & -dz_n^{d-1} & 0 & 0 \\
-d^2z_1^{d-1} & \bzer & 0 & -1 & 0 \\
\bzer & \ddots & \bzer & \vdots & \vdots \\ 
0 & \bzer & -d^2 z_n^{d-1} & -1 & 0 \\ 
-dt^dz_1^{d-1} & \cdots & -dt^d z_n^{d-1} & 0 & -dt^{d-1}(z_1^d + \cdots + z_n^d)
\end{pmatrix},
\]}

\noindent
and a short calculation shows
$D = \det J = -(z_1 \cdots z_n t)^{d-1}(z_1^d + \cdots + z_n^d)(-d)^{n+1} \cdot \det M,$
where $M$ is the $(n+1)\times(n+1)$ matrix 
{\small \[
M := 
\begin{pmatrix}
1 & \cdots & 1 & 0 \\
d & \bzer & 0 & 1\\
0 & \bzer & 0 & \vdots\\
0 & \cdots & d & 1
\end{pmatrix}.
\]}

\noindent
The matrix $M$ is invertible, so $\det M$ is a non-zero constant. The system of homogeneous equations under consideration thus simplifies to
\[ u^d-z_1^d-\cdots-z_n^d = -(z_1^d-z_j^d) = t^d - z_1^d-\cdots-z_n^d = (z_1 \cdots z_n)^{d-1}(z_1^d + \cdots + z_n^d) = 0, \]
which has only the trivial zero solution.  This implies that the
polynomial $P_d$ is non-zero when evaluated at the coefficients of
$H$, and is thus a non-zero polynomial.

\subsection{Generically in the combinatorial case, there is a minimal critical point}

Example~\ref{ex:12} shows the case of a combinatorial rational
function with a
denominator of degree~2 in~$\mathbb{Q}[z,w]$, without any minimal
critical point. In this example, the coefficient of the monomial~$w^2$
in the denominator is~0. This is a reflection of more general
phenomenon: the absence of minimal critical points can only occur if
some of the coefficients of the denominators are~0, making their
presence a generic property. We now show the following.

\begin{proposition}If the reduced combinatorial rational function
$F
(\bz)=G
(\bz)/H(\bz)$ with $H(\mathbf{0})\neq0$ has a denominator of degree~$d$ and no
minimal critical point, then at least one of the monomials~$\bz^\bi$
of degree~$d$ has a 0-coefficient in~$H$.
\end{proposition}
\begin{proof}Up to normalization by a constant, one can write $H(\bz) = 1 - P(\bz)$
with $P=\sum_{|\bi|\le d} c_\bi \bz^\bi$ of degree~$d$ such that $P(0)
= 0$.
By Lemma~\ref{lemma:min-crit}, the continuous map $
\operatorname{Abs}:\bz\mapsto|z_1\dotsm z_n|$ does not reach its
maximum on the boundary $\partial\mD$ of the domain of convergence.
Combinatoriality implies that it does not reach its maximum
on~$\partial\mD\cap\mathbb{R}_{\geq0}^n$.
For any set $S \subset \mathbb{R}$ let $P^{(-1)}(S)$ denote the points mapping to 
$S$ under $P$. Then for any~$\epsilon\in
(0,1)$, the closed set $P^{(-1)}(
[1-\epsilon,1])\cap\overline{\mD}\cap\mathbb{R}_{\geq0}^n$ is unbounded. By
continuity of $P$ this implies 
$P^{(-1)}((1-\epsilon,1))\cap\overline{\mD}\cap\mathbb{R}_{\geq0}^n\subset\mD$
 is unbounded, so there
exists a sequence $(\ba^{(k)})$ in $\mD \cap \mathbb{R}_{\geq0}^n$ with $P
(\ba^{(k)})\in(1-\epsilon,1)$ and $\|\ba^{(k)}\|\rightarrow\infty$. Up
to
extracting a subsequence and renumbering the coordinates, we can
assume that the last coordinate $(a^{(k)}_n)$ tends to infinity. By
combinatoriality, for any $t\in[0,1]$, the points $(ta^{
(k)}_1,\dots,t a^{(k)}_{n-1},a^{(k)}_n)$ where the last coordinate is
fixed to~$a^{(k)}_n$ all belong to~$\mD$. Taking~$t=0$ gives a
sequence $(a^{(k)}_n)$ tending to infinity where $c_{0,\dots,0,d}(a^{
(k)}_n)^d$ is bounded (by~1), which implies that $c_{0,\dots,0,d}=0$.
\end{proof}

\section{Conclusion}
The computation of asymptotics of the diagonal sequences of rational
functions reduces to semi-numerical questions concerning roots of
polynomial systems. If $d$ is the degree of the rational function and
$n$ its number of variables, we have shown that this computation has a
bit complexity that is polynomial in~$d^n$, even in the
non-combinatorial case. Thus this approach is in the same complexity
class as the creative telescoping method, which is more general, but
does not provide an explicit form for the leading constant in the
asymptotic behavior. 

In this article, we have dealt with the simplest geometry of minimal
critical points, which covers many practical examples and allows for an
elementary presentation of the mathematical background. Work is in
Of particular interest for computations, Baryshnikov et al.~\cite{BaryshnikovMelczerPemantle2018a}
give explicit algebraic conditions, checkable using Gr{\"o}bner basis algorithms,
under which the Morse theory framework of Pemantle and Wilson~\cite{PemantleWilson2013} can be
applied rigorously.  With much milder conditions than we consider here, 
one can then write diagonal asymptotics as an \emph{integer} linear combination 
of integrals around (possibly non-minimal) critical points 
which can be asymptotically approximated using saddle-point methods.
Although it is only known how to use ACSV to find these integer coefficients at critical
points which are minimal, or when the singular set is a union of hyperplanes, 
one can apply the creative telescoping methods mentioned above.
This sets up a promising framework for diagonal coefficients, 
however additional work on topics like effective stratifications of 
algebraic varieties is needed before practical implementations
can be created.

\section*{Appendix: Polynomial Bounds and Complexity Estimates}
\label{sec:BoundsEstimates}

In this appendix, we gather a symbolic-numeric toolkit describing
the root separation bounds and algorithms used in our semi-numerical
algorithms. Unless otherwise stated, all logarithms are taken in base~2.

\subsection*{Bounds on Polynomials}
Recall that the \emph{height} $h(P)$ of a polynomial $P\in\mathbb{Z}
[\bz]$
is the maximum of 0 and the base~2 logarithms of the
absolute values of the coefficients of $P$.
\begin{lemma}
\label{lemma:height} 
For univariate polynomials $P_1,\dots,P_{k},P,Q \in \mathbb{Z}[z]$,
\begin{align*}
  h(P_1 + \cdots + P_k) &\leqslant \max_i  h(P_i) + \log k,\\
  h(P_1 \dotsm P_k) & \leqslant \sum_{i=1}^k {h(P_i)} + \sum_{i=1}^{k-1} \log(\deg P_i + 1),\\
  h(P)&\leqslant \deg P + h(PQ) + \log\sqrt{\deg(PQ)+1}.
\end{align*}
\end{lemma}

The first result follows directly from the definition of polynomial
height. The final one---sometimes referred to as `Mignotte's bound on
factors'---follows from Theorem 4 in Chapter 4.4 of Mignotte~\cite{Mignotte1992}.
The second one can be proved by induction on~$k$.

\subsection*{Root Separation Bounds}
\begin{lemma}[Mignotte~\cite{Mignotte1992}]
\label{lemma:roots} 
Let $A \in \mathbb{Z}[z]$ be a polynomial of degree $d\geqslant 2$ and height $h$. If $A(\alpha)=0$ then
\begin{enumerate}[(i)]
  \item if $\alpha\neq0$, then $1/(2^h+1)\le |\alpha|\le 2^h+1$; \label{item:roots1}
  \item if $A(\beta)=0$ and $\alpha\neq\beta$, then $|\alpha-\beta|\ge d^{-(d+2)/2} \cdot \|A\|_2^{1-d}$; \label{item:roots2}
\item if $Q(\alpha)\neq0$ for $Q\in\mathbb{Z}[T]$, then $|Q
(\alpha)|\ge ((\deg Q+1)2^{h(Q)})^{1-d} \cdot (2^h\sqrt{d+1})^{-\deg
Q}$; \label{item:roots3}
  \item if $A$ is square-free then $|A'(\alpha)| \ge 2^{-2dh \, + \, 2h \, + \, 2(1-d)\log d \, + \, (1-d) \log \sqrt{d+1}}$, \label{item:roots4}
\end{enumerate}
where $\|A\|_2$ is the 2-norm of the vector of coefficients, bounded by $2^h\sqrt{d+1}$.
\end{lemma}

The upper bound of statement~\eqref{item:roots1} comes from Theorem
4.2(ii) in Chapter 4 of Mignotte~\cite{Mignotte1992}, and the lower bound is a
consequence of applying the upper bound to the reciprocal polynomial
$z^dA(1/z)$.  Statement~\eqref{item:roots2} comes from Theorem 4.6 in
Section 4.6 of that text.  A proof of~\eqref{item:roots3} can be found
in~\cite[Theorem~A.1]{Bugeaud2004}, while Item~\eqref{item:roots4} is a
special case of~\eqref{item:roots3}. 

\subsection*{Resultant and GCD Bounds}
A height bound on the greatest common divisor of two univariate polynomials is given by Lemma~\ref{lemma:height}, and the complexity of computing gcds is well known~\cite[Corollary 11.11]{GathenGerhard2003}. 

\begin{lemma}
\label{lemma:gcd}
For $P$ and $Q$ in ${\mathbb Z}[U]$ of heights at most $h$ and degrees at most $d$, $\gcd(P,Q)$ has height $\tilde{O}(d+h)$ and can be computed in $\tilde{O}(d^2+hd)$ bit operations. 
\end{lemma}

Similarly, a degree bound for the resultant of two polynomials follows from a direct expansion of the determinant of the Sylvester matrix, and Lemma~\ref{lemma:height} combined with this expansion gives a bound on the resultant height. 

\begin{lemma}
\label{lemma:resultant}
For $P$ and $Q$ in $\mathbb{Z}[T,U]$ let $R = \Res_T(P,Q)$ and
\begin{align*}
\delta &:= \deg_TP\deg_UQ+\deg_TQ\deg_UP \\
\eta &:= h(P)\deg_TQ+ h(Q)\deg_TP+\log((\deg_TP+\deg_TQ)!) + \log(\deg_UP+1) \deg_TQ\\
&\hspace{3.9in} + \log(\deg_UQ+1) \deg_T P.
\end{align*}
Then $\deg R \leq \delta$ and $h(R) \leq \eta$. Furthermore, if all coefficients of $P$ and $Q$ as polynomials in $T$ are monomials in $U$ then $h(R) \leq h(P)\deg_TQ+ h(Q)\deg_TP+\log((\deg_TP+\deg_TQ)!)$.
\end{lemma}
\subsection*{Algorithms for polynomial roots}
\begin{lemma}[Sagraloff and Mehlhorn~\cite{SagraloffMehlhorn2016} and Mehlhorn et al.~\cite{MehlhornSagraloffWang2015}]
\label{lemma:fsolve}
Let $A \in \mathbb{Z}[T]$ be a square-free polynomial of degree $d$ and height $h$.  Then for any positive integer $\kappa$ 
\begin{itemize}
\item[] isolating disks of radius less than $2^{-\kappa}$ can be
computed for all roots of $A(T)$ in $\tilde{O}(d^3+d^2h+d\kappa)$ bit operations;
\item[] isolating intervals of length less than $2^{-\kappa}$ can be
computed for all real roots of $A(T)$ in $\tilde{O}(d^3+d^2h+d\kappa)$ bit operations.
\end{itemize}
\end{lemma}
\begin{proof}
The statement for real roots is Theorem 3 of 
Sagraloff and Mehlhorn~\cite{SagraloffMehlhorn2016}; an implementation is discussed in 
Kobel et al.~\cite{KobelRouillierSagraloff2016}.  The part concerning
intervals follows from Theorem 5 of Mehlhorn et al.~\cite{MehlhornSagraloffWang2015}.
\end{proof}
\begin{lemma}[Kobel and Sagraloff~\cite{KobelSagraloff2015}]
\label{lemma:feval} 
Let $P \in \mathbb{Z}[T]$ be a square-free polynomial of degree $d$
and height $h$, and $t_1,\dots,t_m \in \mathbb{C}$ be a sequence of
length $m=O(d)$.  Then for any positive integer $\kappa$,
approximations $a_1,\dots,a_m \in \mathbb{C}$ such that $|P
(t_j)-a_j|<2^{-\kappa}$ for all $1 \leq j \leq m$ can be computed in $
\tilde{O}(d (h+\kappa+d\log\max_j|t_j| ))$ bit operations, given
$t_1,\dots,t_m$ with $
\tilde{O}(h+\kappa+d\log\max_j |t_j|)$ bits after
the binary point. If
all $t_j$ are real, the approximations $a_j$ are also real.
\end{lemma}
\begin{proof}
This follows from Theorem 10 of Kobel and Sagraloff~\cite{KobelSagraloff2015}; the
statement about real roots follows from the proof given in Appendix B of that paper. 
\end{proof}

\section*{Acknowledgments} The authors would like to thank
Mohab~Safey~El~Din and
\'Eric~Schost for several discussions during the preparation of this
work, and the anonymous referees for their close readings and helpful suggestions. SM was
partially funded by an NSERC Canadian Graduate Scholarship, an 
Eiffel Fellowship, and the Canada-France Research Fund.
BS has been supported in part by FastRelax ANR-14-CE25-0018-01.

\bibliographystyle{plain}
\bibliography{bibl}

\end{document}